\documentclass[twoside]{article}

\usepackage[accepted]{aistats2020} 
\usepackage[round]{natbib}

\usepackage{bm}
\usepackage{placeins}
\usepackage{epsf}
\usepackage{fancyhdr}
\usepackage{graphics}
\usepackage{graphicx}
\usepackage{psfrag}
\usepackage{microtype}
\usepackage{subfigure}
\usepackage{algorithmic}
\usepackage[linesnumbered,ruled]{algorithm2e}
\DontPrintSemicolon
\usepackage{color}
\usepackage{amsthm}
\usepackage{amsfonts}
\usepackage{amsmath}
\usepackage{amssymb,bbm}
\usepackage{url}
\usepackage{hyperref}

\begin{document}

%

%

\newcommand{\Br}{\mathbb{R}}
\newcommand{\BB}{\mathbb{B}}
\newcommand{\BE}{\mathbb{E}}
\newcommand{\BI}{\mathbb{I}}
\newcommand{\T}{\top}
\newcommand{\st}{\textnormal{s.t.}}
\newcommand{\tr}{\textnormal{Tr}\,}
\newcommand{\Tr}{\textnormal{Tr}}
\newcommand{\conv}{\textnormal{conv}}
\newcommand{\diag}{\textnormal{diag}}
\newcommand{\Diag}{\textnormal{Diag}\,}
\newcommand{\Prob}{\textnormal{Prob}}
\newcommand{\var}{\textnormal{var}}
\newcommand{\rank}{\textnormal{rank}}
\newcommand{\sign}{\textnormal{sign}}
\newcommand{\cone}{\textnormal{cone}\,}
\newcommand{\cl}{\textnormal{cl}\,}
\newcommand{\vct}{\textnormal{vec}\,}
\newcommand{\sym}{\textnormal{sym}\,}
\newcommand{\matr}{\boldsymbol M}\,
\newcommand{\vect}{\boldsymbol V}\,
\newcommand{\tens}{\mathscr{T}\,}
\newcommand{\feas}{\textnormal{feas}\,}
\newcommand{\opt}{\textnormal{opt}\,}
\newcommand{\SM}{\mathbf\Sigma}
\newcommand{\SI}{\mathbf S}
\newcommand{\RR}{\mathbf R}
\newcommand{\x}{\mathbf x}
\newcommand{\y}{\mathbf y}
\newcommand{\s}{\mathbf s}
\newcommand{\sa}{\mathbf a}
\newcommand{\g}{\mathbf g}
\newcommand{\e}{\mathbf e}
\newcommand{\z}{\mathbf z}
\newcommand{\w}{\mathbf w}
\newcommand{\su}{\mathbf u}
\newcommand{\sv}{\mathbf v}
\newcommand{\argmin}{\mathop{\rm argmin}}
\newcommand{\argmax}{\mathop{\rm argmax}}
\newcommand{\LCal}{\mathcal{L}}
\newcommand{\CCal}{\mathcal{C}}
\newcommand{\DCal}{\mathcal{D}}
\newcommand{\ECal}{\mathcal{E}}
\newcommand{\GCal}{\mathcal{G}}
\newcommand{\KCal}{\mathcal{K}}
\newcommand{\OCal}{\mathcal{O}}
\newcommand{\QCal}{\mathcal{Q}}
\newcommand{\SCal}{\mathcal{S}}
\newcommand{\Xcal}{\mathcal{X}}
\newcommand{\Ycal}{\mathcal{Y}}
\newcommand{\ICal}{\mathcal{I}}
\newcommand{\JCal}{\mathcal{J}}
\newcommand{\half}{\frac{1}{2}}
\newcommand{\K}{\mbox{$\mathbb K$}}
\newcommand{\Z}{\mbox{$\mathbb Z$}}
\newcommand{\card}{\textnormal{card}}
\newcommand{\trac}{\textnormal{trace}}
\newcommand{\prox}{\textnormal{prox}}
\newcommand{\diam}{\textnormal{diam}}
\newcommand{\dom}{\textnormal{dom}}
\newcommand{\proj}{\textnormal{proj}}
\newcommand{\EP}{{\textbf{E}}}
\newcommand{\CL}{\mathcal L}
\newcommand{\br}{\mathbb{R}}
\newcommand{\bs}{\mathbb{S}}
\newcommand{\bn}{\mathbb{N}}
\newcommand{\ba}{\begin{array}}
\newcommand{\ea}{\end{array}}
\newcommand{\ACal}{\mathcal{A}}
\newcommand{\BCal}{\mathcal{B}}
\newcommand{\ZCal}{\mathcal{Z}}
\newcommand{\FCal}{\mathcal{F}}
\newcommand{\RCal}{\mathcal{R}}
\newcommand{\XCal}{\mathcal{X}}
\newcommand{\bp}{\mathbb{P}}
\newcommand{\bq}{\mathbb{Q}}
\newcommand{\NCal}{\mathcal{N}}
\newcommand{\indc}{\mathbb{I}}
\newcommand{\UCal}{\mathcal{U}}
\newcommand{\red}{\color{red}}
\newcommand{\Rspace}{\mathbb{R}}
\newcommand{\one}{\textbf{1}}
\newcommand{\bigO}{\mathcal{O}}
\newcommand{\bigOtil}{\widetilde{\mathcal{O}}}
\newcommand{\mydefn}{:=}
\newcommand{\defeq}{:=}
\newcommand{\icol}[1]{
  \left(\begin{smallmatrix}#1\end{smallmatrix}\right)%
}

\newtheorem{theorem}{Theorem}[section]
\newtheorem{corollary}[theorem]{Corollary}
\newtheorem{lemma}[theorem]{Lemma}
\newtheorem{proposition}[theorem]{Proposition}
\newtheorem{condition}{Condition}
\newtheorem{definition}{Definition}
\newtheorem{example}{Example}
\newtheorem{question}{Question}
\newtheorem{remark}[theorem]{Remark}
\newtheorem{assumption}[theorem]{Assumption}

\twocolumn[

\aistatstitle{Fast Algorithms for Computational Optimal Transport and Wasserstein Barycenter}



\aistatsauthor{Wenshuo Guo \And Nhat Ho \And Michael I. Jordan }

\aistatsaddress{ UC Berkeley \And  UC Berkeley \And UC Berkeley } ]

\begin{abstract}
We provide theoretical complexity analysis for new algorithms to compute the optimal transport (OT) distance between two discrete probability distributions, and demonstrate their favorable practical performance compared to state-of-art primal-dual algorithms. First, we introduce the \emph{accelerated primal-dual randomized coordinate descent} (APDRCD) algorithm for computing the OT distance. We show that its complexity is $\bigOtil(\frac{n^{5/2}}{\varepsilon})$, where $n$ stands for the number of atoms of these probability measures and $\varepsilon > 0$ is the desired accuracy. This complexity bound matches the best known complexities of primal-dual algorithms for the OT problems, including the adaptive primal-dual accelerated gradient descent (APDAGD) and the adaptive primal-dual accelerated mirror descent (APDAMD) algorithms. Then, we demonstrate the improved practical efficiency of the APDRCD algorithm through comparative experimental studies.  We also propose a greedy version of APDRCD, which we refer to as \emph{accelerated primal-dual greedy coordinate descent} (APDGCD), to further enhance practical performance. Finally, we generalize the APDRCD and APDGCD algorithms to distributed algorithms for computing the Wasserstein barycenter for multiple probability distributions. 
\end{abstract}

\section{Introduction}
Optimal transport has become an important topic in statistical machine learning.  It finds the minimal cost couplings between pairs of probability measures and provides a geometrically faithful way to compare two probability distributions, with diverse applications in areas including Bayesian nonparametrics \citep{Nguyen-2013-Convergence, Nguyen-2016-Borrowing}, scalable Bayesian inference \citep{Srivastava-2015-WASP, Srivastava-2018-Scalable}, topic modeling~\citep{Lin-2018-Sparsemax}, isotonic regression~\citep{Rigollet-2019-Uncoupled}, and deep learning~\citep{Courty-2017-Optimal, Arjovsky-2017-Wasserstein, Tolstikhin-2018-Wasserstein}. 

Nevertheless, the practical impact of OT has been limited by its computational burden. By viewing the optimal transport distance as a linear programming problem, interior-point methods have been employed as a computational solver, with a best known practical complexity of $\bigOtil(n^{3})$ ~\citep{Pele-2009-Fast}. Recently, Lee and Sidford~\citep{Lee-2014-Path} proposed to use the Laplace linear system solver to theoretically improve the complexity of interior-point methods to $\bigOtil(n^{5/2})$. However, it remains a practical challenge to develop efficient interior-point implementations in the high-dimensional settings for OT applications in machine learning. 

Several algorithms have been proposed to circumvent the scalability issue of the interior-point methods, including the Sinkhorn algorithm~\citep{Sinkhorn-1974-Diagonal, Knight-2008-Sinkhorn, Kalantari-2008-Complexity,Cuturi-2013-Sinkhorn}, which has a complexity bound of $\bigOtil(\frac{n^{2}}{\varepsilon^2})$ where $\varepsilon > 0$ is the desired accuracy~\citep{Dvurechensky-2018-Computational}. The Greenkhorn algorithm~\citep{Altschuler-2017-Near} further improves the performance of the Sinkhorn algorithm, with a theoretical complexity of $\bigOtil(\frac{n^{2}}{\varepsilon^2})$~\citep{lin2019efficient}. However, for large-scale applications of the OT problem, particularly in randomized and asynchronous scenarios such as computational Wasserstein barycenters, existing literature has shown that neither the Sinkhorn algorithm nor the Greenkhorn algorithm are sufficiently scalable and flexible~\citep{Cuturi-2014-Fast, dvurechenskii2018decentralize}.

Recent research has demonstrated the advantages of the family of accelerated primal-dual algorithms over the Sinkhorn algorithms. This family includes the adaptive primal-dual accelerated gradient descent (APDAGD) algorithm~\citep{Dvurechensky-2018-Computational} and the adaptive primal-dual accelerated mirror descent (APDAMD) algorithms~\citep{lin2019efficient}, with complexity bounds of $\bigOtil(\frac{n^{2.5}}{\varepsilon})$ and $\bigOtil(\frac{n^{2}\sqrt{\gamma}}{\varepsilon})$, respectively, where $\gamma \leq \sqrt{n}$ is the inverse of the strong complexity constant of Bregman divergence with respect to the $l_{\infty}$-norm. In addition, the primal-dual algorithms possess the requisite flexibility and scalability compared to the Sinkhorn algorithm, which is crucial for computational OT problems in large-scale applications~\citep{Cuturi-2014-Fast,Ho-2018-Probabilistic}. Specifically, they are flexible enough to be generalized to the computation of the Wasserstein barycenter for multiple probability distributions in decentralized and asynchronous settings~\citep{Cuturi-2014-Fast, dvurechenskii2018decentralize}. 

In the optimization literature, primal-dual methods have served as standard techniques that are readily parallelized for high-dimensional applications~\citep{combettes2012primal, wainwright2005map}. On the other hand, coordinate descent methods have been shown to be well suited to the solution of large-scale machine learning problems~\citep{nesterov2012efficiency, richtarik2016parallel, fercoq2015accelerated}. 

\textbf{Our contributions.} The contributions of the paper are three-fold.
\begin{enumerate}
    \item We introduce an \emph{accelerated primal-dual randomized coordinate descent} (APDRCD) algorithm for solving the OT problem. We provide a complexity upper bound of $\bigOtil(\frac{n^{5/2}}{\varepsilon})$ for the APDRCD algorithm, which is comparable to the complexity of state-of-art primal-dual algorithms for OT problems, such as the APDAGD and APDAMD algorithms~\citep{Dvurechensky-2018-Computational, lin2019efficient}. To the best of our knowledge, this is the first accelerated primal-dual coordinate descent algorithm for computing the OT problem.
    
    \item  We show that the APDRCD algorithm outperforms the APDAGD and APDAMD algorithms with experiments on both synthetic and real image datasets. To further improve the practical performance of the APDRCD algorithm, we present a greedy version of it which we refer to as the \emph{accelerated primal-dual greedy coordinate descent} (APDGCD) algorithm. To the best of our knowledge, the APDRCD and APDGCD algorithms achieve the best performance among state-of-art accelerated primal-dual algorithms on solving the OT problem.
    
    \item  We demonstrate that the APDRCD and APDGCD algorithms are suitable and parallelizable for other large-scale problems besides the OT problem, e.g., approximating the Wasserstein barycenter for a finite set of probability measures stored over a distributed network.
\end{enumerate}

\textbf{Organization.} The remainder of the paper is organized as follows. In Section~\ref{Sec:setup}, we provide the formulation of the entropic OT problem and its dual form. In Section~\ref{Sec:APDRCD_algorithm}, we introduce the APDRCD algorithm for solving the regularized OT problem and provide a complexity upper bound for it. We then present the greedy APDGCD algorithm. In Section~\ref{Sec:experiments}, we present comparative experiments between the APDRCD, the APDGCD algorithms and other primal-dual algorithms including the APDAGD and APDAMD algorithms. We conclude the paper with a few future directions in Section~\ref{Sec:discussion}. Finally, the proofs of all the results are included in the Appendix~\ref{sec:supplementary_material}, and the generalized APDRCD and APDGCD algorithms for approximating Wasserstein barycenters are presented in Appendix~\ref{Sec:WB}. Additional experiments are presented in Appendix~\ref{subsec:further_exp}.

\textbf{Notation.} We denote the probability simplex $\Delta^n : = \{u = \left(u_1, \ldots, u_n\right) \in \Rspace^n: \sum_{i = 1}^{n} u_{i} = 1, \ u \geq 0\}$ for $n \geq 2$. Furthermore, $[n]$ stands for the set $\{1, 2, \ldots, n\}$ while $\Rspace^n_+$ stands for the set of all vectors in $\Rspace^n$ with nonnegative components for any $n \geq 1$. For a vector $x \in \Rspace^n$ and $1 \leq p \leq \infty$, we denote $\|x\|_p$ as its $\ell_p$-norm and $\text{diag}(x)$ as the diagonal matrix with $x$ on the diagonal. For a matrix $A \in \Rspace^{n \times n}$, the notation $\text{vec}(A)$ stands for the vector in $\Rspace^{n^2}$ obtained from concatenating the rows and columns of $A$. $\one$ stands for a vector with all of its 
components equal to $1$. $\partial_x f$ refers to a partial gradient of $f$ with respect to $x$. Lastly, given the dimension $n$ and accuracy $\varepsilon$, the notation $a = \bigO\left(b(n,\varepsilon)\right)$ stands for the upper bound $a \leq C \cdot b(n, \varepsilon)$ where $C$ is independent of $n$ and $\varepsilon$. Similarly, the notation $a = \bigOtil(b(n, \varepsilon))$ indicates the previous inequality may depend on the logarithmic function of $n$ and $\varepsilon$, and where $C>0$.

\section{Problem Setup}
\label{Sec:setup}
In this section, we provide necessary background for the entropic regularized OT problem. The objective function for the problem is presented in Section~\ref{Sec:reg_OT} while its dual form as well as the key properties of that dual form are given in Section~\ref{Sec:dual_reg_OT}.
\subsection{Entropic Regularized OT}
\label{Sec:reg_OT}
As shown in~\citep{Kantorovich-1942-Translocation}, the problem of approximating the OT distance between two discrete probability distributions with at most $n$ components is equivalent to the following linear programming problem:
\begin{align}\label{prob:OT}
&\min\limits_{X \in \br^{n \times n}} \left\langle C, X\right\rangle  \\ 
&\ \ \ \st \ \ X\one = r, \ X^\top\one = l, \ X \geq 0 \nonumber 
\end{align}
where $X$ is a \textit{transportation plan}, 
$C = (C_{ij}) \in \br_+^{n \times n}$ is a cost matrix with non negative elements, and $r$ and $l$ refer to two known probability distributions in the probability simplex $\Delta^n$. The best known practical complexity bound for~\eqref{prob:OT} is $\bigOtil(n^3)$~\citep{Pele-2009-Fast}, while the best theoretical complexity bound is $\bigOtil(n^{2.5})$~\citep{Lee-2014-Path}, achieved via interior-point methods. However, these methods are not efficient in the high-dimensional setting of OT applications in machine learning. This motivates the \emph{entropic regularized OT} problem~\citep{Cuturi-2013-Sinkhorn}:
\begin{align}\label{prob:regOT}
&\min\limits_{X \in \br_{+}^{n \times n}}  \left\langle C, X\right\rangle - \eta H(X) \\
& \ \ \ \st \ \ X\one = r, \ X^\top\one = l, \nonumber
\end{align}
where $\eta > 0$ is the \textit{regularization parameter} and 
$H(X)$ is the entropic regularization given by $H(X) \ : = \ - \sum_{i, j = 1}^n X_{ij} \log(X_{ij})$. The main focus of the paper is to determine an  \emph{$\varepsilon$-approximate transportation plan} $\hat{X} \in \br_{+}^{n \times n}$ such that $\hat{X}\one = r$ and $\hat{X}^\top\one = l$ and the following bound holds:
\begin{equation}\label{Criteria:Approximation}
\langle C, \hat{X}\rangle \ \leq \ \langle C, X^*\rangle + \varepsilon, 
\end{equation}
where $X^*$ is an optimal solution; i.e., an optimal transportation plan for the OT problem~\eqref{prob:OT}. To ease the ensuing presentation, we let $\langle C, \hat{X}\rangle$ denote an \emph{$\varepsilon$-approximation} for the OT distance. Furthermore, we define matrix $A$ such that $A\text{vec}(X) : = \begin{pmatrix} X\one \\ X^\top\one \end{pmatrix}$ for any $X \in \br^{n \times n}.$

\subsection{Dual Entropic Regularized OT}
\label{Sec:dual_reg_OT}
The Lagrangian function for problem~\eqref{prob:regOT} is given by
\begin{align*}
\LCal(X, \alpha, \beta) : = &\left\langle C, X\right\rangle - \eta H(X)+ \langle\alpha, r\rangle\\
& + \langle\beta, l\rangle - \langle\alpha, X\one\rangle - \langle\beta, X^\top\one\rangle.
\end{align*}
Given the Lagrangian function, the dual form of the entropic regularized OT problem can be obtained by solving the optimization problem $\min_{X \in \Rspace^{n \times n}} \LCal(X, \alpha, \beta)$. Since the Lagrangian function $\LCal(\cdot, \alpha, \beta)$ is strictly convex, that optimization problem can be solved by setting $\partial_X \LCal(X, \alpha, \beta) = 0$, which leads to the following form of the transportation plan: $X_{ij} \ = \ e^{\frac{-C_{ij} + \alpha_i + \beta_j}{\eta} - 1}$ for all $i, j \in [n]$. With this solution, we have $\min_{X \in \br^{n \times n}} \LCal(X, \alpha, \beta) \ = \ -\eta\sum_{i,j=1}^n e^{- \frac{C_{ij} - \alpha_i - \beta_j}{\eta}-1} + \left\langle \alpha, r\right\rangle + \left\langle \beta, l\right\rangle$. 
The \emph{dual entropic regularized OT} problem is, therefore, equivalent to the following optimization problem:
\begin{align}
\label{eq:dual_entropic}
    \lefteqn{\min_{\alpha, \beta \in \br^{n}} \varphi(\alpha, \beta)}  \nonumber \\ 
     &:= \eta \sum_{i,j=1}^n e^{- \frac{C_{ij} - \alpha_i - \beta_j}{\eta}-1} - \left\langle \alpha, r\right\rangle - \left\langle \beta, l\right\rangle
\end{align}
Building on Lemma 4.1 in~\citep{lin2019efficient}, the dual objective function $\varphi(\alpha, \beta)$ can be shown to be smooth with respect to $\|\cdot\|_{2}$ norm:
\begin{lemma}
\label{lemma:dual_smooth}
 The dual objective function $\varphi$ is smooth with respect to $\|.\|_{2}$ norm:
 \begin{equation*}
\varphi(\lambda_1)-\varphi(\lambda_2)-\left \langle \nabla \varphi(\lambda_2), \lambda_1 - \lambda_2 \right \rangle \leq \frac{2}{\eta}||\lambda_1-\lambda_2||^2_2.
 \end{equation*}
\end{lemma}
\begin{proof}
The proof is straightforward application of the result from Lemma 4.1 in~\citep{lin2019efficient}. Here, we provide the details of this proof for the completeness. Indeed, invoking Lemma 4.1 in~\citep{lin2019efficient}, we find that 
\begin{equation*}
\varphi(\lambda_1)-\varphi(\lambda_2)-\left \langle \nabla \varphi(\lambda_2), \lambda_1 - \lambda_2 \right \rangle \leq \frac{||A||_1^2}{2\eta}||\lambda_1-\lambda_2||^2_\infty.
 \end{equation*}
Since $||A||_1$ is equal to the maximum $\ell_1$-norm of a column of A and each column of A contains only two nonzero elements which are equal to one, we have $||A||_1 = 2$. Combining with the fact that $||\lambda_1-\lambda_2||^2_\infty \leq ||\lambda_1-\lambda_2||^2_2 $, we establish the result.
\end{proof}

\section{Accelerated Primal-Dual Coordinate Descent Algorithms}
\label{Sec:APDRCD_algorithm}
In this section, we present and analyze an accelerated primal-dual coordinate descent algorithms to obtain an $\varepsilon$-approximate transportation plan for the OT problem~\eqref{prob:OT}. First, in Section~\ref{sec:algorithmic_framework}, we introduce the accelerated primal-dual randomized coordinate descent (APDRCD) method for the entropic regularized OT problem. Then, following the approximation scheme of~\citep{Altschuler-2017-Near}, we show how to approximate the OT distance within the APDRCD algorithm; see Algorithm~\ref{Algorithm:ApproxOT_APDCD} for the pseudo-code for this problem. Furthermore, we provide theoretical analysis to establish a complexity bound of $\mathcal{O}(\frac{n^{\frac{5}{2}}\sqrt{||C||_\infty\log(n)}}{\varepsilon})$ for the APDRCD algorithm to achieve an $\varepsilon$-approximate transportation plan for the OT problem in Section~\ref{sec:complexity}. This complexity upper bound of the APDRCD algorithm matches the best known complexity bounds of the APDAGD~\citep{Dvurechensky-2018-Computational} and APDAMD algorithms~\citep{lin2019efficient}. Finally, to further improve the practical performance of the algorithm, we present a greedy variant---the accelerated primal-dual greedy coordinate descent (APDGCD) algorithm---in Section~\ref{Sec:APDGCD}. 

\begin{algorithm}[!ht]
\caption{APDRCD ($C, \eta, A, b, \varepsilon'$)}
\label{Algorithm:APDCD}
\textbf{Input:}
$\{\theta_i| \theta_0=1, \frac{1-\theta_{i+1}}{\theta_{i+1}^2}=\frac{1}{\theta_i^2}\},  C_0 = 1, \lambda^0 = z^0 = k =0, L = \frac{4}{\eta}$ \\ 
\While{$||Ax^k-b||_1 > \varepsilon'$}
{Set $y^k = (1-\theta_k)\lambda^k+\theta_k z^k$ \\  
Compute $x^{k} = \frac{1}{C_{k}} \biggr(\sum_{j = 0}^{k} \dfrac{x(y^j)}{\theta_{j}} \biggr)$ \\
\textbf{Randomly sample one coordinate} $i_k$  \textbf{where} $i_k \in \{1,2,..., 2n\}$:\\
Update
\begin{equation}\label{a2}
    \lambda^{k+1}_{i_k} = y^k_{i_k} - \frac{1}{L}\nabla_{i_k}\varphi(y^k)
\end{equation}  \\
Update
\begin{equation}\label{a3}
    z^{k+1}_{i_k} = z^k_{i_k}-\frac{1}{2n L \theta_k} \nabla_{i_k} \varphi(y^k)
\end{equation}   \\
Update $k = k+1$ and $C_k = C_k  + \frac{1}{\theta_k}$
}
\textbf{Output:} $X^k$ where $x^k = vec(X^k)$
\end{algorithm}
\subsection{Accelerated Primal-Dual Randomized Coordinate Descent (APDRCD)}
\label{sec:algorithmic_framework}
We denote by $L$ the Lipschitz constant for the dual objective function $\varphi$, which means that $L \defeq \frac{4}{\eta}$, and $x(\lambda) : = \mathop{\arg \max}\limits_{x \in \br^{n \times n}} \biggr\{ - \left\langle C, x\right\rangle - \left\langle A^{\top}\lambda, x\right\rangle \biggr\}$. The APDRCD algorithm is initialized with the auxiliary sequence $\{\theta_i\}$ and two auxiliary dual variable sequences $\{\lambda_i\}$ and $\{\z_i\}$, where the first auxiliary sequence $\{\theta_k\}$ is used for the key averaging step and the two dual variable sequences are used to perform the accelerated randomized coordinate descent on the dual objective function $\varphi$ as a subroutine. The APDRCD algorithm is composed of two main parts. First, exploiting the convexity property of the dual objective function, we perform a randomized accelerated coordinate descent step on the dual objective function as a subroutine in step \ref{a2} and  \ref{a3}. In the second part, we take a weighted average over the past iterations to get a good approximate solution for the primal problem from the approximate solutions to the dual problem~\eqref{eq:dual_entropic}. Notice that the auxiliary sequence $\{\theta_k\}$ is decreasing and the primal solutions corresponding to the more recent dual solutions have larger weight in this average. 

\begin{algorithm}[!ht]
\caption{Approximating OT by APDRCD} \label{Algorithm:ApproxOT_APDCD}
\begin{algorithmic}
\STATE \textbf{Input:} $\eta = \dfrac{\varepsilon}{4\log(n)}$ and $\varepsilon'=\dfrac{\varepsilon}{8\left\|C\right\|_\infty}$. 
\STATE \textbf{Step 1:} Let $\tilde{r} \in \Delta_n$ and $\tilde{l} \in \Delta_n$ be defined as
\begin{equation*}
\left(\tilde{r}, \tilde{l}\right) = \left(1 - \frac{\varepsilon'}{8}\right) \left(r, l\right) + \frac{\varepsilon'}{8n}\left(\one, \one\right).  
\end{equation*}
\STATE \textbf{Step 2:} Let $A \in \mathbb{R}^{2n \times n^2}$ and $b \in \mathbb{R}^{2n}$ be defined by
\begin{equation*}
    Avec(X) = \icol{X\boldsymbol{1}\\X^T\boldsymbol{1}} \ \  \text{and} \ \  b = \icol{\tilde{r}\\\tilde{l}}
\end{equation*}
\STATE \textbf{Step 3:} Compute $\tilde{X} = \text{APDRCD}\left(C, \eta, A, b, \varepsilon'/2\right)$ with $\varphi$ defined in~\ref{eq:dual_entropic}.
\STATE \textbf{Step 4:} Round $\tilde{X}$ to $\hat{X}$ by Algorithm 2~\citep{Altschuler-2017-Near} such that $\hat{X}\one = r$, $\hat{X}^\top\one = l$. 
\STATE \textbf{Output:} $\hat{X}$.  
\end{algorithmic}
\end{algorithm}

\subsection{Complexity Analysis of APDRCD}
\label{sec:complexity}
Given the updates from APDRCD algorithm in Algorithm~\ref{Algorithm:APDCD}, we have the following result regarding the difference of the values of $\varphi$ at $\lambda^{k + 1}$ and $y^{k}$:
 \begin{lemma}\label{lemma2}
 Given the updates $\lambda^{k + 1}$ and $y^{k}$ from the APDRCD algorithm, we have the following inequality:
 \begin{equation*}
      \varphi(\lambda^{k+1})-\varphi(y^k) \leq -\frac{1}{2L}|\nabla_{i_k} \varphi (y^k)|^2,
 \end{equation*}
 where $i_{k}$ is chosen in the APDRCD algorithm.
 \end{lemma} 
\begin{proof}
For convenience, we define a vector-valued function $h(i_k)\in \mathbb{R}^{2n}$ such that $h(i_k)_i=1$ if $i=i_k$, and 
$h(i_k)_i=0$ otherwise.
By the update in Eq.~\eqref{a2} of Algorithm~\ref{Algorithm:APDCD}, we obtain: 
\begin{equation} \label{lemma1eq1}
    \varphi(\lambda^{k+1})-\varphi(y^k) = \varphi \biggr( y^k-h(i_k) \cdot \frac{1}{L}\nabla_{i_k}\varphi(y^k) \biggr)-\varphi(y^k).
\end{equation}
Due to the smoothness of $\varphi$ with respect to $\|.\|_{2}$ norm in Lemma~\ref{lemma:dual_smooth}, the following inequalities hold:
\begin{align} \label{lemma1eq2}
    \lefteqn{\varphi \biggr( y^k-h(i_k)\frac{1}{L}\nabla_{i_k}\varphi(y^k) \biggr)-\varphi(y^k)} \nonumber \\
    \leq & \left \langle \nabla \varphi(y^k),- h(i_k)\frac{1}{L}\nabla_{i_k}\varphi(y^k)\right\rangle \nonumber  \\
    & \ + \frac{L}{2}||h(i_k)\frac{1}{L}(\nabla_{i_k}\varphi(y^k))||^2 \nonumber \\
    =& -\frac{1}{L}\left \langle \nabla\varphi(y^k),\nabla_{i_k}\varphi(y^k)h(i_k)\right \rangle + \frac{1}{2L}(\nabla_{i_k}\varphi(y^k))^2 \nonumber \\
    =& -\frac{1}{L}(\nabla_{i_k}\varphi(y^k))^2 + \frac{1}{2L}(\nabla_{i_k}\varphi(y^k))^2 \nonumber \\
    =& -\frac{1}{2L}(\nabla_{i_k}\varphi(y^k))^2.  
\end{align}
Combining the results of Eq.~\eqref{lemma1eq1} and Eq.~\eqref{lemma1eq2} completes the proof of the lemma.
\end{proof}
The result of Lemma~\ref{lemma2} is vital for establishing an upper bound for $\mathbb{E}_{i_k}\varphi(\lambda^{k+1})$, as shown in the following lemma.
\begin{lemma}\label{lemma 3} For each iteration ($k>0$) of the APDRCD algorithm, we have
\begin{align*}
    \lefteqn{\mathbb{E}_{i_k} \big[\varphi(\lambda^{k+1}) \big]} \\
    & \leq  (1-\theta_k)\varphi(\lambda^k) +\theta_k \big[ \varphi(y^k)+(\lambda-y^k)^T\nabla\varphi(y^k) \big] \\
    & \ +2 L^2 n^2 \theta_k^2 \biggr(||\lambda-z^k||^2 - \mathbb{E}_{i_k} \big[||\lambda-z^{k+1}||^2 \big]\biggr),
\end{align*}
where the outer expectation in the above display is taken with respect to the random coordinate $i_{k}$ in Algorithm~\ref{Algorithm:APDCD}.
\end{lemma}

The proof of Lemma~\ref{lemma 3} is in Appendix~\ref{subsec:proof:lemm_3}. Now, equipped with the result of Lemma~\ref{lemma 3}, we are ready to provide the convergence guarantee and complexity bound of the APDRCD algorithm for approximating the OT problem. First, we start with the following result regarding an upper bound on $k$ to reach the stopping rule $||A \text{vec}(X^k)-b||_1 \leq \varepsilon'$ for $\varepsilon' = \dfrac{\varepsilon}{8\left\|C\right\|_\infty}$. Here, the outer expectation is taken with respect to the random coordinates $i_{j}$ in Algorithm~\ref{Algorithm:APDCD} for $1 \leq j \leq k$.
\begin{theorem} \label{thm: OTconvergence}
The APDRCD algorithm for approximating optimal transport (Algorithm~\ref{Algorithm:ApproxOT_APDCD}) returns an output $X^k$ that satisfies the stopping criterion $\mathbb{E} \big[||A \text{vec}(X^k)-b||_1\bigr] \leq \varepsilon'$ in a number of iterations k bounded as follows:
\begin{equation*}
    k \leq 12 n^{\frac{3}{2}}\sqrt{\frac{R+1/2}{\varepsilon}} +1,
\end{equation*}
where $R \defeq \frac{||C||_\infty}{\eta}+ \log(n) - 2 \log(\min \limits_{1\leq i,j\leq n} \{r_i,l_i\})$. Here, $\varepsilon'$ and $\eta$ are chosen in Algorithm~\ref{Algorithm:ApproxOT_APDCD}.
\end{theorem} 
The proof of Theorem~\ref{thm: OTconvergence} is provided in Appendix~\ref{subsec:proof:thm: OTconvergence}. Given an upper bound on $k$ for the stopping rule in Theorem~\ref{thm: OTconvergence} where $\varepsilon' = \dfrac{\varepsilon}{8\left\|C\right\|_\infty}$ in Theorem~\ref{thm: OTconvergence}, we obtain the following complexity bound for the APDRCD algorithm. 
\begin{theorem}
\label{theorem:complex_APDRCD}
The APDRCD algorithm for approximating optimal transport (Algorithm \ref{Algorithm:ApproxOT_APDCD}) returns $\hat{X}\in \mathbb{R}^{n\times n}$ satisfying $\hat{X}\textbf{1} = r$, $\hat{X}^T\textbf{1} = l$ and $\mathbb{E} [\langle C, \hat{X} \rangle ] - \left \langle C, X^\ast \right \rangle  \leq \epsilon$ in a total number of  
\begin{equation*}
    \mathcal{O} \biggr(\frac{n^{\frac{5}{2}}\sqrt{||C||_\infty\log(n)}}{\varepsilon}\biggr)
\end{equation*}
arithmetic operations.
\end{theorem}
The proof of Theorem~\ref{theorem:complex_APDRCD} is provided in Appendix~\ref{subsec:proof:theorem:complex_APDRCD}. We show in Appendix~\ref{subsec:proof:theorem:complex_APDRCD} that Theorem~\ref{theorem:complex_APDRCD} also directly implies a complexity bound for obtaining an $\epsilon$-optimal solution with high probability. Theorem~\ref{theorem:complex_APDRCD} indicates that the complexity upper bound of APDRCD matches the best known complexity $\bigOtil(\frac{n^{5/2}}{\varepsilon})$ of the APDAGD~\citep{Dvurechensky-2018-Computational} and APDAMD~\citep{lin2019efficient} algorithms. Furthermore, that complexity of APDRCD is better than that of the Sinkhorn and Greenkhorn algorithms, which is $\bigOtil(\frac{n^{2}}{\varepsilon^2})$, in terms of the desired accuracy $\varepsilon$. Later in Section~\ref{Sec:experiments}, we demonstrate that the APDRCD algorithm also has better practical performance than APDAGD and APDAMD algorithms on both synthetic and real datasets.

\begin{algorithm}[h]
\caption{APDGCD ($C, \eta, A, b, \varepsilon'$)}
\label{Algorithm:APDGCD}
\textbf{Input:} $\{\theta_i| \theta_0=1, \frac{1-\theta_{i+1}}{\theta_{i+1}^2}=\frac{1}{\theta_i^2}\},  C_0 = 1, \lambda^0 = z^0 = k =0, L = \frac{4}{\eta}$ \\
\While{$||Ax^k-b||_1 > \varepsilon'$}
{Set $y^k = (1-\theta_k)\lambda^k+\theta_k z^k$ \\
Compute $x^{k} = \frac{1}{C_{k}} \biggr(\sum_{j = 0}^{k} \dfrac{x(y^j)}{\theta_{j}} \biggr)$ \\
\textbf{Select coordinate} $i_k =\argmax \limits_{i_k \in \{1,2,..., 2n\} } |\nabla_{i_k}\varphi(y^k)|$:
Update 
\begin{align*}
    \lambda^{k+1}_{i_k} = y^k_{i_k} - \frac{1}{L}\nabla_{i_k}\varphi(y^k)
\end{align*}
Update 
\begin{align*}
    z^{k+1}_{i_k} = z^k_{i_k}-\frac{1}{2n L \theta_k} \nabla_{i_k} \varphi(y^k)
\end{align*} 
Update $k = k+1$ and $C_k = C_k  + \frac{1}{\theta_k}$
}
\textbf{Output:} $X^k$ where $x^k = vec(X^k)$
\end{algorithm}


\begin{figure*}[!ht]
\begin{minipage}[b]{.5\textwidth}
\includegraphics[width=65mm,height=48mm]{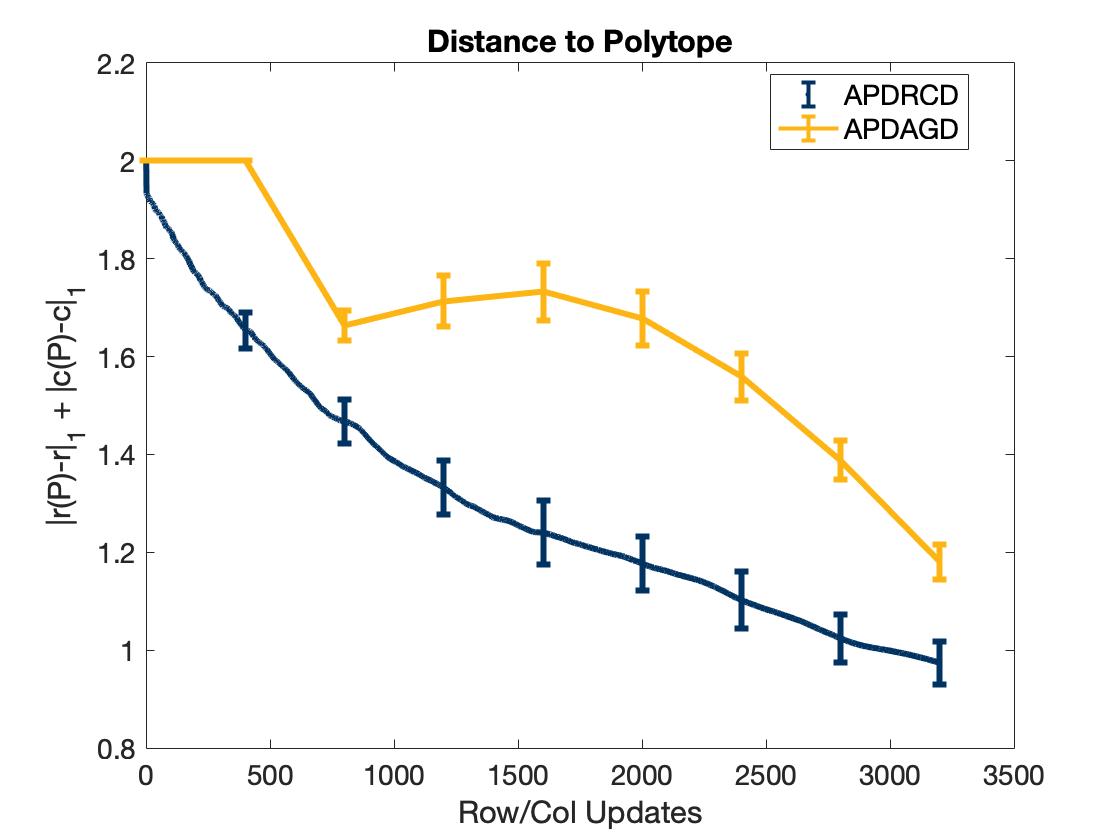}
\end{minipage}
\quad
\begin{minipage}[b]{.5\textwidth}
\includegraphics[width=65mm,height=48mm]{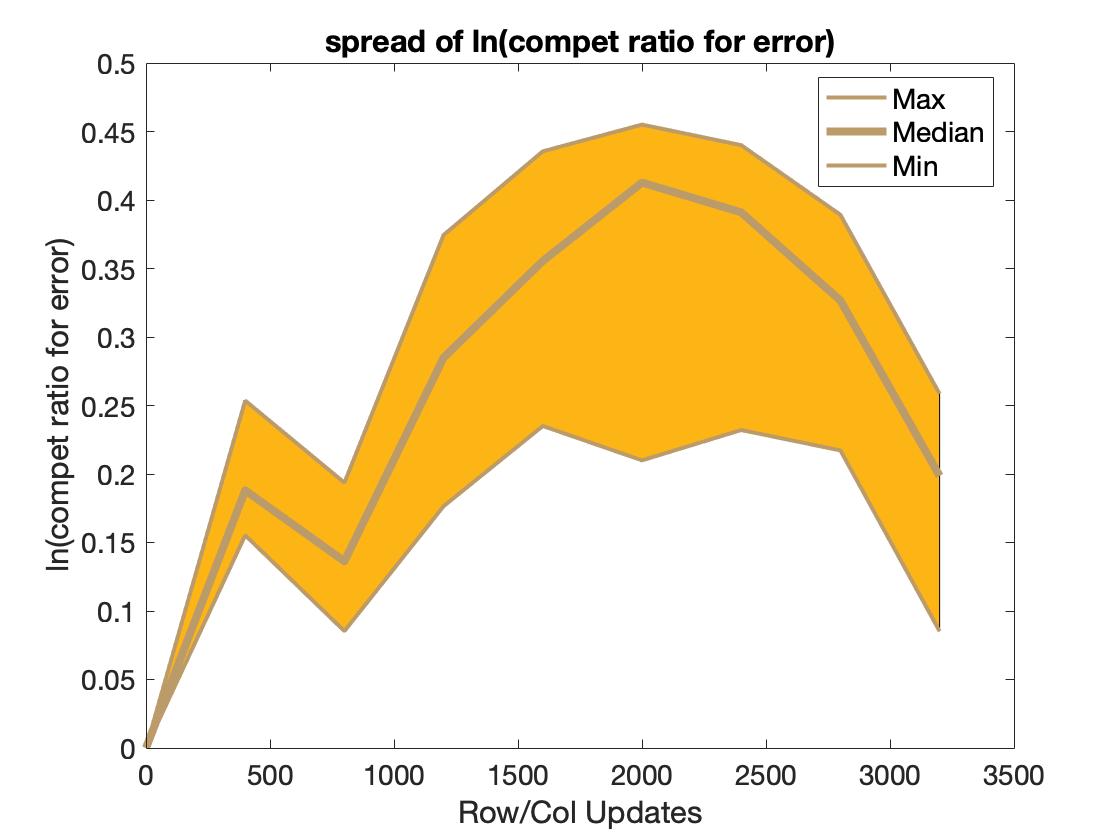}
\end{minipage}
\\
\begin{minipage}[b]{.5\textwidth}
\includegraphics[width=65mm,height=48mm]{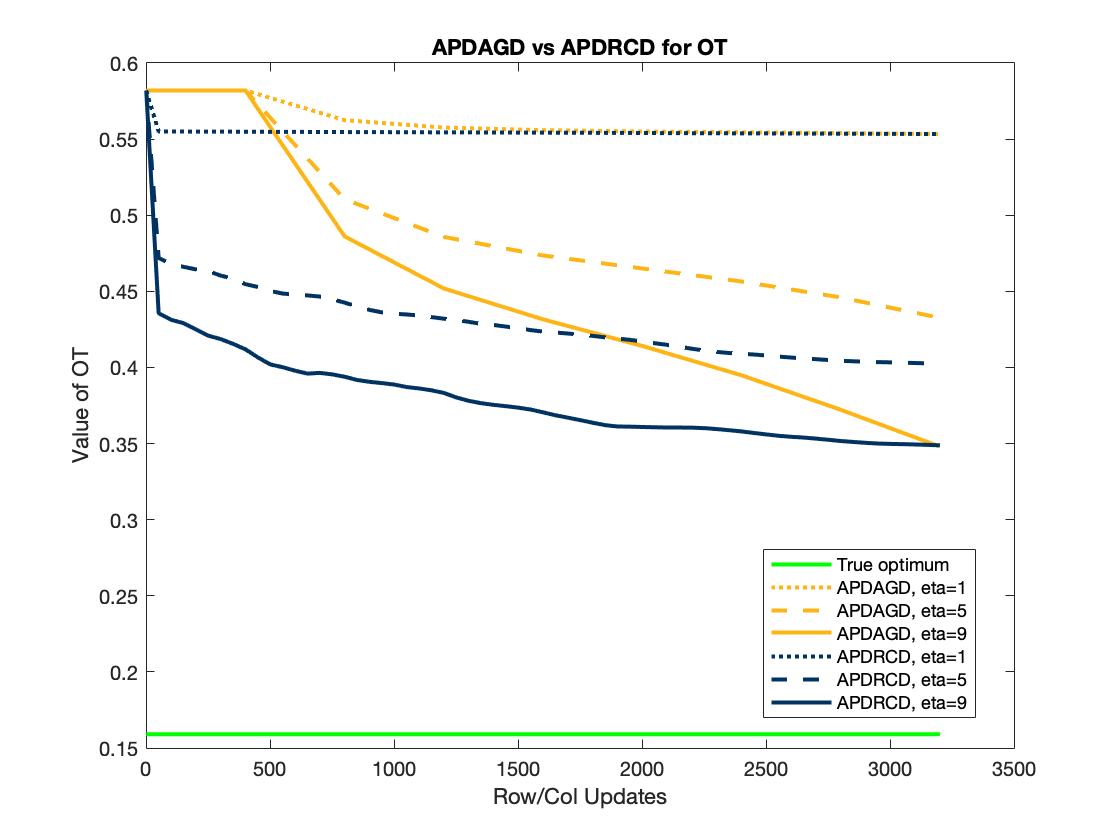}
\end{minipage}
\quad
\begin{minipage}[b]{.5\textwidth}
\includegraphics[width=65mm,height=48mm]{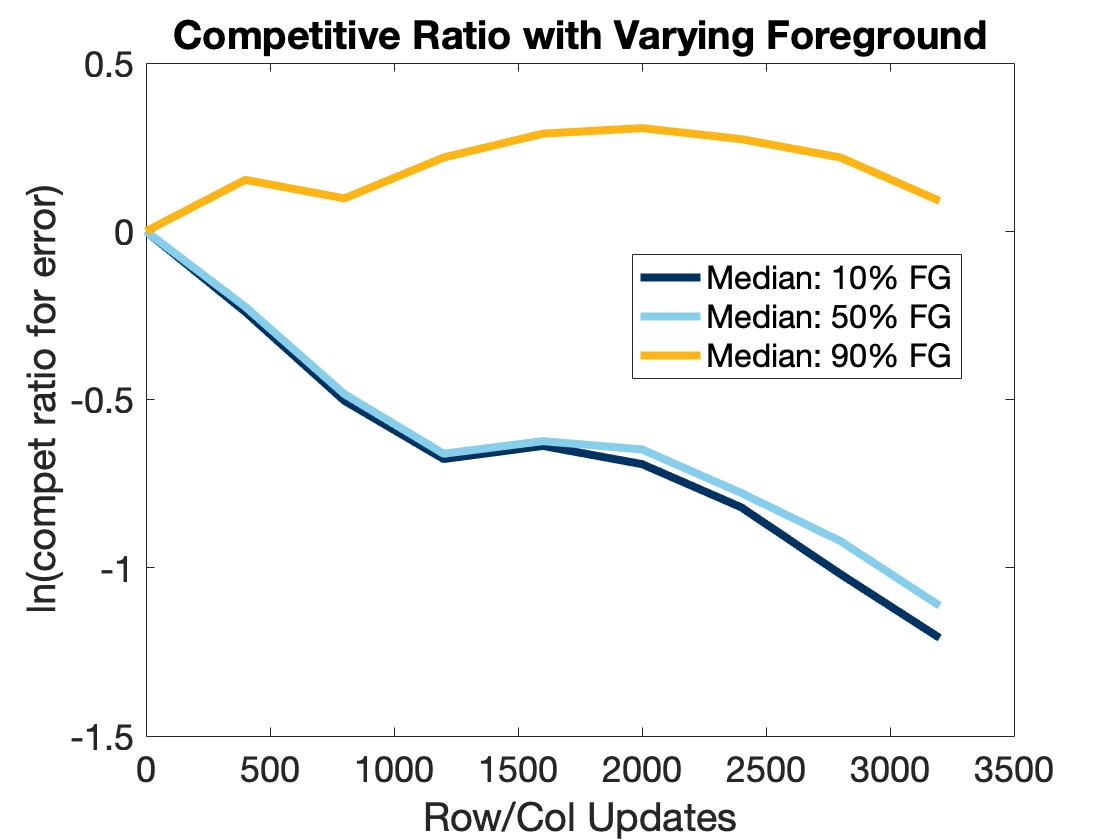}
\end{minipage}
\caption{\small Performance of APDRCD and APDAGD algorithms on the synthetic images. In
the top two images, the comparison is based on using the distance $d(P)$ to the transportation
polytope, and the maximum, median and minimum of competitive ratios
on ten random pairs of images. In the bottom left image, the comparison is based on varying the
regularization parameter $\eta\in\{1, 5, 9\}$ and reporting the optimal value of the original optimal
transport problem without entropic regularization. Note that the foreground covers 10\% of
the synthetic images here. In the bottom right image, we compare the algorithms by using the median of
competitive ratios with varying coverage ratio of foreground in the range of $\{0.1, 0.5, 0.9\}$.}
\label{fig:rcd-agd-syn}
\end{figure*}
\subsection{Accelerated Primal-Dual Greedy Coordinate Descent (APDGCD)}
\label{Sec:APDGCD}
We next present a greedy version of APDRCD algorithm, which we refer to as the \emph{accelerated primal-dual greedy coordinate descent} (APDGCD) algorithm. The detailed pseudo-code of that algorithm is in Algorithm~\ref{Algorithm:APDGCD} while an approximating scheme of OT based on the APDGCD algorithm is summarized in Algorithm~\ref{Algorithm:ApproxOT_APDGCD}. 

Both the APDGCD and APDRCD algorithms follow along the general accelerated primal-dual coordinate descent framework.
Similar to the APDRCD algorithm, the algorithmic framework of APDGCD is composed by
two main parts: First, instead of performing
randomized accelerated coordinate descent on the dual objective function as a subroutine, the APDGCD algorithm chooses the best coordinate that maximizes the absolute value of the gradient of the dual objective function of regularized OT problem among all the coordinates. In the second part, we follow the key averaging step in the APDRCD algorithm by taking a weighted average over
the past iterations to get a good approximated solution for the primal problem from the approximated
solutions to the dual problem. Since the auxiliary sequence is decreasing, the
primal solutions corresponding to the more recent dual solutions have larger weight in this average.

We demonstrate that the APDGCD algorithm enjoys favorable practical performance than APDRCD algorithm on both synthetic and real datasets (cf.\ Appendix~\ref{subsec:further_exp}).

\begin{algorithm}[h]
\caption{Approximating OT by APDGCD} \label{Algorithm:ApproxOT_APDGCD}
\begin{algorithmic}
\STATE \textbf{Input:} $\eta = \dfrac{\varepsilon}{4\log(n)}$ and $\varepsilon'=\dfrac{\varepsilon}{8\left\|C\right\|_\infty}$. 
\STATE \textbf{Step 1:} Let $\tilde{r} \in \Delta_n$ and $\tilde{l} \in \Delta_n$ be defined as
\begin{equation*}
\left(\tilde{r}, \tilde{l}\right) = \left(1 - \frac{\varepsilon'}{8}\right) \left(r, l\right) + \frac{\varepsilon'}{8n}\left(\one, \one\right).  
\end{equation*}
\STATE \textbf{Step 2:} Let $A \in \mathbb{R}^{2n \times n^2}$ and $b \in \mathbb{R}^{2n}$ be defined by
\begin{equation*}
    Avec(X) = \icol{X\boldsymbol{1}\\X^T\boldsymbol{1}} \ \  \text{and} \ \  b = \icol{\tilde{r}\\\tilde{l}}
\end{equation*}
\STATE \textbf{Step 3:} Compute $\tilde{X} = \text{APDGCD}\left(C, \eta, A, b, \varepsilon'/2\right)$ with $\varphi$ defined in~\ref{eq:dual_entropic}.
\STATE \textbf{Step 4:} Round $\tilde{X}$ to $\hat{X}$ by Algorithm 2~\citep{Altschuler-2017-Near} such that $\hat{X}\one = r$, $\hat{X}^\top\one = l$. 
\STATE \textbf{Output:} $\hat{X}$.  
\end{algorithmic}
\end{algorithm}

\begin{figure*}[!ht]
\begin{minipage}[b]{.5\textwidth}
\includegraphics[width=65mm,height=48mm]{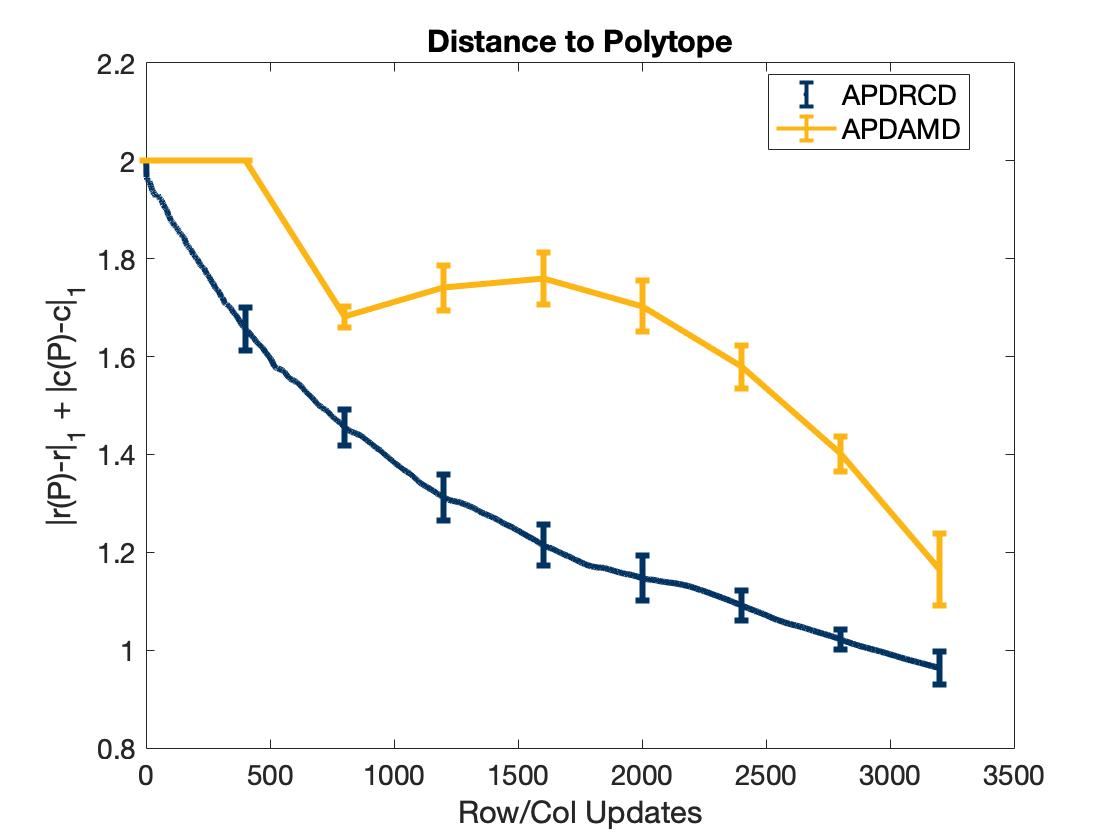}
\end{minipage}
\quad
\begin{minipage}[b]{.5\textwidth}
\includegraphics[width=65mm,height=48mm]{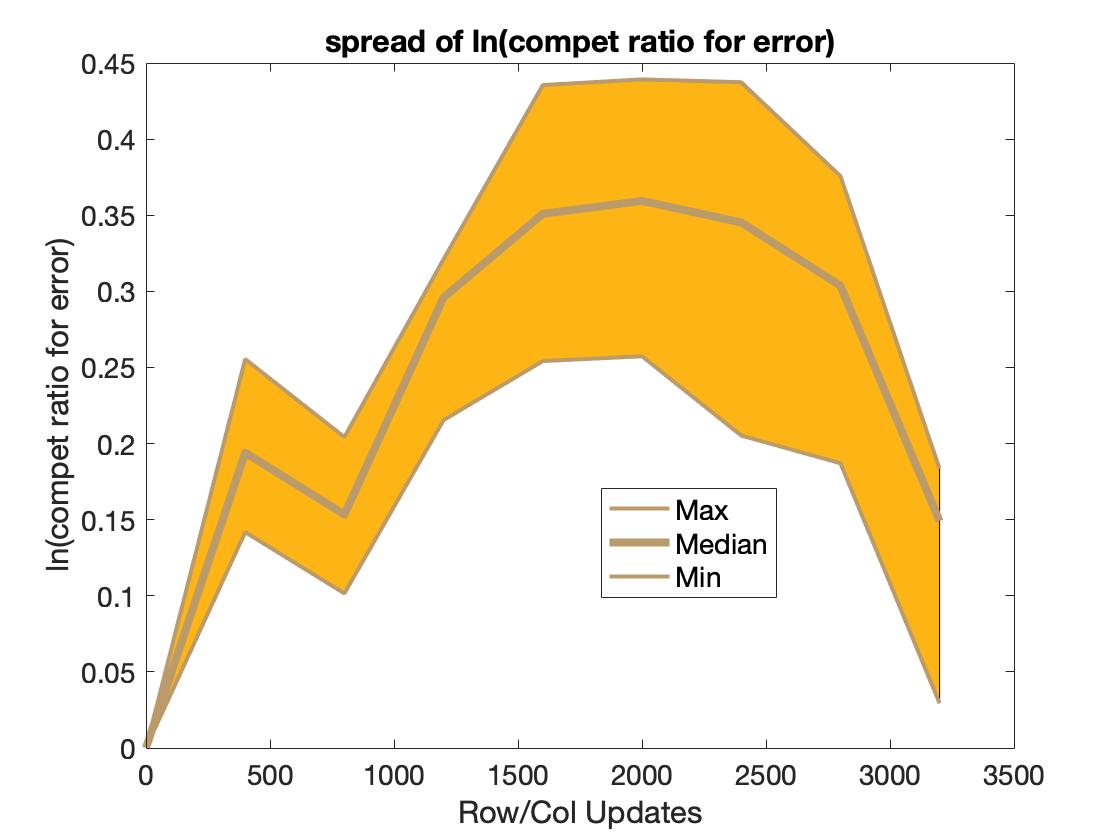}
\end{minipage}
\\
\begin{minipage}[b]{.5\textwidth}
\includegraphics[width=65mm,height=48mm]{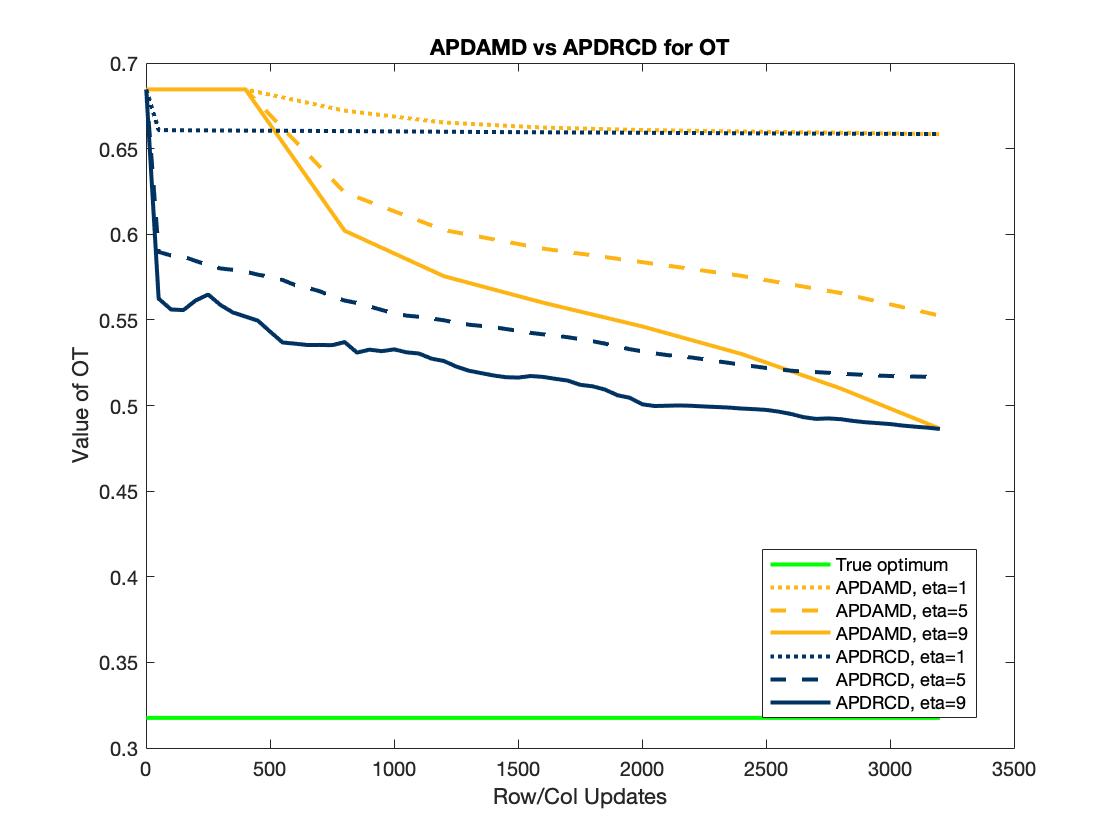}
\end{minipage}
\quad
\begin{minipage}[b]{.5\textwidth}
\includegraphics[width=65mm,height=48mm]{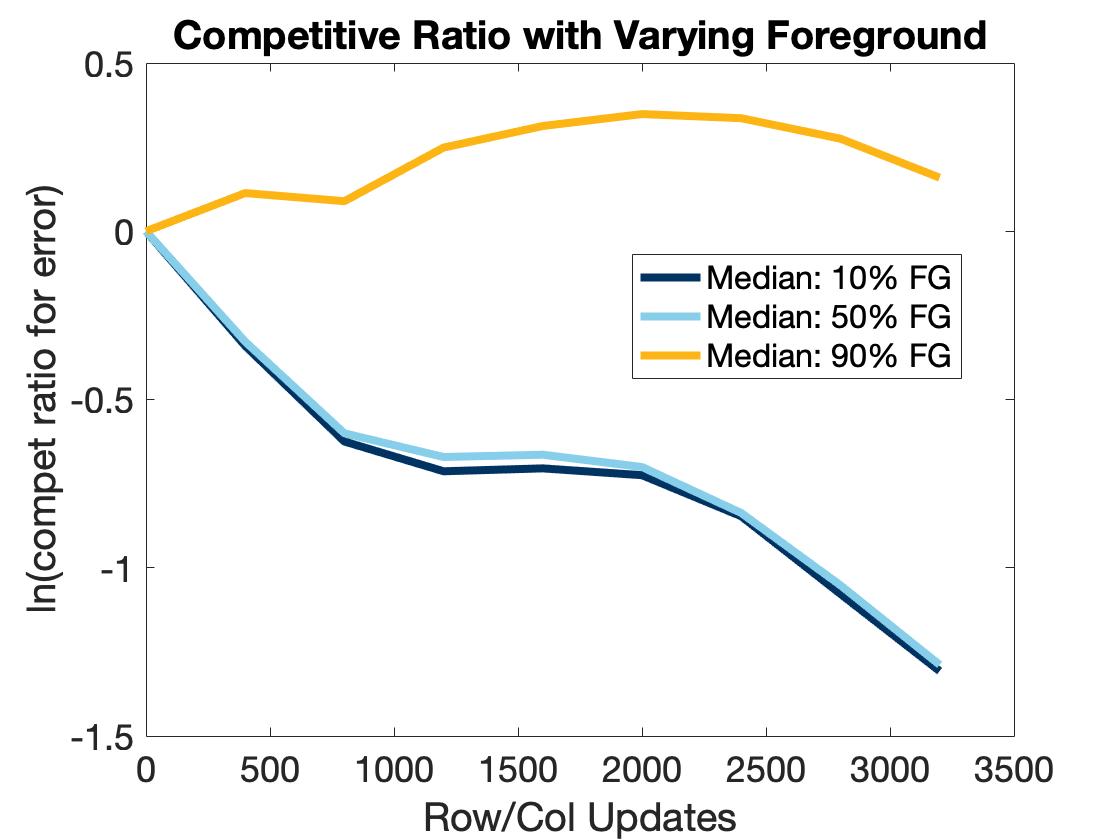}
\end{minipage}
\caption{\small Performance of APDRCD and APDAMD algorithms on the synthetic images. The organization of the images is similar to those in Figure~\ref{fig:rcd-agd-syn}.}
\label{fig:rcd-amd-syn}
\end{figure*}

\section{Experiments}
\label{Sec:experiments}

We carry out comparative experiments between the APDRCD, APDGCD algorithms and the state-of-art primal-dual algorithms for the OT problem including the APDAGD and APDAMD algorithms, on both synthetic images and the MNIST Digits dataset.\footnote{\href{http://yann.lecun.com/exdb/mnist/}{http://yann.lecun.com/exdb/mnist/}} Due to space constraints, the comparative experiments between the APDGCD algorithm and APDAGD/APDAMD algorithms, further experiments for the APDRCD algorithm on larger synthetic datasets and CIFAR10 dataset are deferred to Appendix~\ref{subsec:further_exp}. Note that for the above comparisons, we also utilize the default linear programming solver in MATLAB to obtain the optimal value of the
original optimal transport problem without entropic regularization. 
\subsection{APDRCD Algorithm with Synthetic Images}
\label{subsec:APDRCD_synthetic}
We compared the performance of the APDRCD algorithm with the APDAGD and APDAMD algorithms on synthetic images. The generation of synthetic images follows the procedure of~\citep{Altschuler-2017-Near, lin2019efficient}. The images are of size $20\times20$ and generated by randomly placing a foreground square in the otherwise black background. For the intensities of the background pixels and foreground pixels, we choose uniform distributions on [0,1] and [0, 50] respectively. We vary the proportion of the size of the foreground square in {0.1, 0.5, 0.9} of the full size of the image and implement all the algorithms on different kinds of synthetic
images.

\textbf{Evaluation metric:} We utilize the metrics from~\citep{Altschuler-2017-Near}. The first metric is the distance between the
output of the algorithm and the transportation polytope, $d(X) : = ||r(X)-r||_1+||l(X)-1||_1$, where $r(X)$ and $l(X)$ are the row and column marginal vectors of the output matrix $X$ while $r$ and
$l$ stand for the true row and column marginal vectors. The second metric is the competitive ratio, defined by $\log(\frac{d(X1)}{d(X2)})$ where $d(X1)$ and $d(X2)$ refer to the distance between the
outputs of two algorithms and the transportation polytope.

\begin{figure*}[!ht]
\begin{minipage}[b]{.3\textwidth}
\includegraphics[width=52mm,height=39mm]{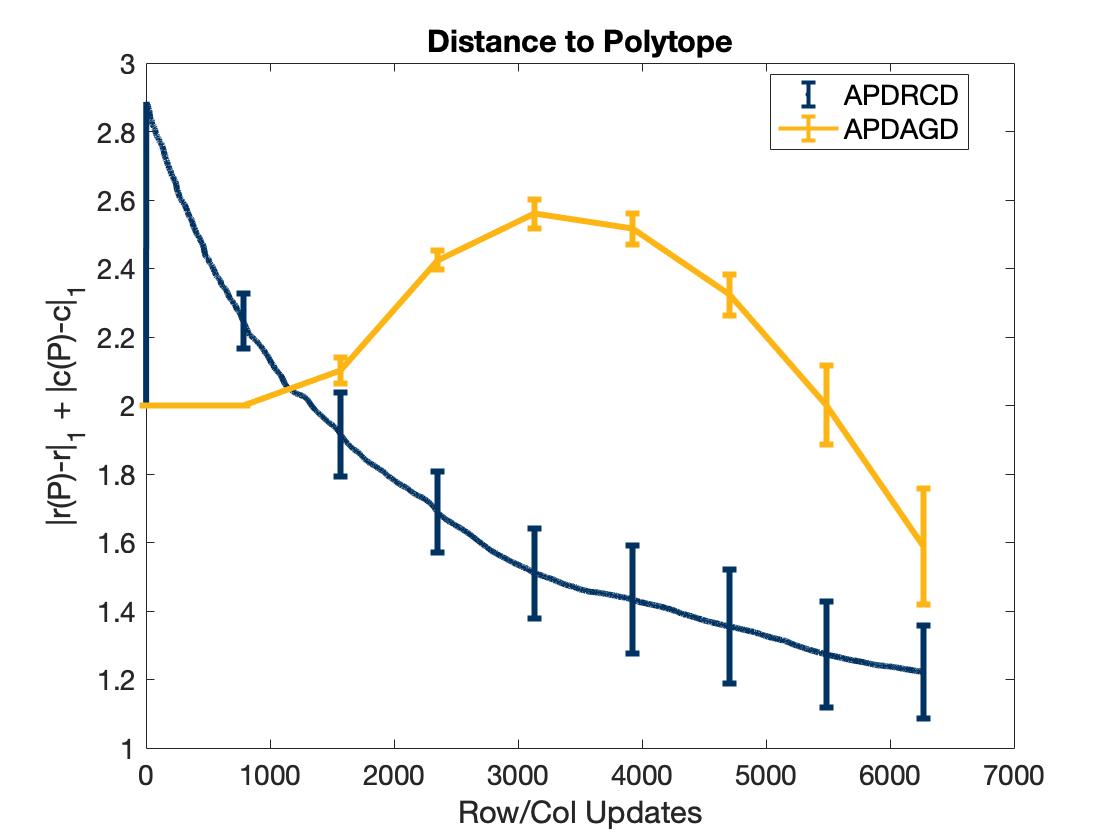}
\end{minipage}
\quad
\begin{minipage}[b]{.3\textwidth}
\includegraphics[width=52mm,height=39mm]{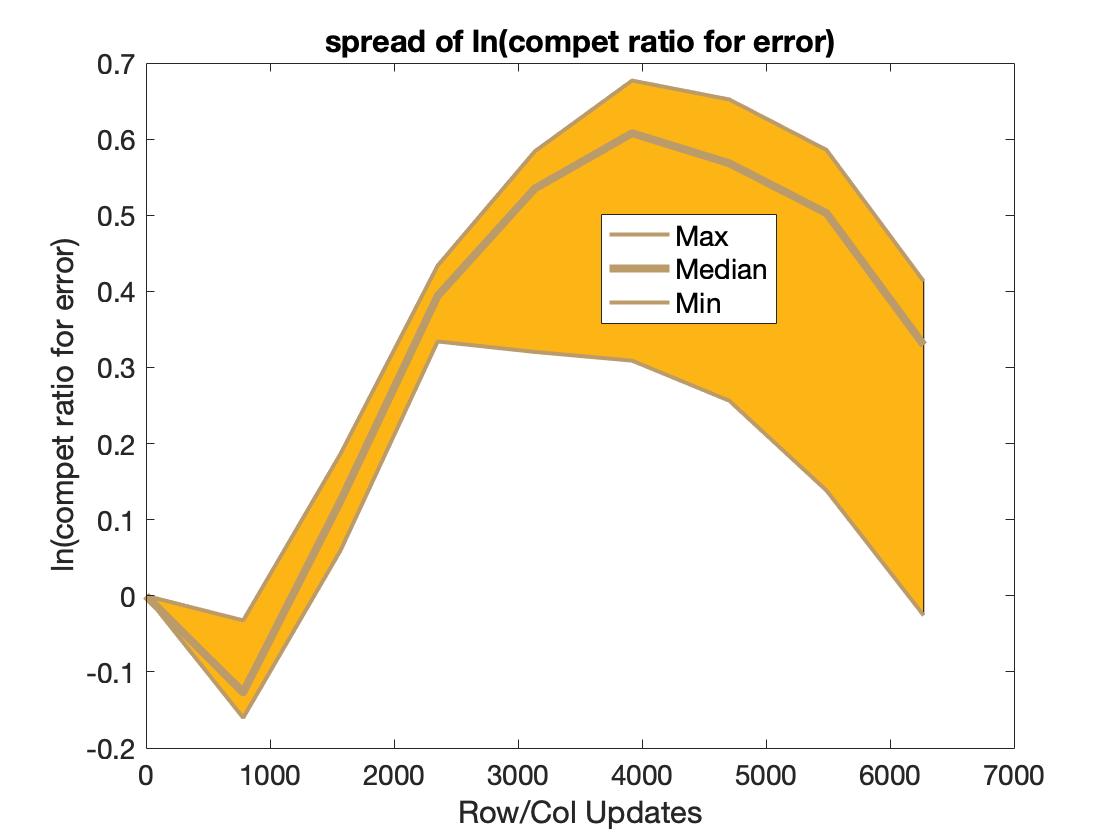}
\end{minipage}
\quad
\begin{minipage}[b]{.3\textwidth}
\includegraphics[width=52mm,height=39mm]{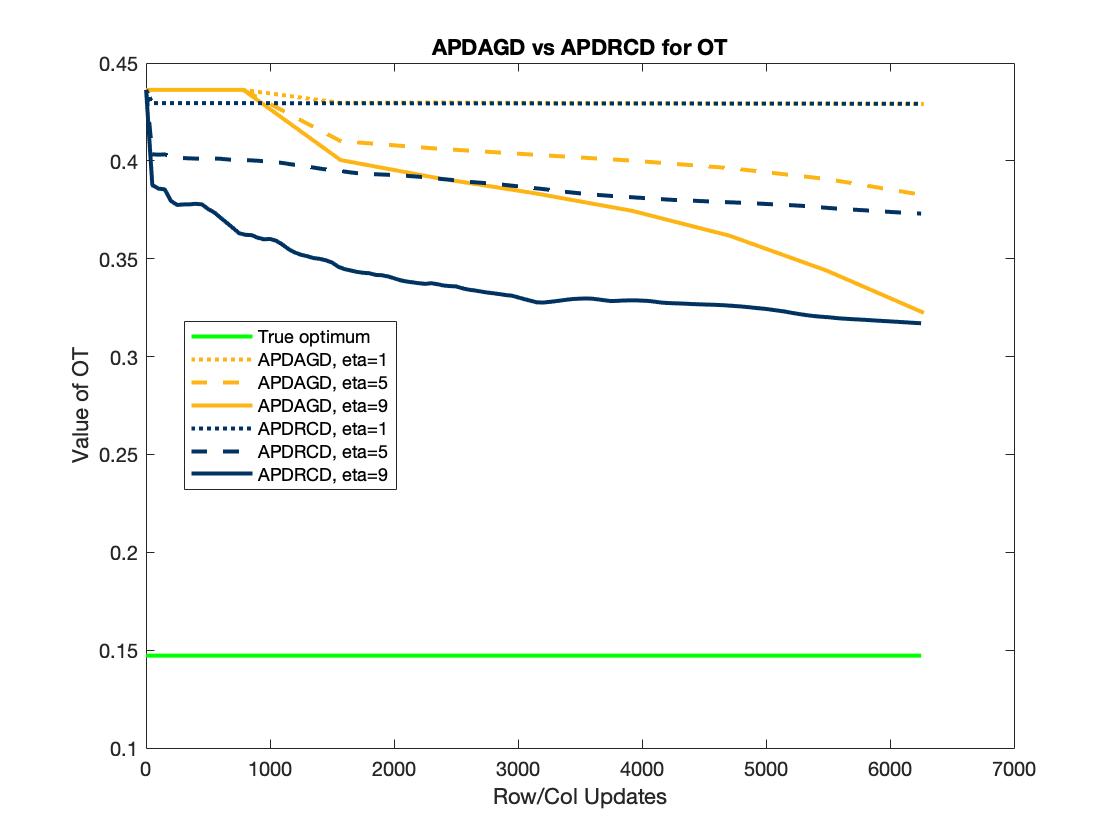}
\end{minipage}
\\
\begin{minipage}[b]{.3\textwidth}
\includegraphics[width=52mm,height=39mm]{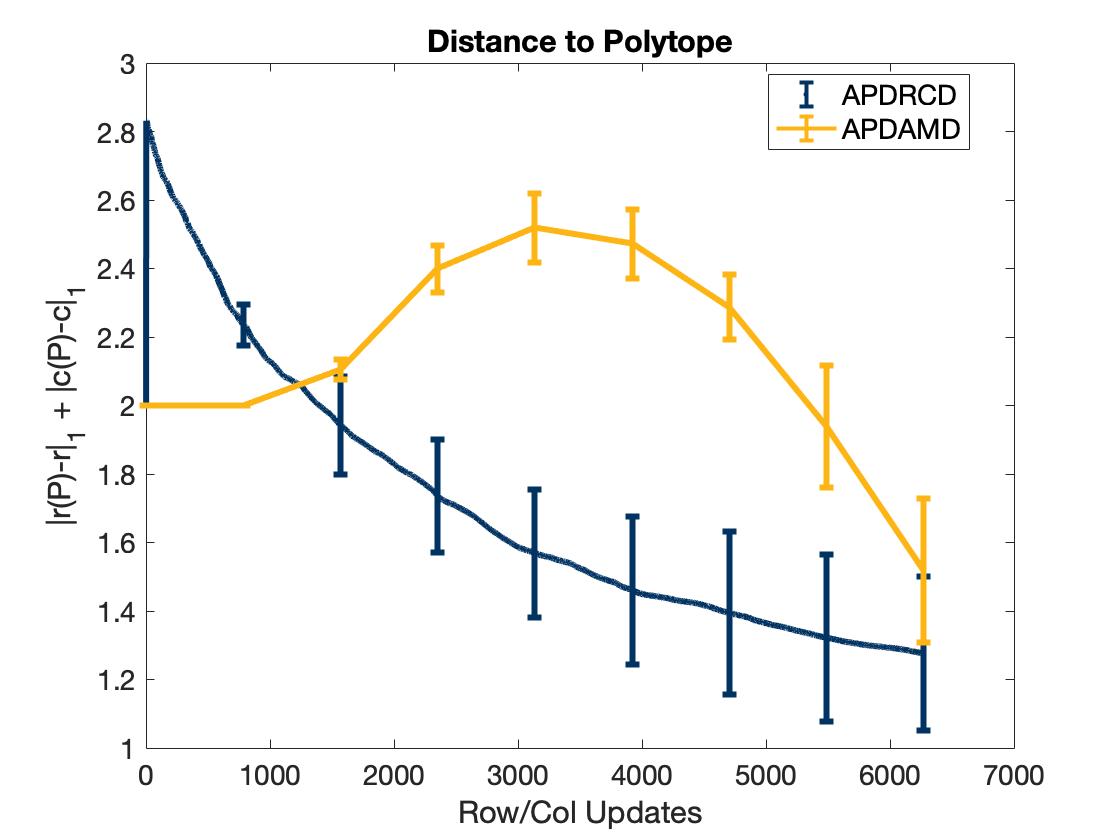}
\end{minipage}
\quad
\begin{minipage}[b]{.3\textwidth}
\includegraphics[width=52mm,height=39mm]{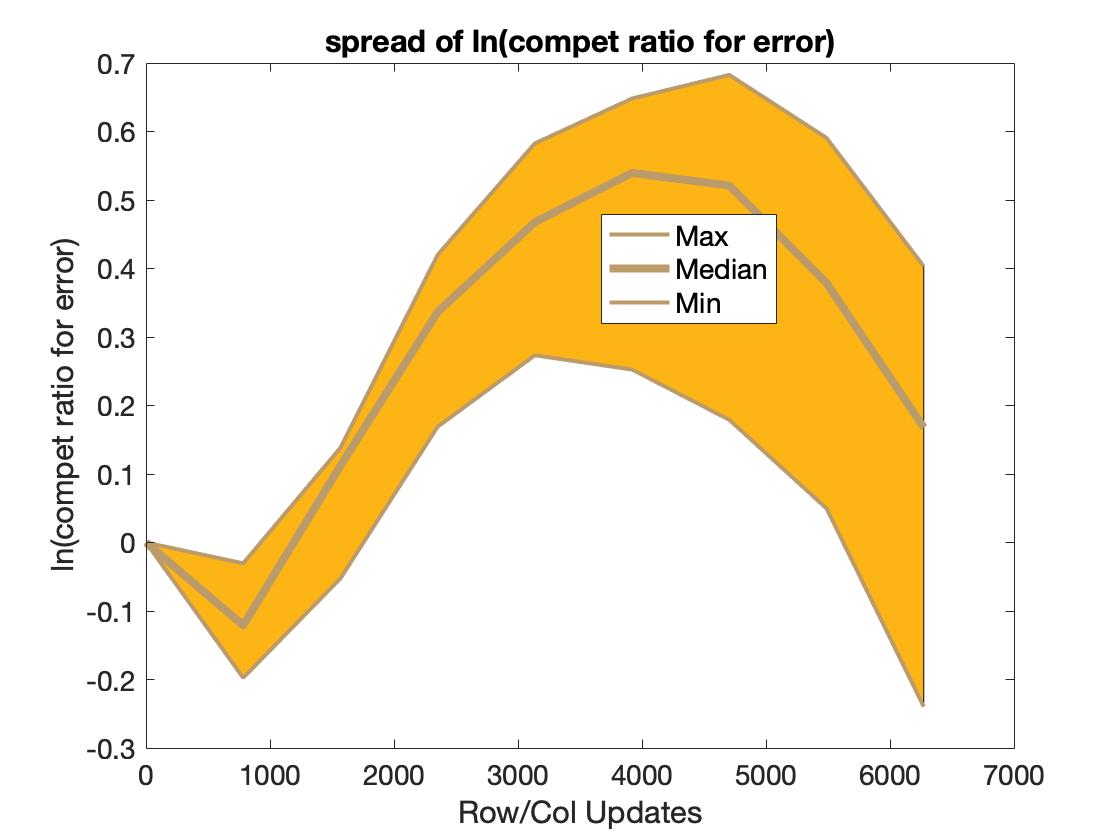}
\end{minipage}
\quad
\begin{minipage}[b]{.3\textwidth}
\includegraphics[width=52mm,height=39mm]{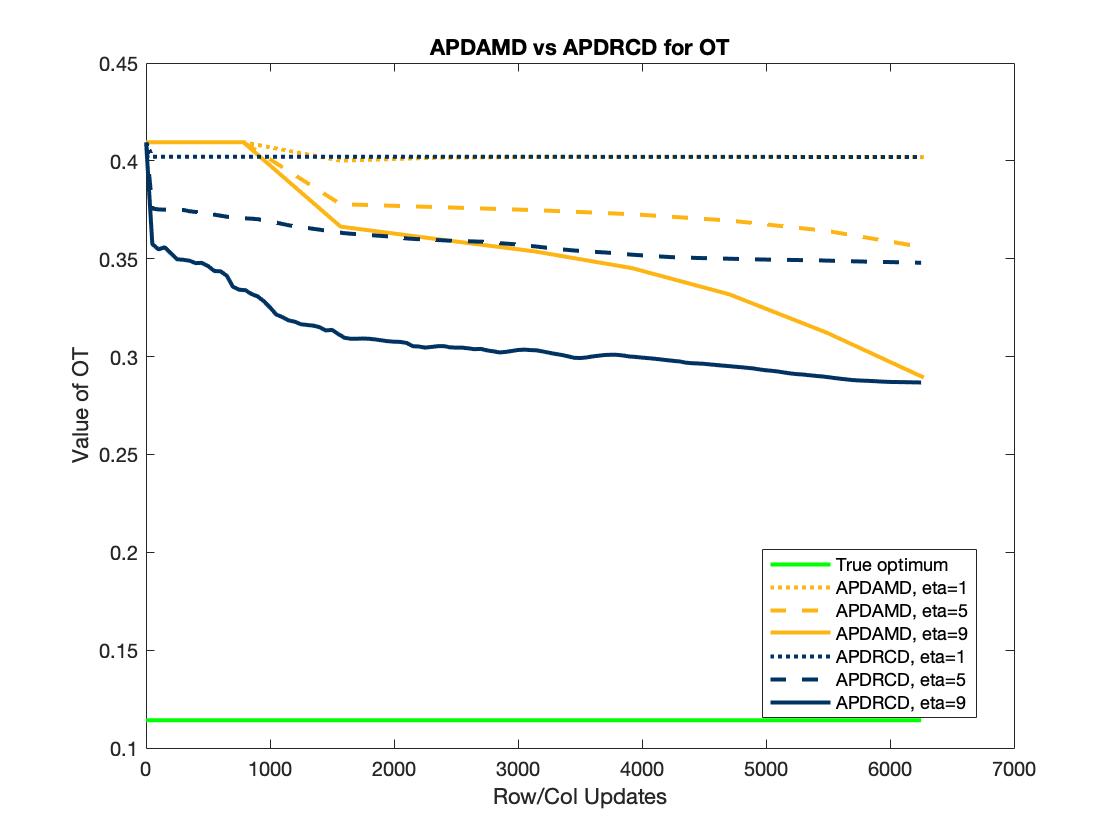}
\end{minipage}
\caption{\small Performance of the APDRCD, APDAGD and APDAMD algorithms on
the MNIST images. In the first row of images, we compare the APDRCD and APDAGD
algorithms in terms of iteration counts. The leftmost image specifies the distances $d(P)$ to
the transportation polytope for two algorithms; the middle image specifies the maximum,
median and minimum of competitive ratios on ten random pairs of MNIST images; the rightmost image specifies the values of regularized OT with varying regularization parameter $\eta = \{1,5,9\}$. In addition, the second row of images present comparative results for
APDRCD versus APDAMD.}
\label{fig:rcd-agd-amd-mnist}
\end{figure*}
\textbf{Experimental settings and results:}
We perform two pairwise comparative experiments for the APDRCD algorithm versus the APDAGD and APDAMD algorithms by running these algorithms with ten randomly selected pairs of synthetic images. We also evaluate
all the algorithms with varying regularization parameter $\eta \in \{1, 5, 9\}$ and the optimal value of the original OT problem without the entropic regularization, as suggested by~\citep{Altschuler-2017-Near, lin2019efficient}.

We present the results in Figure~\ref{fig:rcd-agd-syn} and Figure~\ref{fig:rcd-amd-syn}. The APDRCD algorithm has better performance than the APDAGD and APDAMD algorithms in terms of the iterations. When the number of iterations is small, the APDRCD algorithm achieves faster and more stable decrements than the other two algorithms with regard to both the distance to polytope and the value of OT during the computing process, which is beneficial for the purposes of tuning and illustrates the advantage of using randomized coordinate descent on the dual regularized problem.

\subsection{APDRCD Algorithm with MNIST Images}
\label{subsec:APDRCD_mnist}
We compare the APDRCD algorithm and the APDAGD and APDAMD algorithms on MNIST images with the same set of evaluation metrics. The image pre-processing follows the same pre-processing procedure as suggested in~\citep{lin2019efficient}. We omit the details for the sake of brevity.

We present the results with the MNIST images in Figure~\ref{fig:rcd-agd-amd-mnist} with various values for the  regularization parameter $\eta \in \{1, 5, 9\}$. We also evaluate the algorithms with the optimal value
of the original OT problem without entropic regularization. As shown in Figure~\ref{fig:rcd-agd-amd-mnist}, the APDRCD algorithm outperforms both the APDAGD and APDAMD algorithms on the MNIST dataset in terms of the number of iterations. Additionally, the APDRCD algorithm displays faster and smoother convergence than the other algorithms at small iteration numbers with regard to both the evaluation metrics, which implies the advantage that it is easier to be tuned in practice. 

\FloatBarrier

\section{Discussion}
\label{Sec:discussion}
We have proposed and analyzed new accelerated primal-dual coordinate descent algorithms for approximating the optimal transport distance between two discrete probability measures. These accelerated primal-dual coordinate descent algorithms have comparable theoretical complexity as that of existing accelerated primal-dual algorithms while enjoying better experimental performance.
Furthermore, we show that the APDRCD and APDGCD algorithms are suitable for other large-scale problems apart from computing the OT distance; we propose extensions that approximate the Wasserstein barycenters for multiple probability distributions. 

There are several directions for future work. Given the favorable practical performance of the APDRCD and the APDGCD algorithms over existing primal-dual algorithms, it is of interest to carry out more experiments of the distributed APDRCD and APDGCD algorithms for computing Wasserstein barycenters. Another important direction is to construct fast distributed algorithms for the case of time-varying and directed machine networks to approximate the Wasserstein barycenters. It remains as an interesting and important open research question how the dynamics of the network affect the performance of these algorithms. 

\subsubsection*{Acknowledgements}
We thank Tianyi Lin for many helpful discussions.
This work was supported in part by the Mathematical Data Science program of the Office of Naval Research under grant number N00014-18-1-2764.

\newpage
\bibliography{References}

\clearpage
\begin{center}
\textbf{\Large{Appendix}}
\end{center}
In the appendix, we first provide detailed proofs of all the remaining results in Section~\ref{sec:supplementary_material}. Then, we present an application of APDRCD and APDGCD algorithms for the Wasserstein barycenter problem in Section~\ref{Sec:WB}. Finally,  further comparative experiments between APDGCD algorithm versus APDRCD, APDAGD and APDAMD algorithms, and experiments on larger synthetic image datasets and CIFAR10 dataset are in Section~\ref{subsec:further_exp}.
\appendix
\section{Proofs for all results} 
\label{sec:supplementary_material}
In this appendix, we provide the complete proofs for remaining results in the main text.
\subsection{Proof of Lemma~\ref{lemma 3}}
\label{subsec:proof:lemm_3}
By Eq.~\eqref{a2}, we have the following equations 
\begin{align*}
    \lefteqn{|\nabla\varphi_{i_k}(y^k)|^2} \\
    &=  2|\nabla\varphi_{i_k}(y^k)|^2 -|\nabla\varphi_{i_k}(y^k)|^2\nonumber \\
    &=2 L\big(y_{i_k}^k-\lambda_{i_k}^{k+1}\big)\nabla_{i_k}\varphi(y^k)-L^2(\lambda_{i_k}^{k+1}-y_{i_k}^{k+1})^2.
\end{align*}
Combining the above equations with Lemma~\ref{lemma2}, we have
\begin{align} \label{15}
    \lefteqn{\varphi(\lambda^{k+1}) \nonumber}\\
    &\leq \varphi(y^k)-\frac{1}{2L}(\nabla_{i_k}\varphi(y^k))^2 \nonumber \\
    &= \varphi(y^k)+(\lambda_{i_k}^{k+1}-y_{i_k}^k)\nabla_{i_k}\varphi(y^k)+\frac{L}{2}(\lambda_{i_k}^{k+1}-y_{i_k}^{k+1})^2. 
\end{align}
Furthermore, the results from Eq.~\eqref{a2} and Eq.~\eqref{a3} lead to the following equations
\begin{align*}
    \lambda_{i_k}^{k+1}-y_{i_k}^{k} & = -\frac{1}{L}\nabla_{i_k}\varphi(y^k), \\
    z_{i_k}^{k+1}-z_{i_k}^{k} & = -\frac{1}{2n L\theta_k}\nabla_{i_k}\varphi(y^k).
\end{align*}
Therefore, we have
\begin{equation*}
     \lambda_{i_k}^{k+1}-y_{i_k}^{k} =2n\theta_k(z_{i_k}^{k+1}-z_{i_k}^{k}).
\end{equation*}
The above equation together with Eq.~\eqref{15} yields the following inequality
\begin{align}\label{17}
 \lefteqn{\varphi(\lambda^{k+1})}\nonumber \\
 &\leq \varphi(y^k)+ 2n\theta_k \big(z_{i_k}^{k+1}-z_{i_k}^k \big)\nabla_{i_k}\varphi_{i_k}(y^k) 
 \nonumber \nonumber\\ 
 &+ 2n^2L\theta_k^2 \big(z_{i_k}^{k+1}-z_{i_k}^k \big)^2. 
\end{align}
By the result of Eq.~\eqref{a3}, we have 
\begin{equation*}
    (z_{i_k}^{k+1}-z_{i_k}^k)+\frac{1}{2n L\theta_k}\nabla_{i_k}\varphi(y^k) = 0.
\end{equation*}
Therefore, for any $\lambda \in \mathbb{R}^{2n}$, we find that
\begin{equation*} 
    (\lambda_{i_k}-z_{i_k}^{k+1})[(z_{i_k}^{k+1}-z_{i_k}^k)+\frac{1}{2n L\theta_k}\nabla_{i_k}\varphi(y^k)] = 0.
\end{equation*}
Note that the above equation is equivalent to the following:
\begin{align*}
    \lefteqn{\frac{1}{n L\theta_k}(\lambda_{i_k}-z_{i_k}^{k+1})\nabla_{i_k}\varphi(y^k)} \\
    &= -2 (\lambda_{i_k}-z_{i_k}^{k+1})(z_{i_k}^{k+1}-z_{i_k}^k) \nonumber\\
    &=(\lambda_{i_k}-z_{i_k}^{k+1})^2-(\lambda_{i_k}-z_{i_k}^k)^2+(z_{i_k}^{k+1}-z_{i_k}^k)^2
\end{align*}
where the second equality in the above display comes from simple algebra. Rewriting the above equality, we have: 
\begin{align*}
\lefteqn{(z_{i_k}^{k+1}-z_{i_k}^k)^2} \\
&=\frac{1}{n L\theta_k}(\lambda_{i_k} -z_{i_k}^{k+1})\nabla_{i_k}\varphi(y^k)\\
&-(\lambda_{i_k}-z_{i_k}^{k+1})^2 + (\lambda_{i_k}-z_{i_k}^k)^2  
\end{align*}
Combining the above equation with Eq.~\eqref{17} yields the following inequality:
\begin{align}\label{21}
    \lefteqn{\varphi(\lambda^{k+1})}\nonumber\\
    &\leq  \varphi(y^k)+ 2n \theta_k (\lambda_{i_k}-z_{i_k}^k)\nabla_{i_k}\varphi(y^k) \nonumber\\
    & + 2 n^2 L\theta^2_k \biggr[(\lambda_{i_k}-z_{i_k}^k)^2-(\lambda_{i_k}-z_{i_k}^{k+1})^2 \biggr]. 
\end{align}
Recall the definition of $y^{k}$ in Step 3 of Algorithm~\ref{Algorithm:APDCD} as follows: 
\begin{equation*}
    y^k = (1-\theta_k)\lambda^k+\theta^k z^k,
\end{equation*}
which can be rewritten as: 
\begin{equation}\label{simple22}
    \theta_k(\lambda-z^k)=\theta_k(\lambda-y^k)+(1-\theta_k)(\lambda^k-y^k)
\end{equation}
for any $\lambda \in \mathbb{R}^{2n}$. 

The above equation implies that
\begin{align*}
\lefteqn{\varphi(y^k)+\theta_k(\lambda-z^k)^T\nabla\varphi(y^k)} \nonumber \\
&\leq \theta_k[\varphi(y^k)+(\lambda-y^k)^T\nabla\varphi(y^k)] \\
&+ (1-\theta_k)[\varphi(y^k)+(\lambda^k-y^k)^T\nabla\varphi(y^k)] \nonumber \\
&\leq \theta_k[\varphi(y^k)+(\lambda-y^k)^T\nabla\varphi(y^k)] + (1-\theta_k)\varphi(\lambda^k)
\end{align*}
where the last inequality comes from the convexity of $\varphi$.
Combining this equation and taking expectation over $i_k$ for the first two terms of Eq.~\eqref{21}, we have:
\begin{align}\label{lemma3part1}
     \lefteqn{\varphi(y^k)+(2n\theta_k) \mathbb{E}_{i_k}[(\lambda_{i_k}-z_{i_k}^k)\nabla_{i_k}\varphi(y^k)] \nonumber} \\
     &= \varphi(y^k)+\theta_k(\lambda-z^k)^T\nabla\varphi(y^k) \nonumber \nonumber \nonumber \\
     &\leq \theta_k[\varphi(y^k)+(\lambda-y^k)^T\nabla\varphi(y^k)]+(1-\theta_k)\varphi(\lambda^k)
\end{align}
where we use Eq.~\eqref{simple22} and the convexity of the dual function in the last step. For the last term in the right hand side of Eq.~\eqref{21}, by taking expectation over $i_k$, we have: 
\begin{align}\label{lemma3part2}
    \lefteqn{\mathbb{E}_{i_k} \biggr[2n^2L\theta^2_k \bigr[(\lambda_{i_k}-z_{i_k}^k)^2-(\lambda_{i_k}-z_{i_k}^{k+1})^2\bigr]\biggr] }\nonumber \\
    &= 2n^2L \theta_k^2 \mathbb{E}_{i_k} \bigr[||\lambda-z^k||^2-||\lambda-z^{k+1}||^2\bigr] 
\end{align}
where the last equality comes from the following equations:
\begin{align*}
    \lefteqn{\mathbb{E}_{i_k} \bigr[(\lambda_{i_k}-z_{i_k}^k)^2- \bigr(\lambda_{i_k}-z_{i_k}^{k+1}\bigr)^2 \bigr]} \\
    &= \mathbb{E}_{i_k} \biggr[(\lambda_{i_k}-z_{i_k}^k)^2- \biggr(\lambda_{i_k}-z_{i_k}^k+\frac{1}{2n L\theta_k}\nabla_{i_k}\varphi(y^k) \biggr)^2 \biggr]\\
    &= \frac{1}{2n}(\|\lambda-z^k\|^2-\sum_{i_k=0}^{2n} \biggr(\lambda_{i_k}-z_{i_k}^k+\frac{1}{2n L\theta_k}\nabla_{i_k}\varphi(y^k)\biggr)^2) \\
    &= \frac{1}{2n}\|\lambda-z^k\|^2-\frac{1}{2n} \|\lambda - z^k+\frac{1}{2n L\theta_k}\nabla\varphi(y^k)\|^2 \\
    &= \frac{1}{2n} \biggr[-\frac{(\lambda-z^k)}{n L\theta_k}\nabla \varphi(y^k)-\frac{1}{4 n^2L^2\theta_k^2}\|\nabla\varphi(y^k)\|^2 \biggr] \\
    &= (-2)(\lambda-z^k)\mathbb{E}_{i_k}[z^k-z^{k+1}]-  \mathbb{E}_{i_k}[||z^k-z^{k+1}||^2]
\end{align*}
where the last inequality is due to the fact that $\nabla \varphi(y^k)= 4 \mathbb{E}_{i_k}[(z^k-z^{k+1})n^2L\theta_k]$ and Jensen's inequality. Therefore, by simple algebra, we have
\begin{align*}
    \lefteqn{\mathbb{E}_{i_k} \biggr[(\lambda_{i_k}-z_{i_k}^k)^2-(\lambda_{i_k}-z_{i_k}^{k+1})^2\biggr] }\\
    &= (-2)(\lambda-z^k) \mathbb{E}_{i_k}[z^k-z^{k+1}]- \mathbb{E}_{i_k}[||z^k-z^{k+1}||^2]\\
    &= \mathbb{E}_{i_k}[\|\lambda-z^k\|^2-\|\lambda-z^{k+1}\|^2].
\end{align*}
Notice that equation (\ref{21}) holds for any value of $i_k$. Hence, by combining the results from Eq.~\eqref{lemma3part1} and Eq.~\eqref{lemma3part2} with Eq.~\eqref{21}, at each iteration with a certain value of $i_k$, we obtain that
\begin{align*}
    \lefteqn{\mathbb{E}_{i_k} [\varphi(\lambda^{k+1})]} \\
    &\leq (1-\theta_k)\varphi(\lambda^k)+\theta_k[\varphi(y^k)+(\lambda-y^k)^{\top} \nabla\varphi(y^k)] \\
    &+ 2 n^2 L \theta_k^2 \biggr(||\lambda-z^k||^2- \mathbb{E}_{i_k}[||\lambda-z^{k+1}||^2]\biggr).
\end{align*}
As a consequence, we achieve the conclusion of the lemma.

\begin{figure*}[!ht]
\begin{minipage}[b]{.5\textwidth}
\includegraphics[width=65mm,height=45mm]{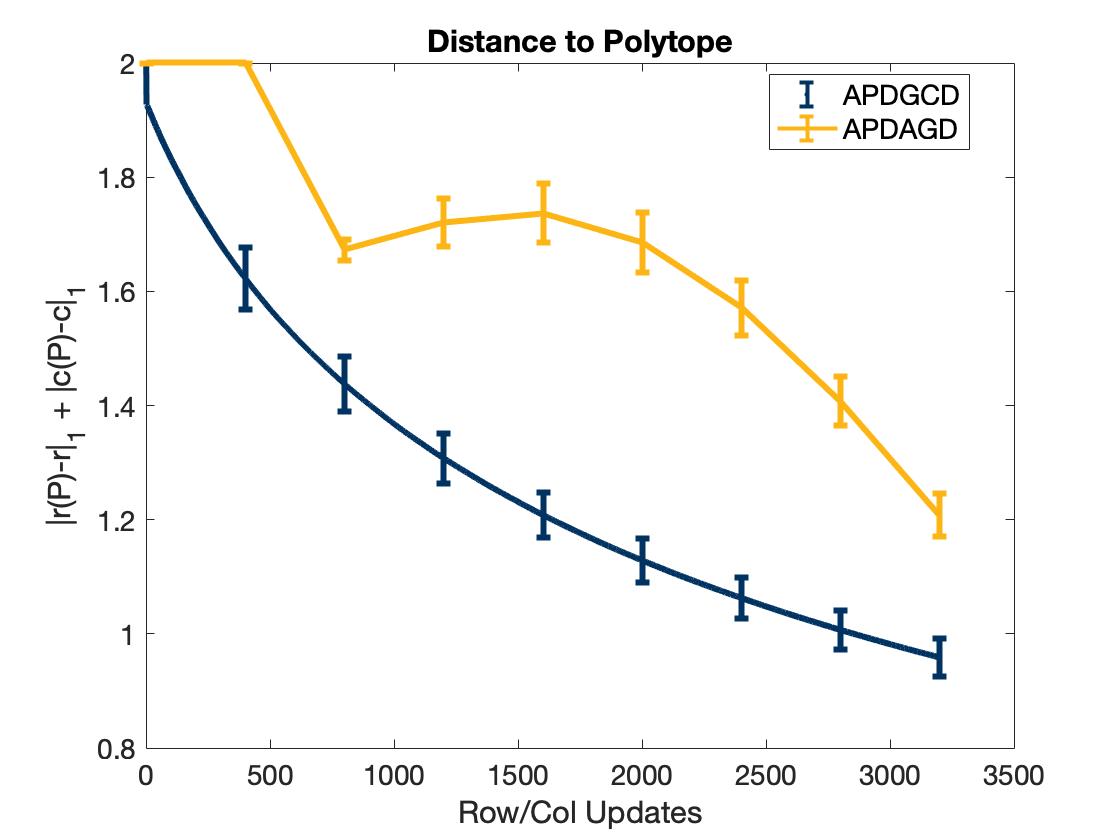}
\end{minipage}
\quad
\begin{minipage}[b]{.5\textwidth}
\includegraphics[width=65mm,height=45mm]{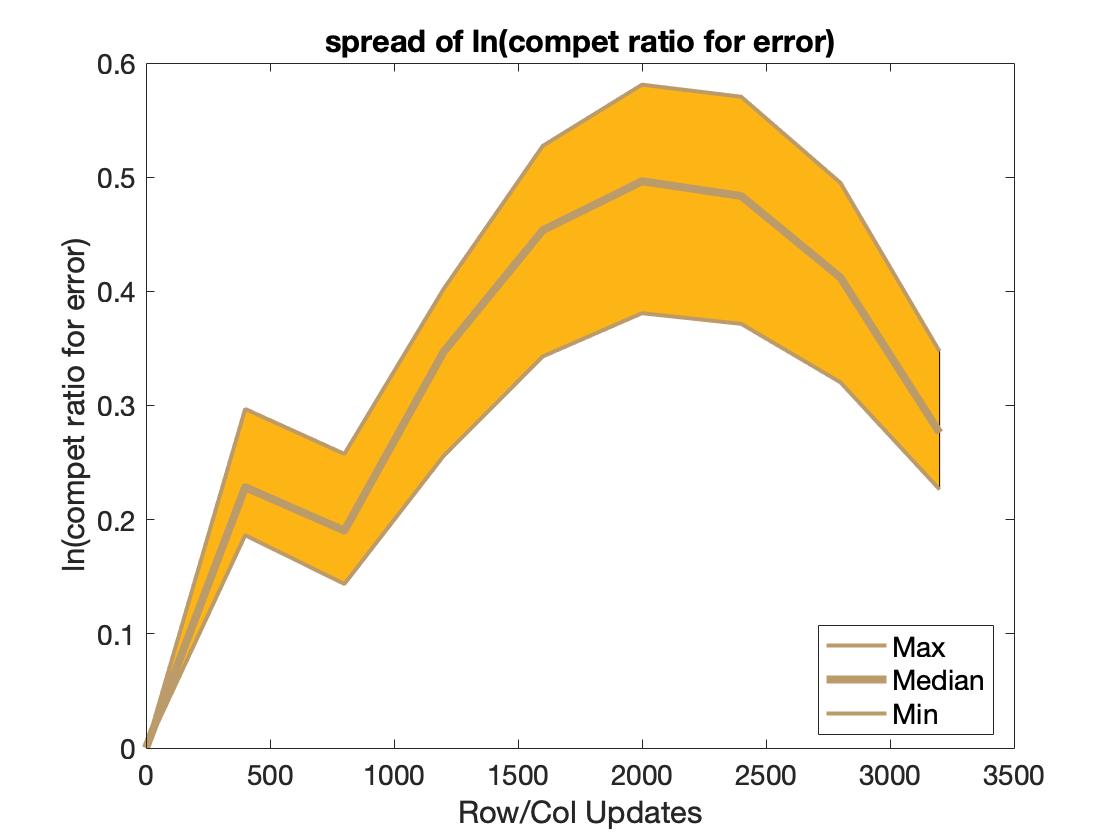}
\end{minipage}
\\
\begin{minipage}[b]{.5\textwidth}
\includegraphics[width=65mm,height=45mm]{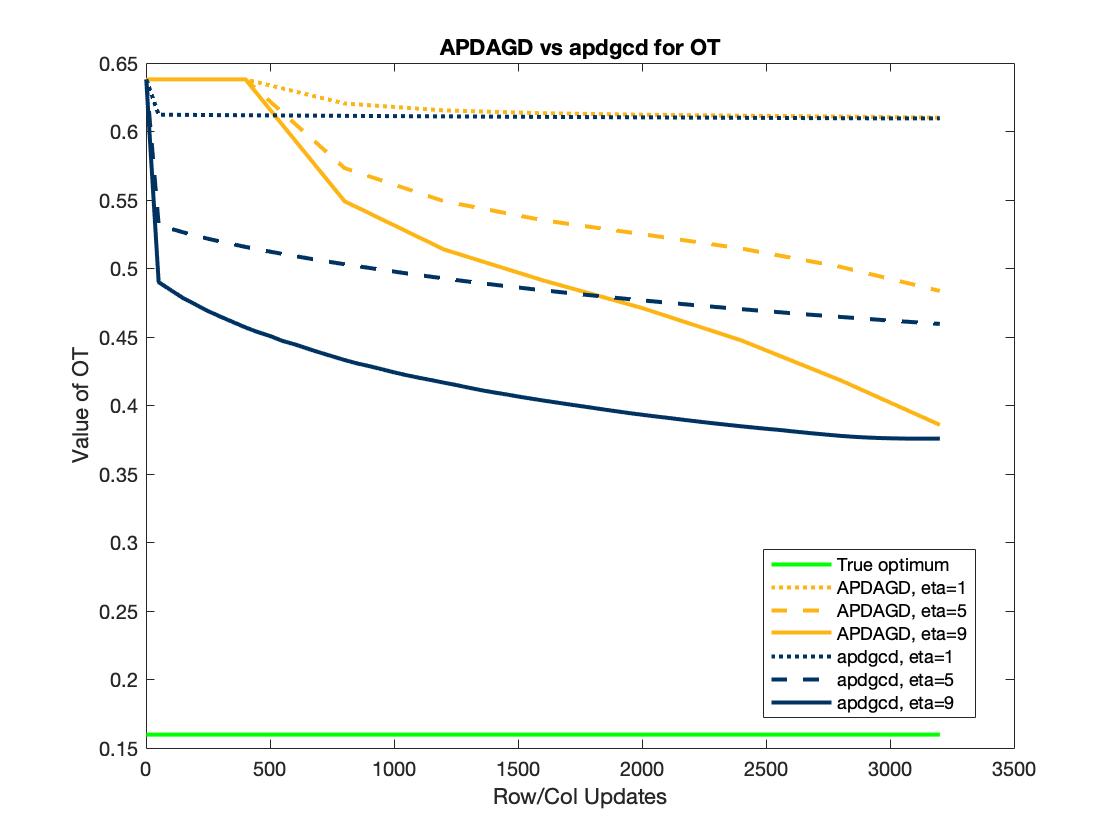}
\end{minipage}
\quad
\begin{minipage}[b]{.5\textwidth}
\includegraphics[width=65mm,height=45mm]{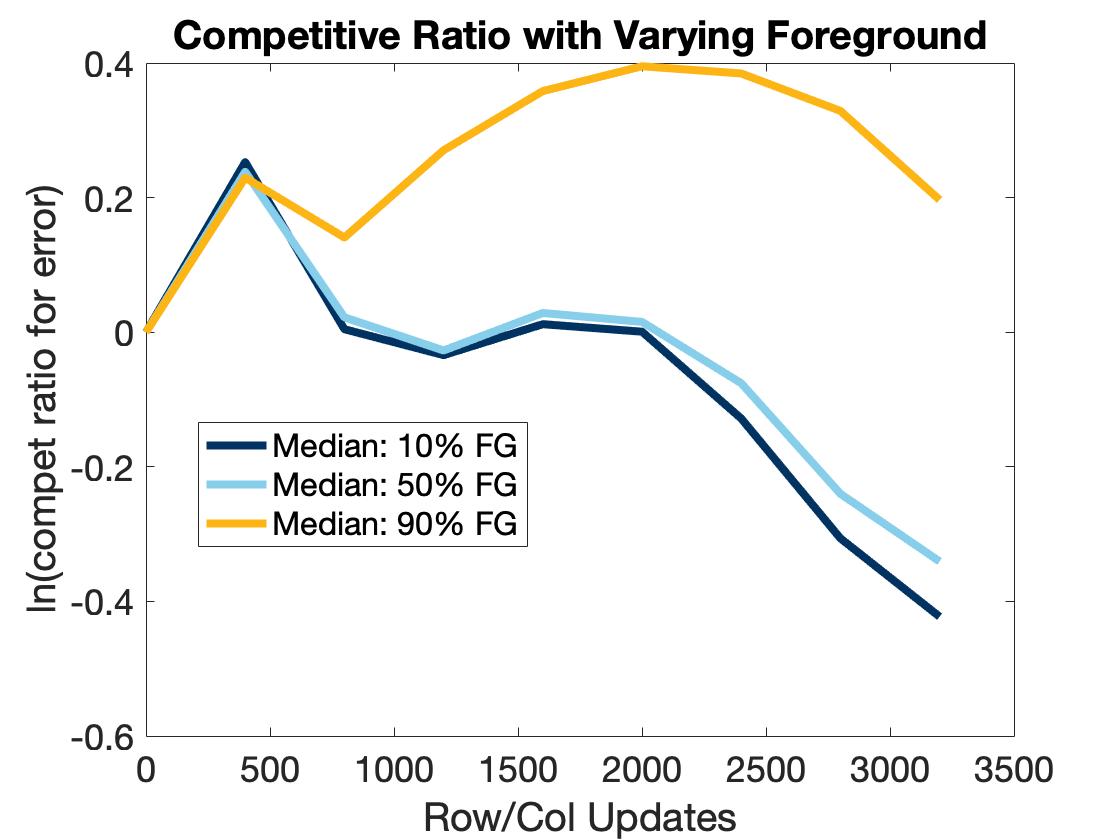}
\end{minipage}
\caption{\small Performance of APDGCD and APDAGD algorithms on the synthetic images. The organization of the images is similar to those in Figure~\ref{fig:rcd-agd-syn}.}
\label{fig:gcd-agd-syn}
\end{figure*}

\subsection{Proof of Theorem~\ref{thm: OTconvergence}}
\label{subsec:proof:thm: OTconvergence}
By the result of Lemma~\ref{lemma 3} and the definition of the sequence $\{\theta_k\}$ in Algorithm~\ref{Algorithm:APDCD}, we obtain the following bounds:
\begin{align*}
    \lefteqn{\mathbb{E} \biggr[\frac{1}{\theta_k^2}\varphi(\lambda^{k+1})\biggr] }\\
    &\leq \mathbb{E} \biggr[ \frac{1-\theta_k}={\theta_k^2}\varphi(\lambda^k)+\frac{1}{\theta_k}[\varphi(y^k)+(\lambda-y^k)^T\nabla\varphi(y^k)] \\
    &+ 2 L n^2 \big(||\lambda-z^k||^2-||\lambda-z^{k+1}||^2\big) \biggr] \\ 
    &=  \mathbb{E} \biggr[ \frac{1}{\theta_{k-1}^2}\varphi(\lambda^k)+\frac{1}{\theta_k}[\varphi(y^k)+(\lambda-y^k)^T\nabla\varphi(y^k)] \\
    &+ 2 L n^2 \big( ||\lambda-z^k||^2-||\lambda-z^{k+1}||^2 \big) \biggr]  
\end{align*}
where the outer expectations are taken with respect to the random sequence of the coordinate indices in Algorithm~\ref{Algorithm:APDCD}. Keep iterating the above bound and using the fact that $\theta_0 = 1$ and $C_{k} = 1/ \theta_{k}^2$, we arrive at the following inequalities:
\begin{align*}
    \lefteqn{C_k  \mathbb{E} \big[ \varphi(\lambda^{k+1})\big]}  \\
    &\leq \sum_{i=0}^k \dfrac{1}{\theta_{i}} \mathbb{E} \big[ \varphi(y^i)+ \left \langle \nabla \varphi(y^i),\lambda-y^i\right \rangle \big] \\
    &+ 2 L n^2 \bigr( ||\lambda - z^0||^2 - \mathbb{E} \big [||\lambda-z^{k+1}||^2 \big] \bigr)  \\
    &\leq \min \limits_{\lambda \in \mathbb{R}^{2n}} \biggr( \sum_{i=0}^k \dfrac{1}{\theta_{i}} \mathbb{E} \big[ \varphi(y^i)+ \left \langle \nabla \varphi(y^i),\lambda-y^i\right \rangle \big] \\
    & + 2 L n^2 \bigr( ||\lambda - z^0||^2 - \mathbb{E} \big [||\lambda-z^{k+1}||^2 \big] \bigr) \biggr) \\
    &\leq \min \limits_{\lambda \in B_2 (2\hat{R})} \biggr( \sum_{i=0}^k \dfrac{1}{\theta_{i}} \mathbb{E} \big[ \varphi(y^i)+ \left \langle \nabla \varphi(y^i),\lambda-y^i\right \rangle \big] \\
    & + 2 L n^2 \bigr( ||\lambda - z^0||^2 - \mathbb{E} \big [||\lambda-z^{k+1}||^2 \big] \bigr) \biggr)
\end{align*}
 where $\hat{R} : = \eta \sqrt{n}(R+\frac{1}{2})$ is the upper bound for $l_2$-norm of optimal solutions of dual regularized OT problem~\eqref{eq:dual_entropic} according to Lemma 3.2 in~\citep{lin2019efficient} and $\mathbb{B}_2(r)$ is defined as
 \begin{equation*}
     \mathbb{B}_2(r) \defeq \{\lambda\in \mathbb{R}^{2n} \ | \ ||\lambda||_2 \leq r\}.
 \end{equation*}
As $\mathbb{E} \big [||\lambda-z^{k+1}||^2 \big] \geq 0$, the inequality in the above display can be further rewritten as 
\begin{align}
    \lefteqn{C_k  \mathbb{E} \big[ \varphi(\lambda^{k+1})\big] } \nonumber \\
    &\leq \min \limits_{\lambda \in \mathbb{B}_2 (2\hat{R})} \biggr( \sum_{i=0}^k \dfrac{1}{\theta_{i}} \mathbb{E} \big[ \varphi(y^i)+ \left \langle \nabla \varphi(y^i),\lambda-y^i\right \rangle \big] \nonumber \\
    &+ 2 L n^2 ||\lambda - z^0||^2 \biggr)  \nonumber \nonumber \\
    &\leq \min \limits_{\lambda \in \mathbb{B}_2 (2\hat{R})} \biggr( \sum_{i=0}^k \dfrac{1}{\theta_{i}} \mathbb{E} \big[ \varphi(y^i)+ \left \langle \nabla \varphi(y^i),\lambda-y^i\right \rangle \big] \nonumber \\
    &+ 8 L n^2 \hat{R}^2 \biggr) \label{eq:iter_count_first}
\end{align}
where the last inequality is due to $z^{0} = 0$. Furthermore, by the definition of the dual entropic regularized OT objective function $\varphi(\lambda)$, we can verify the following equations:
\begin{align*}
    \lefteqn{\varphi(y^i) + \left \langle \nabla \varphi(y^i), \lambda-y^i \right \rangle } \\
    &= \left \langle y^i,b-Ax(y^i)\right \rangle -f(x(y^i))+\left \langle \lambda-y^i,b-Ax(y^i)
\right \rangle \\
&= -f(x(y^i))+\left \langle \lambda, b-Ax(y^i) \right \rangle
\end{align*}
where $f(x) : = \left\langle C, x\right\rangle$, $x(\lambda) : = \mathop{\arg \max}\limits_{x \in \br^{n \times n}} \biggr\{ - f(x) - \left\langle A^{\top}\lambda, x\right\rangle \biggr\}$, and $b = \begin{pmatrix} r \\ l \end{pmatrix} $. The above equation leads to the following inequality:
\begin{align}
 \lefteqn{\sum_{i=0}^k \dfrac{1}{\theta_{i}} \mathbb{E} \big[ \varphi(y^i)+ \left \langle \nabla \varphi(y^i),\lambda-y^i\right \rangle \big]}  \nonumber \\
 &= \sum_{i=0}^k \dfrac{1}{\theta_{i}} \mathbb{E} \big[-f(x(y^i))+\left \langle \lambda, b-Ax(y^i) \right \rangle \big]  \nonumber \\
 &\leq - C_k f(\mathbb{E} \big[x^k\big]) + \sum_{i=0}^k \dfrac{1}{\theta_{i}} \left \langle \lambda,b-A \mathbb{E} \big[x(y^i) \big] \right \rangle  \nonumber \\
 &=  \ C_{k} \big(- f(\mathbb{E} \big[x^k\big]) + \left \langle \lambda,b-A \mathbb{E} \big[x^{k} \big] \right \rangle \big) \label{eq:iter_count_second}
\end{align}
where the second inequality is due to the convexity of $f$. Combining the results from~\eqref{eq:iter_count_first} and~\eqref{eq:iter_count_second}, we achieve the following bound 
\begin{align*}
    \lefteqn{C_k  \mathbb{E} \big[ \varphi(\lambda^{k+1})\big]} \\
 &\leq  - C_k f(\mathbb{E} \big[x^k\big]) + \min \limits_{\lambda \in \mathbb{B}_2(2\hat{R})} \{  C_k \left \langle \lambda,b-A \mathbb{E} \big[x^{k} \big] \right \rangle \} \\
 &+ 8Ln^2\hat{R}^2 \\
    &\leq - C_k f(\mathbb{E} \big[x^k\big]) + 8Ln^2\hat{R}^2 - 2 C_k \hat{R} \mathbb{E} \big[ ||Ax^k-b||_2 \big].
\end{align*}
which is equivalent to
\begin{align}\label{eq:iter_count_third}
&f(\mathbb{E}\big[x^k\big]) + \mathbb{E} \big[\varphi(\lambda^{k+1})\big] + 2\hat{R} \mathbb{E} \big[||Ax^k-b||_2\big] \nonumber\\
&\leq \frac{8 L n^2\hat{R}^2}{C_k}.
\end{align}

\begin{figure*}[!ht]
\begin{minipage}[b]{.5\textwidth}
\includegraphics[width=65mm,height=45mm]{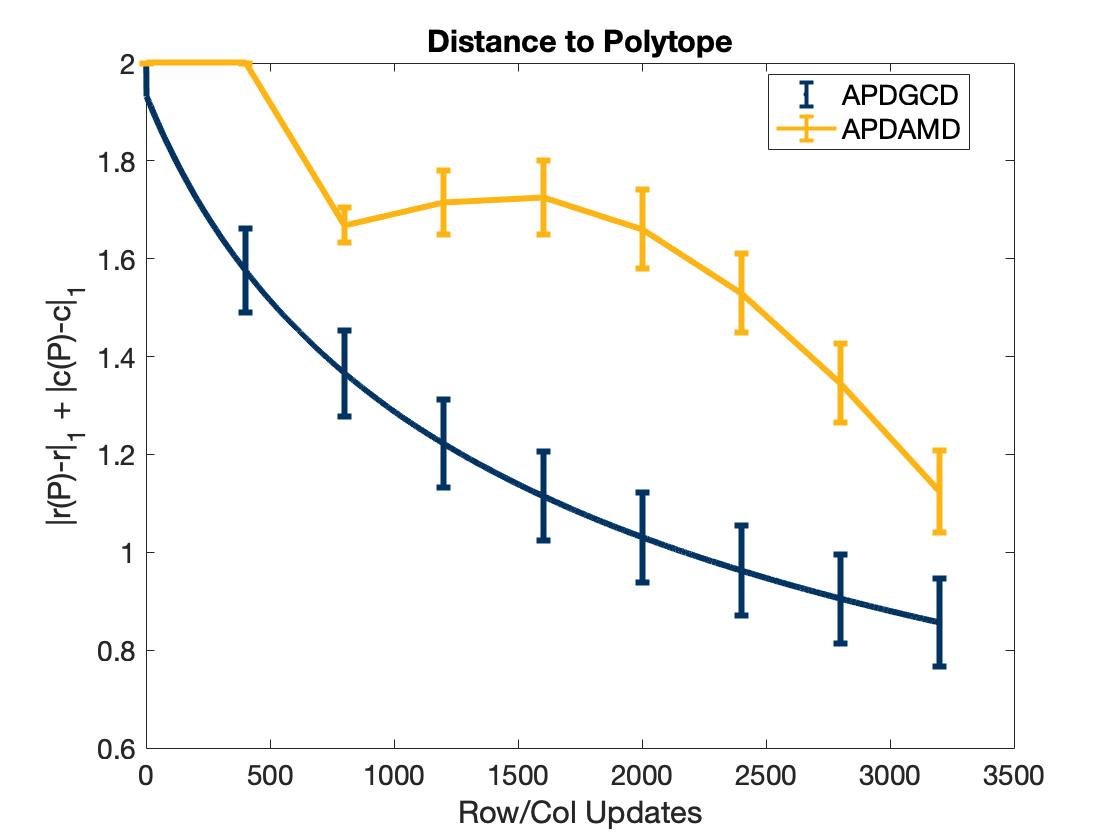}
\end{minipage}
\quad
\begin{minipage}[b]{.5\textwidth}
\includegraphics[width=65mm,height=45mm]{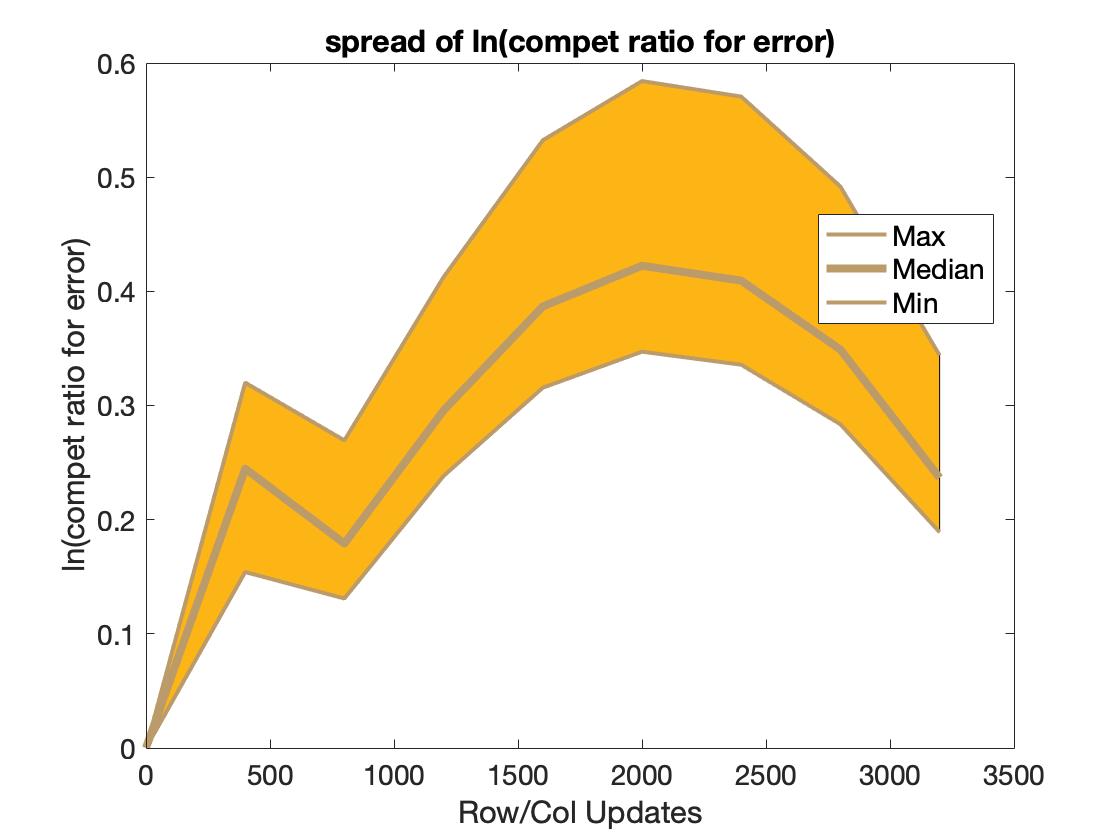}
\end{minipage}
\\
\begin{minipage}[b]{.5\textwidth}
\includegraphics[width=65mm,height=45mm]{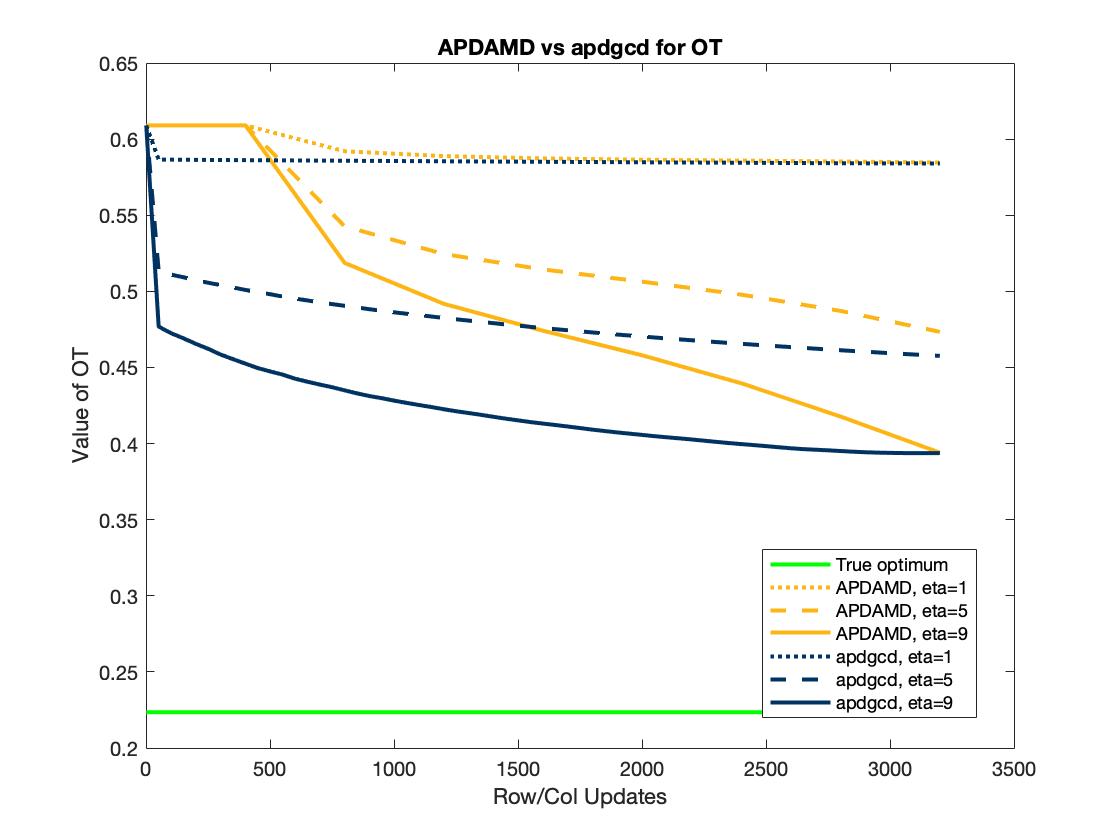}
\end{minipage}
\quad
\begin{minipage}[b]{.5\textwidth}
\includegraphics[width=65mm,height=45mm]{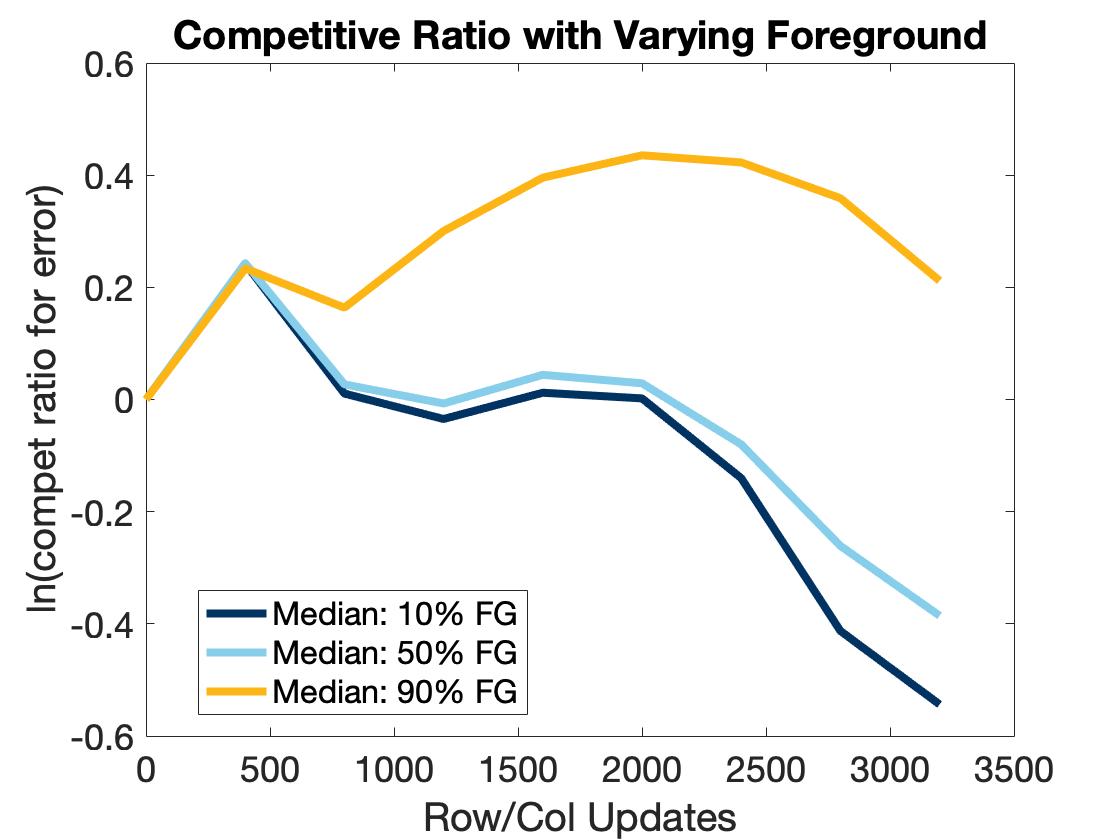}
\end{minipage}
\caption{\small Performance of APDGCD and APDAMD algorithms on the synthetic images. The organization of the images is similar to those in Figure~\ref{fig:gcd-amd-syn}.}
\label{fig:gcd-amd-syn}
\end{figure*}

Denoting $\lambda^*$ as the optimal solution for the dual entropic regularized OT problem~\eqref{eq:dual_entropic}. Then, we can verify the following inequalities
\begin{align}
    \lefteqn{f(\mathbb{E} \big[x^k\big]) + \mathbb{E} \big[ \varphi(\lambda^{k+1})\big]} \nonumber \\
    &\geq f(\mathbb{E} \big[x^k\big]) + \varphi (\lambda^*) \nonumber \\
    &= f(\mathbb{E} \big[x^k\big]) + \left \langle \lambda^*,b \right \rangle + \max \limits_{x \in \mathbb{R}^{n \times n}}\{-f(x)-\left \langle A^{\top}\lambda^*,x \right \rangle\} \nonumber \\
    &\geq f(\mathbb{E} \big[x^k\big]) + \left \langle \lambda^*,b \right \rangle - f(\mathbb{E} \big[x^k\big]) - \left \langle \lambda^*,A \mathbb{E} \big[x^k\big] \right \rangle \nonumber \\
    &= \left \langle \lambda^*,b - A \mathbb{E} \big[x^k\big]\right \rangle \nonumber \\
    &\geq - \hat{R} \mathbb{E} \big[ ||Ax^k-b||_2 \big] \label{eq:iter_count_fourth}
\end{align}
where the last inequality comes from H\"{o}lder inequality and the fact that $||\lambda^{*}||_2 \leq \hat{R}$. Plugging the inequality in~\eqref{eq:iter_count_fourth} to the inequality in~\eqref{eq:iter_count_third} leads to the following bound: 
\begin{align}
    \lefteqn{\mathbb{E} \big[ ||Ax^k-b||_2 \big]} \\
    &\leq  \frac{8 L n^2 \hat{R}}{C_k} \nonumber \\
    &= \frac{8 ||A||_1^2n^{\frac{5}{2}} (R+1/2)}{C_k} \nonumber \\
    &=  \frac{32 n^{\frac{5}{2}} (R+1/2)}{C_k}. \label{eq:iter_count_fifth}
\end{align}

Therefore,  $\mathbb{E} \big[ ||Ax^k-b||_1 \big] \leq \frac{32 n^3 (R+1/2)}{C_k}$. It remains to bound $C_k$. We will use induction to show that $\theta_k \leq \frac{2}{k+2}$ for all $k \geq 0$. The inequality clearly holds for $k = 0$ as $\theta_{0} = 1$. Suppose that the hypothesis holds for $k \geq 0$, namely, $\theta_k \leq \frac{2}{k+2}$ for $k \geq 0$. By the definition of $\theta_{k+1}$ and simple algebra, we obtain that
\begin{align*}
    \theta_{k+1}=\frac{\theta_k^2}{2} \biggr( \sqrt{1+\frac{4}{\theta_k^2}} - 1 \biggr) \leq \frac{2}{k+3}
\end{align*}
where the above inequality is due to $\theta_k \leq \frac{2}{k+2}$. Therefore, we achieve the conclusion of the hypothesis for $k + 1$. Now, simple algebra demonstrates that $C_k \geq \frac{1}{4}(k+1)(k+4) \geq \frac{1}{4}(k+1)^2$. Combining this lower bound of $C_{k}$ and the inequality in~\eqref{eq:iter_count_fifth} leads to the following result: 
\begin{equation*}
    \mathbb{E} \big[ ||Ax^k-b||_1 \big] \leq \frac{128 n^3 (R+1/2)}{(k+1)^2}.
\end{equation*}
As a consequence, we conclude the desired bound on the number of iterations $k$ required to satisfy the bound $\mathbb{E} \big[||A\text{vec}(X^k)-b||_1\big]\leq \varepsilon^\prime$.

\begin{figure*}[!ht]
\begin{minipage}[b]{.5\textwidth}
\includegraphics[width=65mm,height=45mm]{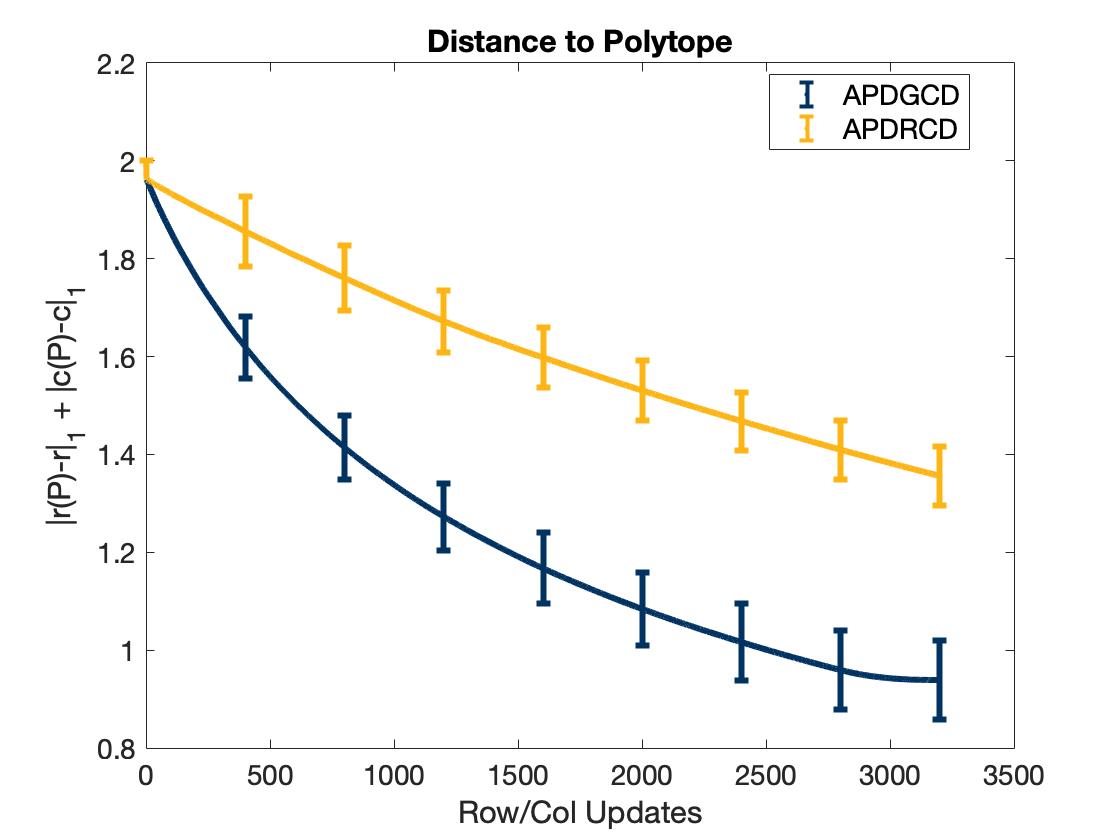}
\end{minipage}
\quad
\begin{minipage}[b]{.5\textwidth}
\includegraphics[width=65mm,height=45mm]{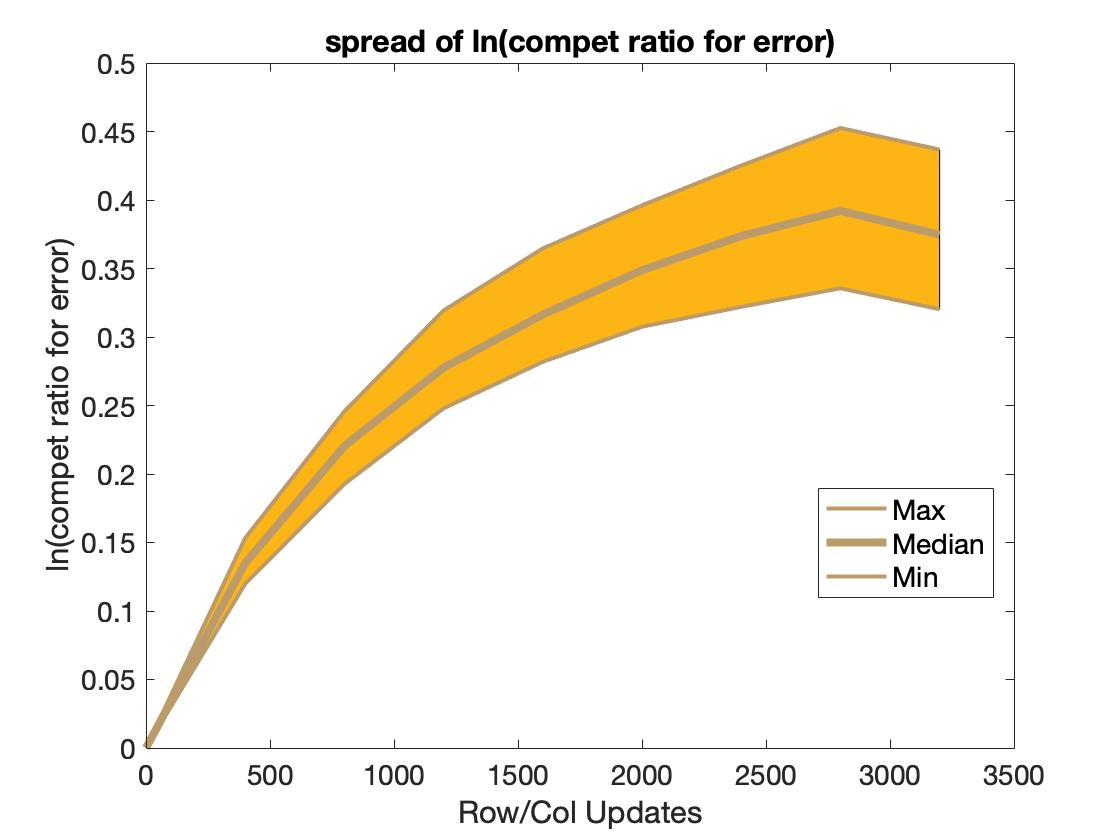}
\end{minipage}
\\
\begin{minipage}[b]{.5\textwidth}
\includegraphics[width=65mm,height=45mm]{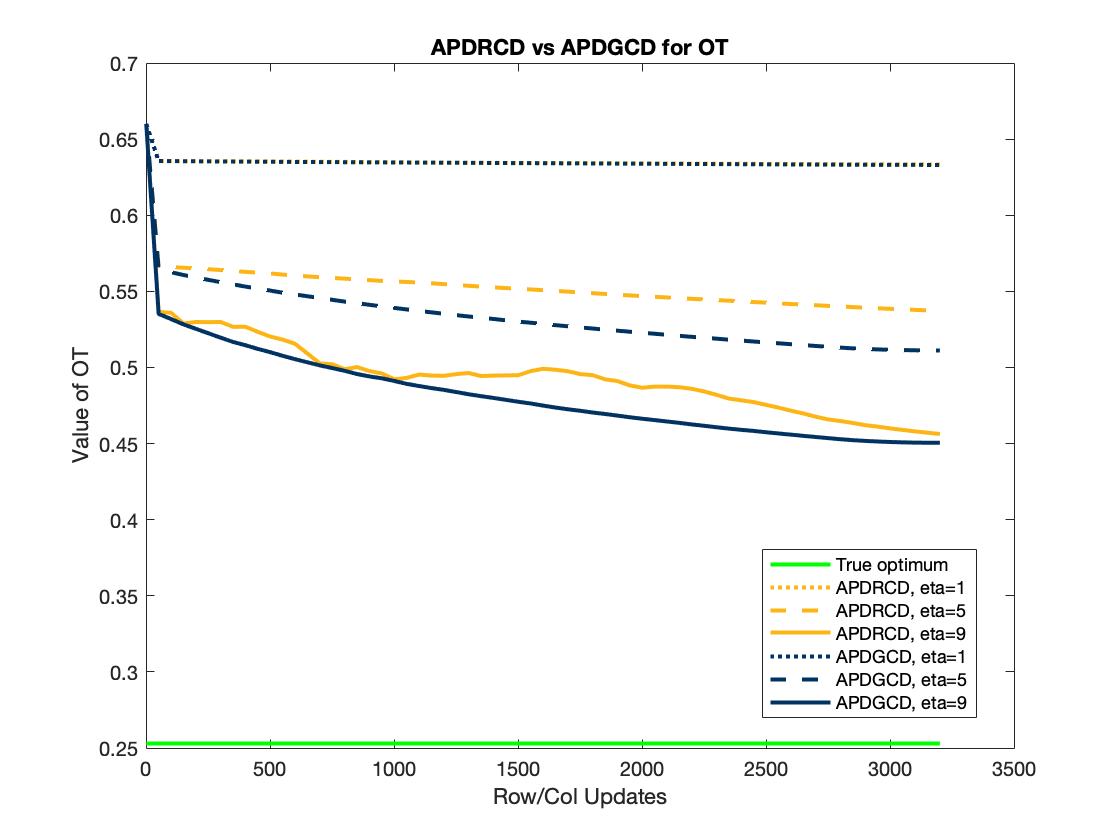}
\end{minipage}
\quad
\begin{minipage}[b]{.5\textwidth}
\includegraphics[width=65mm,height=45mm]{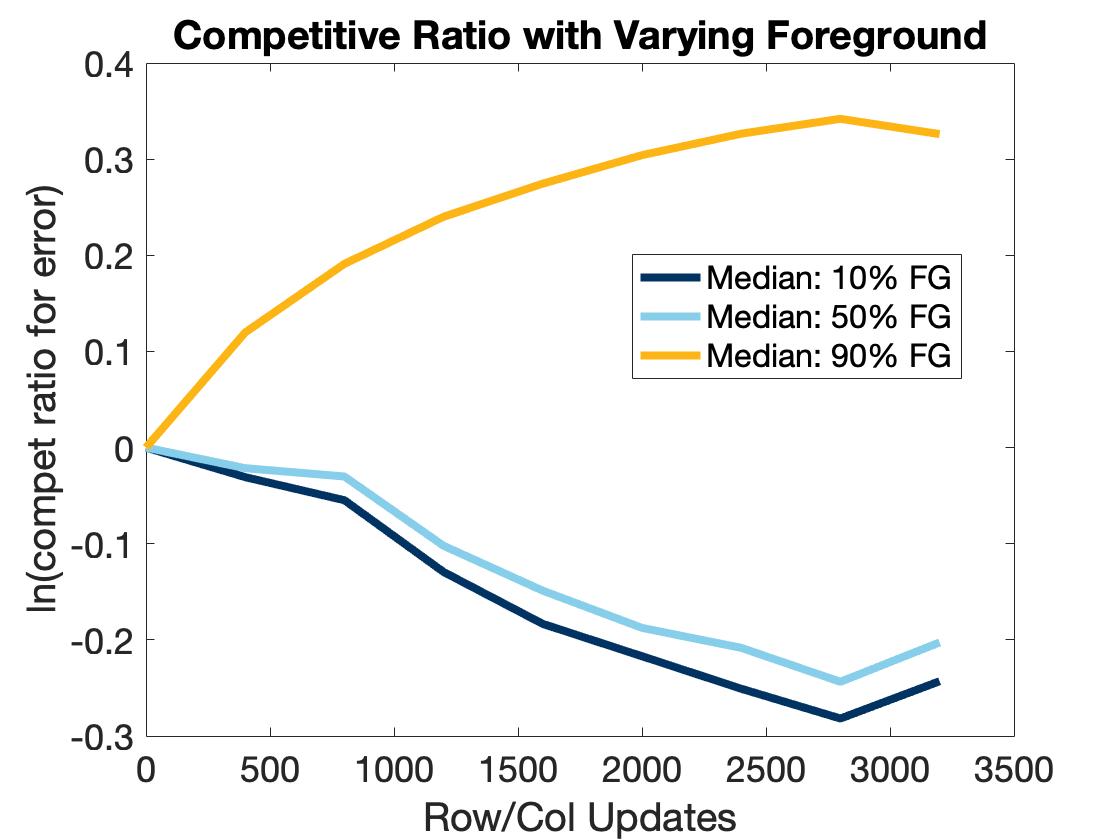}
\end{minipage}
\caption{\small Performance of APDGCD and APDRCD algorithms on the synthetic images. The organization of the images is similar to those in Figure~\ref{fig:rcd-agd-syn}.}
\label{fig:gcd-rcd-syn}
\end{figure*}

\subsection{Proof of Theorem~\ref{theorem:complex_APDRCD}}
\label{subsec:proof:theorem:complex_APDRCD}
The proof of the theorem follows the same steps as those in the proof of Theorem 1 in~\citep{Altschuler-2017-Near}. Here, we provide the detailed proof for the completeness. In particular, we denote $\tilde{X}$ the matrix returned by the APDRCD algorithm (Algorithm~\ref{Algorithm:APDCD}) with $\tilde{r}$, $\tilde{l}$ and $\varepsilon'/2$. Recall that, $X^*$ is a solution to the OT problem. Then, we obtain the following inequalities:
\begin{align*}
\lefteqn{\mathbb{E}[\langle C, \hat{X}\rangle] - \left\langle C, X^*\right\rangle}  \\
&\leq 2\eta\log(n) + 4\mathbb{E}\Big[ ||\tilde{X}\one 
- r||_1 + ||\tilde{X}^\top\one - l||_1 \Big]||C||_\infty \\
&\leq \frac{\varepsilon}{2}  + 4\mathbb{E}\Big[ ||\tilde{X}\one 
- r||_1 + ||\tilde{X}^\top\one - l||_1 \Big]||C||_\infty, 
\end{align*}
where the last inequality in the above display holds since $\eta = \frac{\varepsilon}{4\log(n)}$. Furthermore, we have
\begin{align*}
\lefteqn{\mathbb{E}\Big[ ||\tilde{X}\one 
- r||_1 + ||\tilde{X}^\top\one - l||_1 \Big]}  \\
& \leq \mathbb{E}\Big[\left\|\tilde{X}\one - \tilde{r}\right\|_1 + \left\|\tilde{X}^\top\one - \tilde{l}\right\|_1\Big] + \left\| r - \tilde{r}\right\|_1 + \left\| l - \tilde{l}\right\|_1 \\
& \leq \frac{\varepsilon'}{2} + \frac{\varepsilon'}{2} \ = \ \varepsilon'.
\end{align*}
Since $\varepsilon'=\dfrac{\varepsilon}{8\left\|C\right\|_\infty}$, the above inequalities demonstrate $\mathbb{E}[\langle C, \hat{X}\rangle] - \left\langle C, X^*\right\rangle \leq \varepsilon$. Hence, we only need to bound the complexity. Following the approximation scheme in Step 1 of Algorithm~\ref{Algorithm:ApproxOT_APDCD}, we achieve the following bound
\begin{align*}
   \lefteqn{R  = \frac{||C||_\infty}{\eta}+ \log(n) - 2 \log(\min \limits_{1\leq i,j\leq n} \{\widetilde{r}_i,\widetilde{l}_i\})} \\
     & \leq \frac{4\left\|C\right\|_\infty\log(n)}{\varepsilon} + \log(n) - 2\log\left(\frac{\varepsilon}{64n\left\|C\right\|_\infty}\right).
\end{align*}
Given the above bound with $R$, we have the following bound with the iteration count:
\begin{align*}
     \lefteqn{k \leq   1+ 12 n^{\frac{3}{2}}\sqrt{\frac{R+1/2}{\varepsilon^\prime}} } \\
      &\leq    1 + 12 n^{\frac{3}{2}}\biggr\{\frac{8 ||C||_\infty}{\varepsilon} \biggr(\frac{4||C||_\infty \log (n)}{\varepsilon}+\log(n) \\
     &- 2\log \big(\frac{\varepsilon}{64n||C||_\infty}\big)+\frac{1}{2}\biggr)\biggr\}^{\frac{1}{2}} \\
     &=  \mathcal{O} \biggr(\frac{n^{\frac{3}{2}} ||C||_\infty \sqrt{\log(n)}}{\varepsilon} \biggr).
\end{align*}
Combining the above result with the fact that each iteration the APDRCD algorithm requires $\mathcal{O}(n)$ arithmetic operations, we conclude that the total number of arithmetic operations required for the APDRCD algorithm for approximating optimal transport is $\mathcal{O} \biggr(\frac{n^{\frac{5}{2}} ||C||_\infty\sqrt{\log(n)}}{\varepsilon}\biggr)$. Furthermore, the column $\tilde{r}$ and row $\tilde{l}$ in Step 2 of Algorithm~\ref{Algorithm:ApproxOT_APDCD} can be found in $\bigO(n)$ arithmetic operations while Algorithm~2 in \citep{Altschuler-2017-Near} requires $\bigO(n^2)$ arithmetic operations. As a consequence, we conclude that the total number of arithmetic operations is $\mathcal{O} \biggr(\frac{n^{\frac{5}{2}} ||C||_\infty\sqrt{\log(n)}}{\varepsilon}\biggr)$. 

Note that by Markov inequality, $$P(\langle C, \hat{X} \rangle > a) \leq \frac{\mathbb{E} [\langle C, \hat{X} \rangle ]}{a}$$ for $a \geq 0$. Combining with theorem~\ref{theorem:complex_APDRCD} gives us a high probability bound for obtaining an $\epsilon$-optimal solution.

\begin{figure*}[!ht]
\begin{minipage}[b]{.3\textwidth}
\includegraphics[width=50mm,height=35mm]{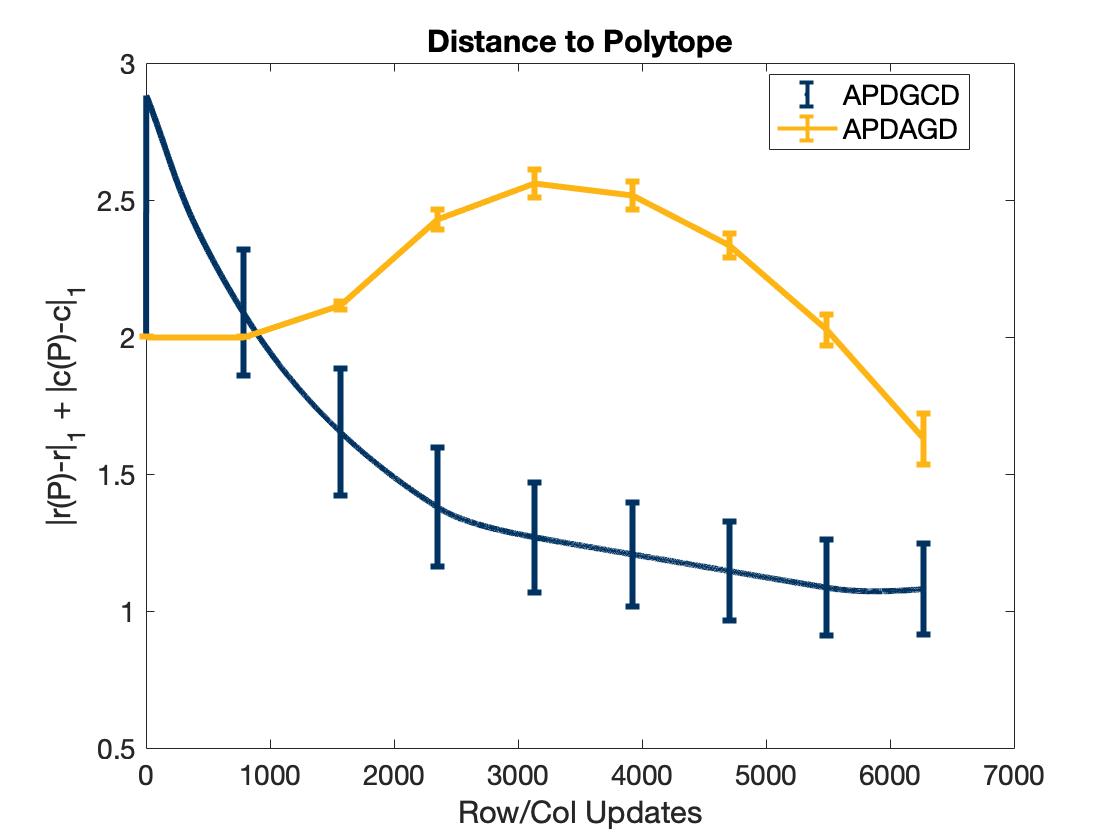}
\end{minipage}
\quad
\begin{minipage}[b]{.3\textwidth}
\includegraphics[width=50mm,height=35mm]{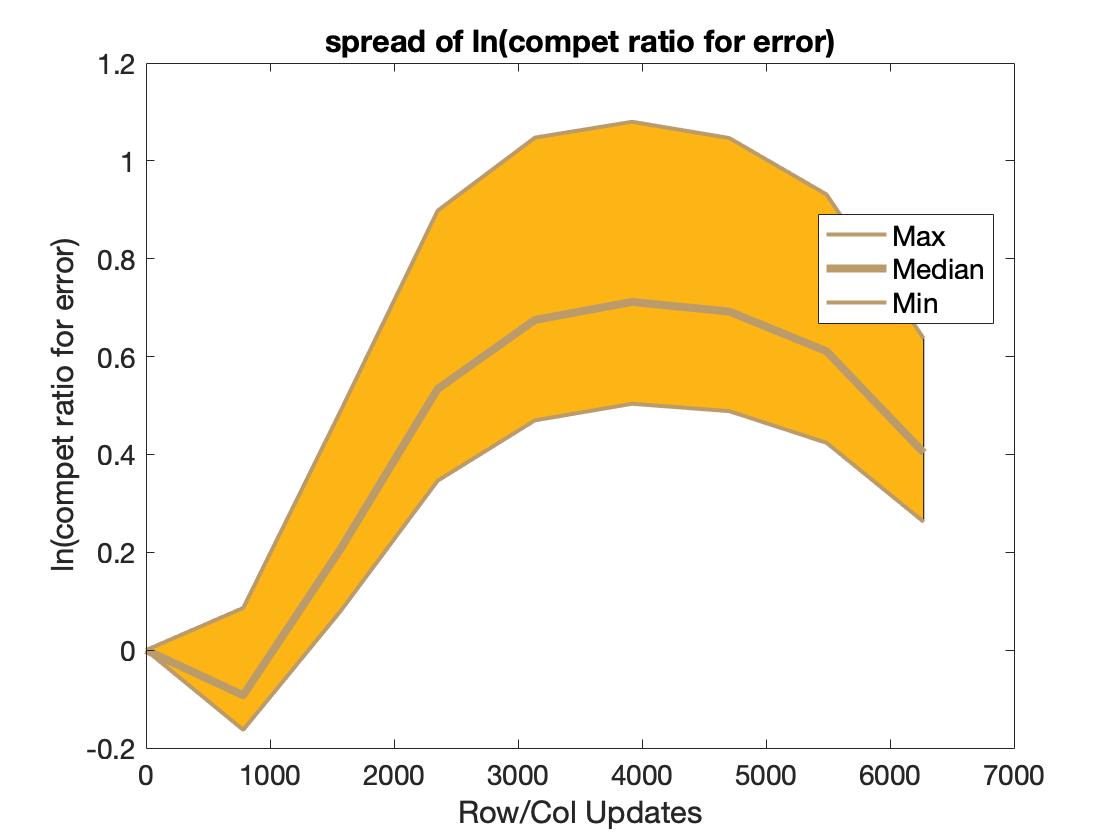}
\end{minipage}
\quad
\begin{minipage}[b]{.3\textwidth}
\includegraphics[width=50mm,height=35mm]{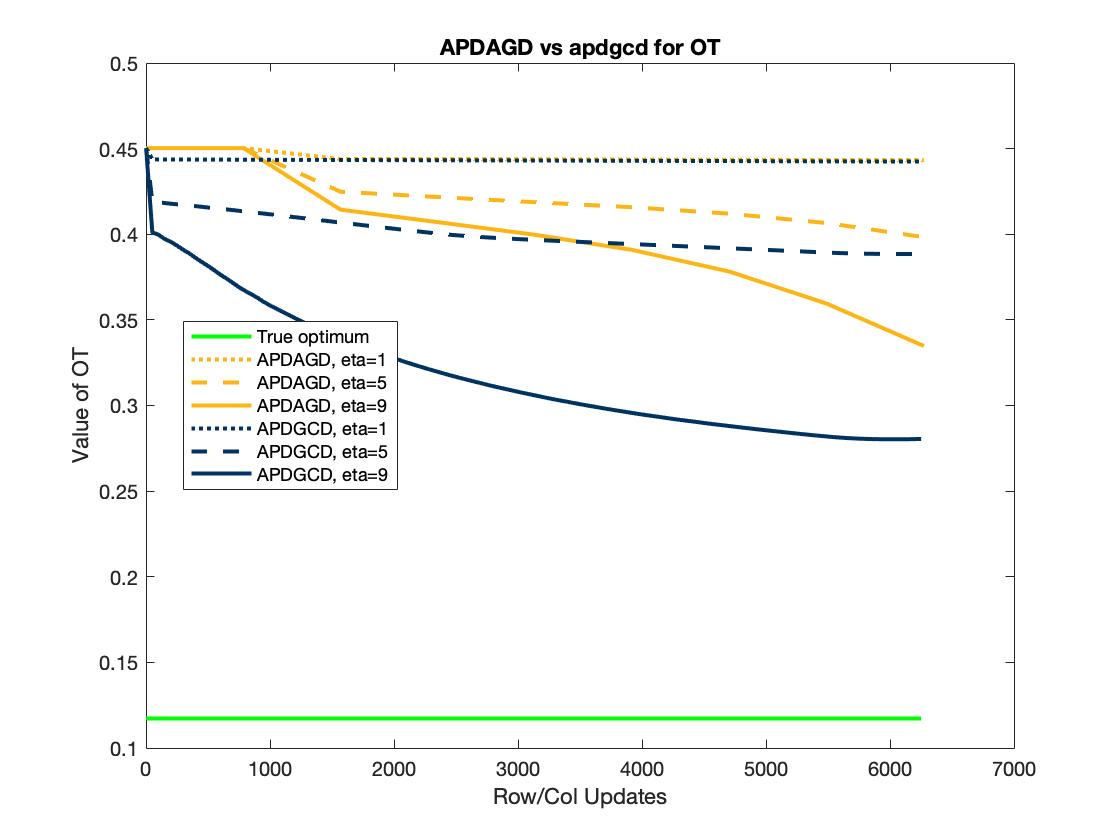}
\end{minipage}
\\
\begin{minipage}[b]{.3\textwidth}
\includegraphics[width=50mm,height=35mm]{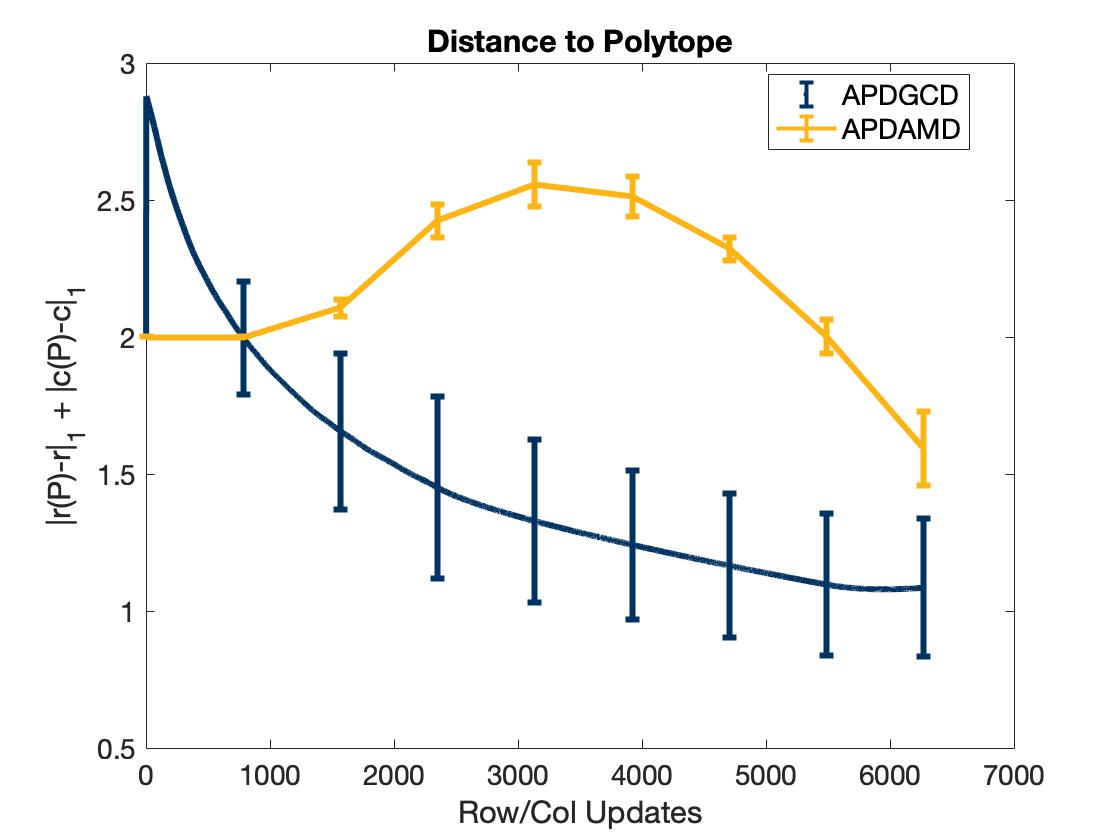}
\end{minipage}
\quad
\begin{minipage}[b]{.3\textwidth}
\includegraphics[width=50mm,height=35mm]{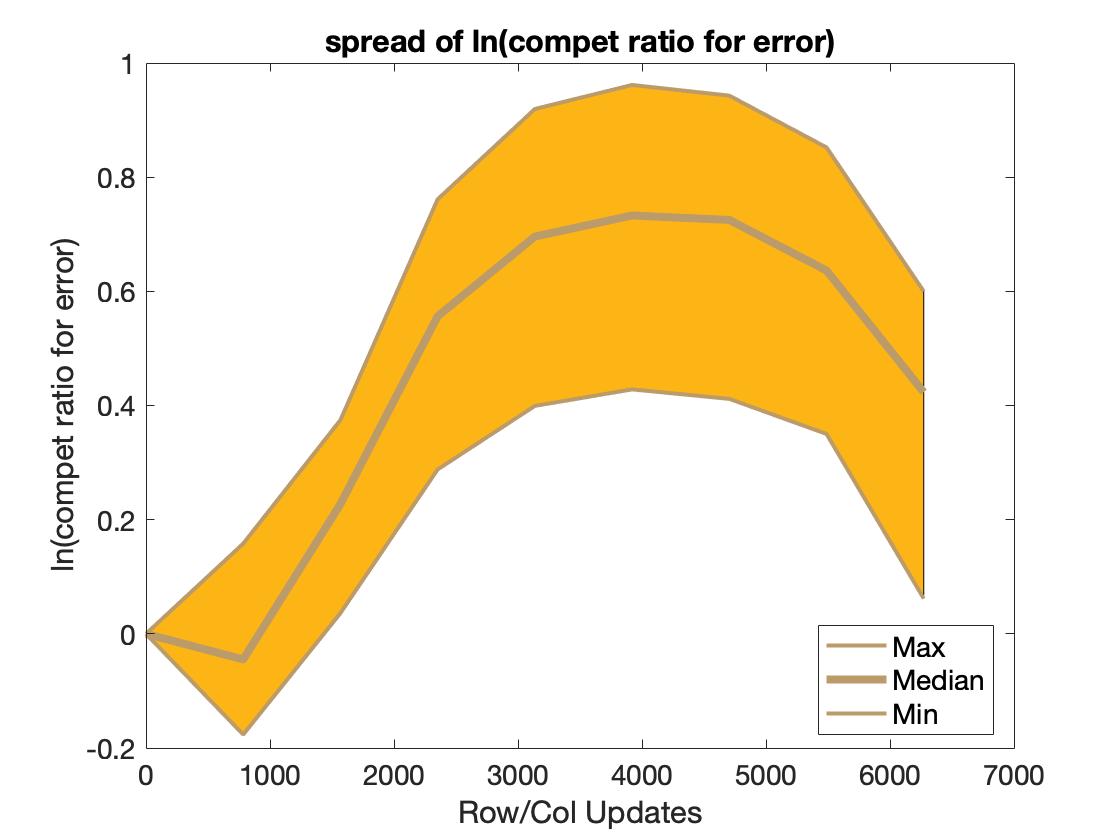}
\end{minipage}
\quad
\begin{minipage}[b]{.3\textwidth}
\includegraphics[width=50mm,height=35mm]{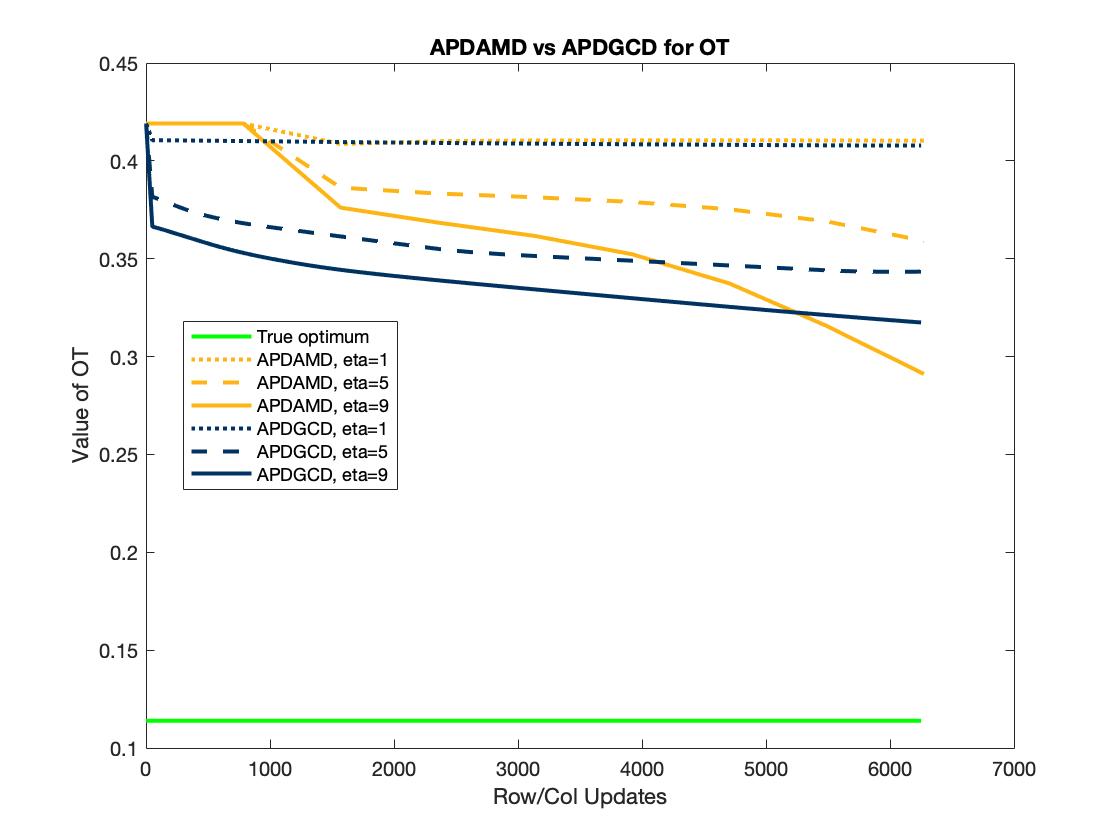}
\end{minipage}
\\
\begin{minipage}[b]{.3\textwidth}
\includegraphics[width=50mm,height=35mm]{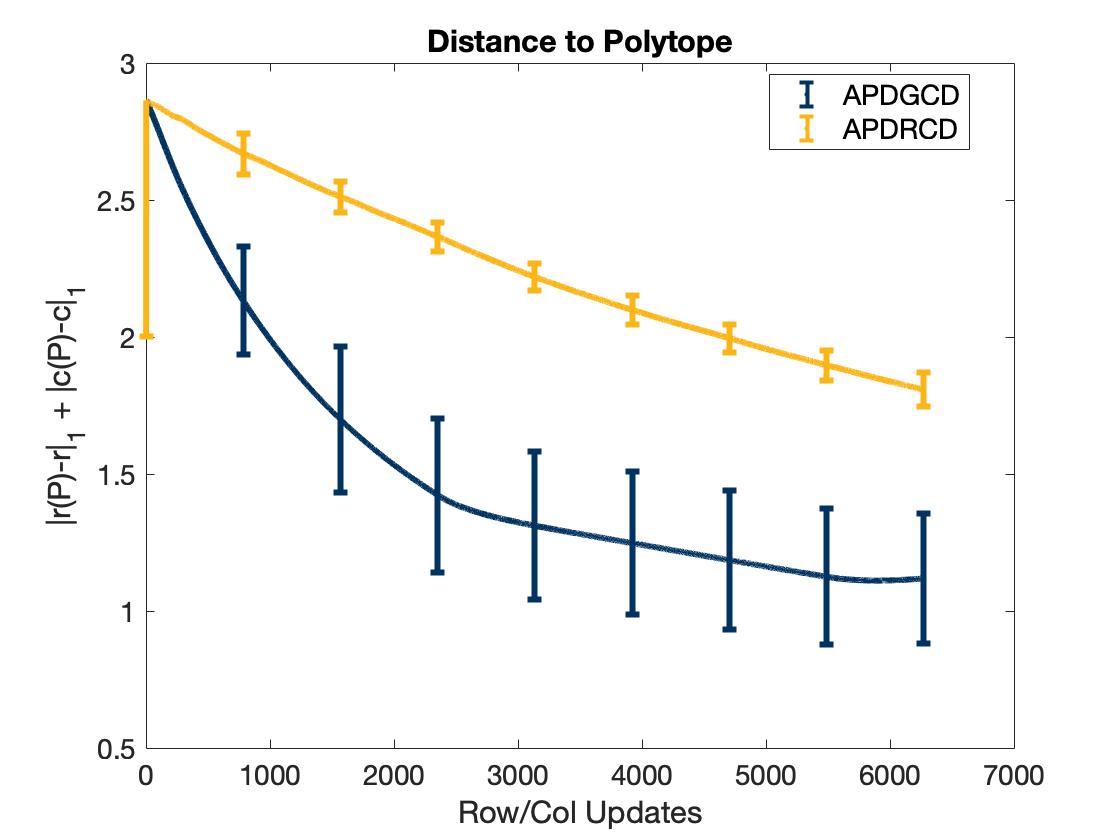}
\end{minipage}
\quad
\begin{minipage}[b]{.3\textwidth}
\includegraphics[width=50mm,height=35mm]{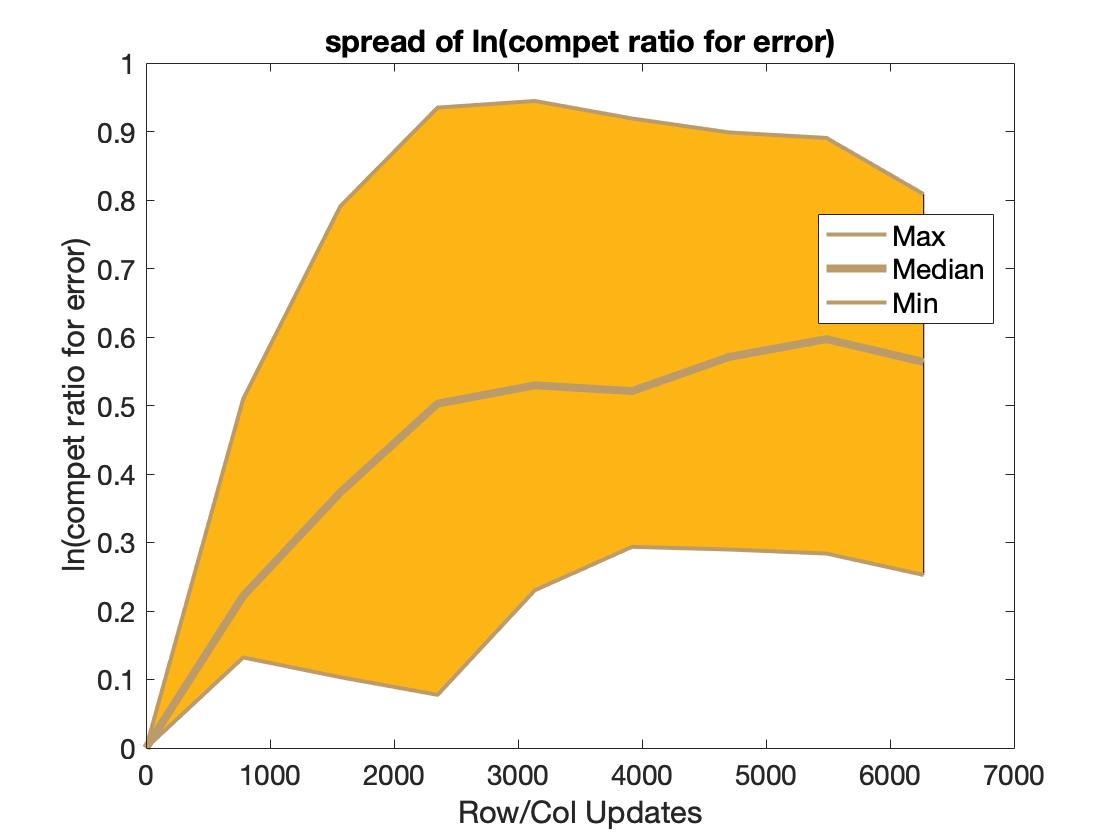}
\end{minipage}
\quad
\begin{minipage}[b]{.3\textwidth}
\includegraphics[width=50mm,height=35mm]{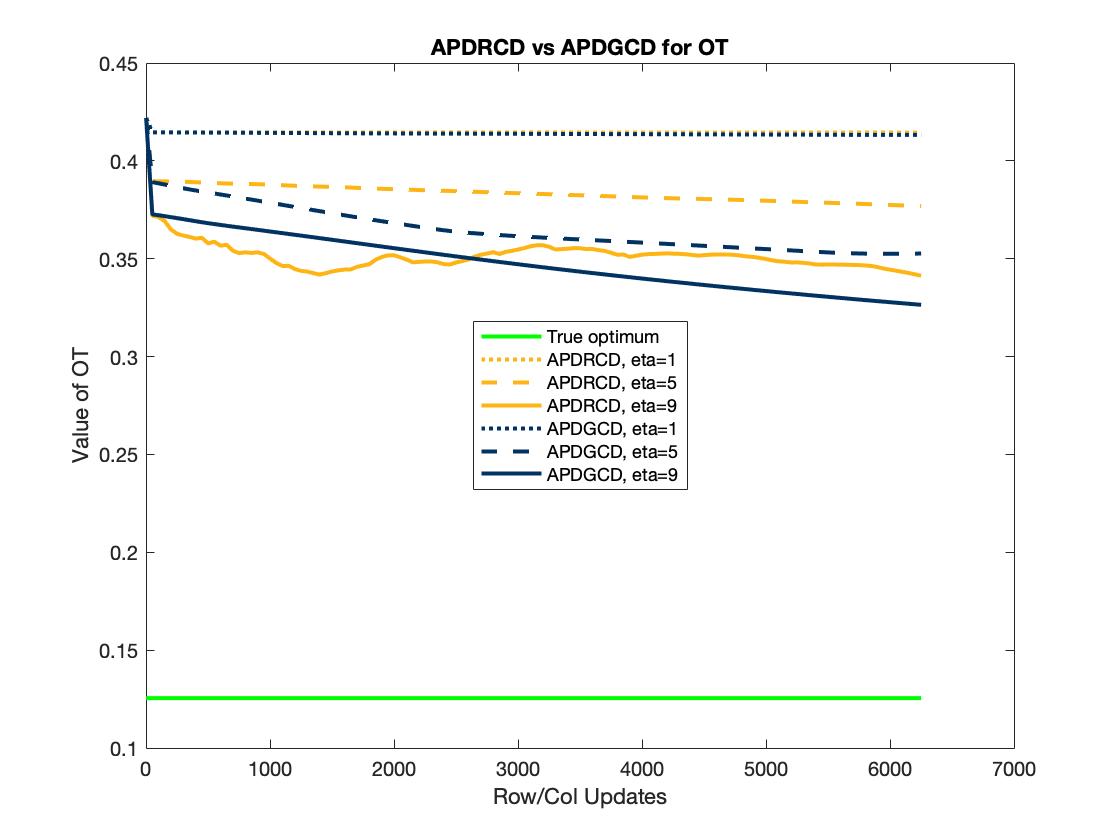}
\end{minipage}
\caption{\small Performance of the APDGCD, APDAGD, APDAMD, and APDRCD algorithms on
the MNIST real images. The organization of the images is similar to those in Figure~\ref{fig:rcd-agd-amd-mnist}.}
\label{fig:gcd-agd-amd-rcd-mnist}
\end{figure*}

\section{Approximating Wasserstein Barycenter with the APDRCD algorithm }
\label{Sec:WB}

In this section, we introduce the distributed Wasserstein barycenter problem and present the adapted Accelerated Primal-Dual Coordinate Descent algorithm to approximate the Wasserstein barycenter efficiently for a family of probability measures. We first introduce the setup of the distributed Wasserstein barycenter problem and its entropic regularization. Then we construct the dual form of the problem. We further generalize the APDRCD and the APDGCD algorithms to compute the Wasserstein barycenter and demonstrate its flexibility for computations in the decentralized distributed setting. 

\subsection{Distributed Wasserstein Barycenter Problem}

Given a network of probability measures, the optimal transport distance naturally defines the mean representative of the given set of measures. The Wasserstein barycenter problem consider finding the probability measure that is closest to all the given measures in terms of regularized Wasserstein distance. Wasserstein barycenters captur the structure of the given set of objects in a geometrically faithful way. For simplicity, we present the Wasserstein barycenter problem for $m$ discrete measures or histograms with entropic regularizations, but the algorithm could be easily generalized to any continous measures by drawing samples from the given measures. 

As introduced in Eq~\ref{prob:OT}, given two probability measures $r, l \in \Delta^n$, we define the regularized Wasserstein distance between $r$ and $l$ as: 
\begin{align}
\mathcal{W}(r,l) := \min_{\pi \in \Pi(r, l)} \left \langle \pi,C \right \rangle
\end{align}
where $\Pi(r,l) = \{\pi \in \br_+^{n\times n}: \pi \one = r, \pi^T\one=l\}$ is the set of all coupling measures between the measure $r$ and $l$.  Using the entropic regularization as in Eq~\ref{prob:regOT}, the regularized OT distance is defined as:
\begin{align}
\mathcal{W}_\gamma(r,l) := \min_{\pi \in \Pi(r, l)} \{\left \langle \pi,C \right \rangle + \gamma H(\pi)\}
\end{align}
where $\gamma \geq 0$ is the regularization parameter. 

For a given set of probability measures $r_1, r_2, ..., r_m$ and corresponding cost matrices $C_1, C_2, ..., C_m \in \br_+^{n\times n}$, the weighted regularized Wasserstein barycenter problem is therefore: 
\begin{align} \label{regWB}
 \min_{q \in \Delta^n}\Sigma_{k=1}^m w_k \mathcal{W}_\gamma(r_k, q)
\end{align}
where $w_k\geq 0, k=1, ...m, \Sigma_{k=1}^m w_k = 1$ are the weights over the given measures.

\begin{figure*}[!ht]
\begin{minipage}[b]{.5\textwidth}
\includegraphics[width=65mm,height=48mm]{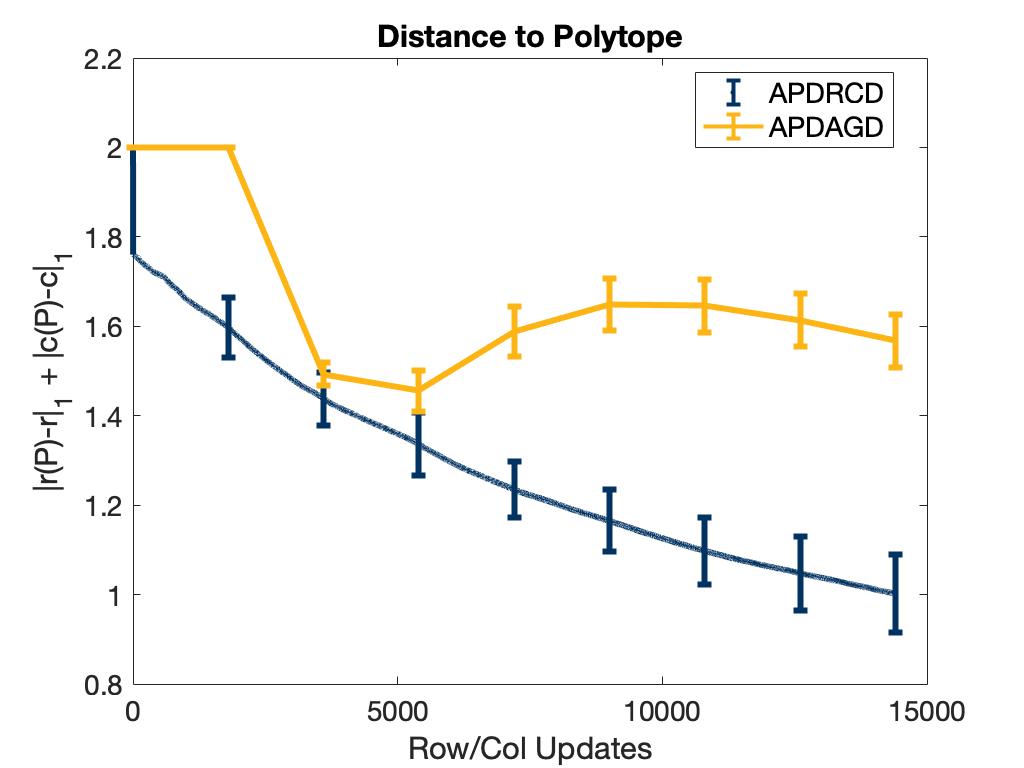}
\end{minipage}
\quad
\begin{minipage}[b]{.5\textwidth}
\includegraphics[width=65mm,height=48mm]{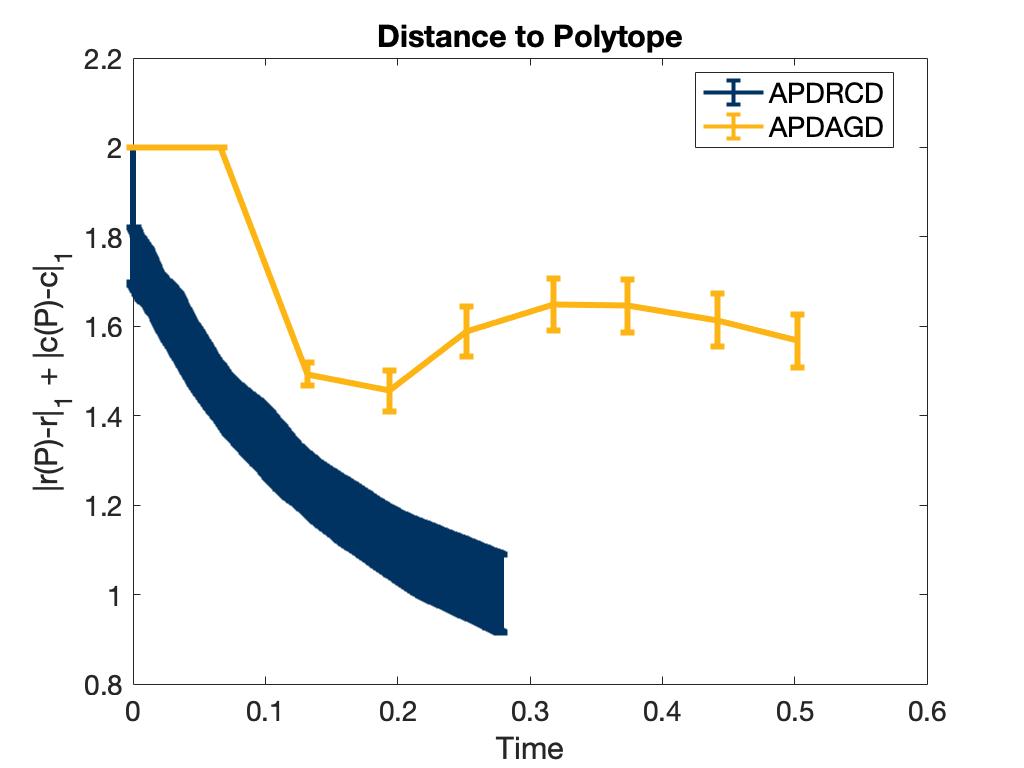}
\end{minipage}
\\
\begin{minipage}[b]{.5\textwidth}
\includegraphics[width=65mm,height=48mm]{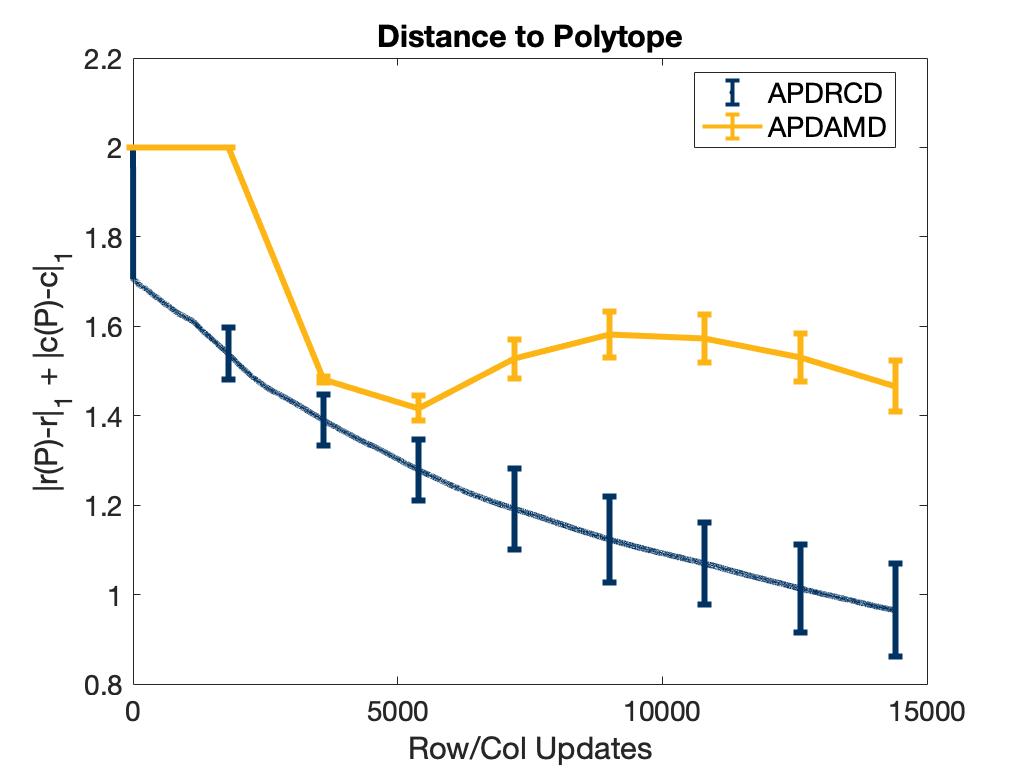}
\end{minipage}
\quad
\begin{minipage}[b]{.5\textwidth}
\includegraphics[width=65mm,height=48mm]{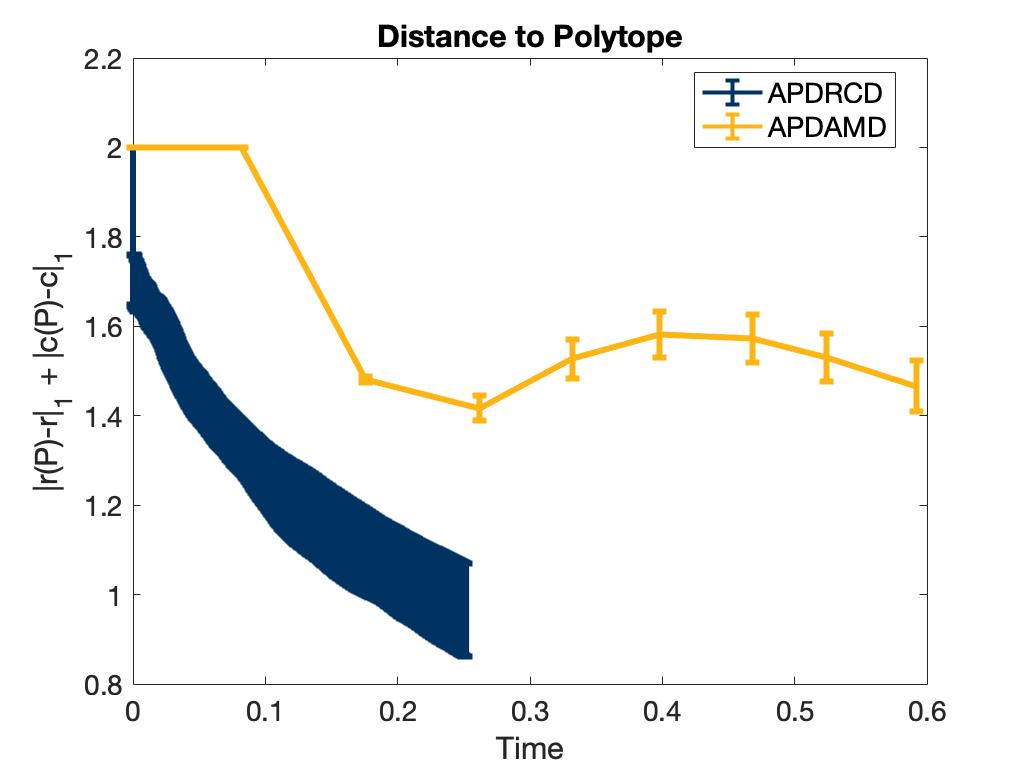}
\end{minipage}
\caption{\small Performance of APDRCD, APDAGD and APDAMD algorithms on $30*30$ synthetic images.}\label{fig:syn30}
\end{figure*}

\subsection{Network Constrains in the Barycenter Problem}
The given set of probability measures form a network, where each measure $r_i$ is held by an agent $i$ on the network. Such a network can be modeled as a fixed connected undirected graph $\mathcal{G} = (V, E)$ where $V$ is the set of $m$ nodes and $E$ is the set of edges. For convenience, we assume that the graph doesn't contain self-loops. The network structure add information constraints: each node can only exchange information with its direct neighbors. 

The communication constraints can be well-captured by the Laplacian matrix $\overline{\mathit{W} }\in \br^{m \times m}$of the graph $\mathcal{G}$ such that $[\overline{\mathit{W} }]_{ij} = -1$ if $(i,j) \in E$, $[\overline{\mathit{W} }]_{ij} = deg(i)$ if $i = j$, $[\overline{\mathit{W} }]_{ij} = 0$ otherwise. We further define the communitation matrix by $\mathit{W} := \overline{\mathit{W}}\otimes \mathit{I}_n$. 

We note some properties of the matrix $\mathcal{W}$. First, for connected and undirected $\mathcal{G}$, both $\overline{\mathit{W}}$ and $\mathcal{W}$ are positive semidefinite. Besides, $\sqrt{\mathcal{W}}\mathrm{q}=0$ if and only if $q_1 = ... = q_m$ where  $\mathrm{q} = [q_1, ...,q_m]^T \in \br^{mn}$.

\begin{algorithm}[!ht]
\caption{Generalized APDRCD for Computing Wasserstein Barycenters}
\label{Algorithm:WB-RCD}
\textbf{Input:}: Each agent $k \in V$ is assigned measure $r_l$, an upper bound $L$ for the Lipschitz constant of the gradient of the dual objective, and $N$. \\
\textbf{For all agents} $k \in V$, set $r_k = (1-\frac{\epsilon}{8})(r_k+\frac{\epsilon}{n(8-\epsilon)}\one)$,$ \gamma(k) = \frac{\epsilon}{4mw_k \ln n}, \eta_k^0 = \xi_k^0 = \lambda_k^0 = \hat{q}_k^0 = \bm{0} \in \br^n, A_0 = \alpha_0 = 0$ \\
\textbf{For each agents}  $k \in V$:\\
\textbf{for} $t = 0, ..., N-1$ \textbf{do} \\
Compute $\alpha_{t+1}$ as the largest root of \ \  $A_{t+1} := A_t + \alpha_{t+1} = 2L\alpha_{t+1}^2$  \\
Update $\lambda_k^{t+1} =\frac{\alpha_{t+1}\xi_k^t + A_k\eta_k^t}{A_{t+1}} $. \\
Calculate $\nabla \mathcal{W}_{\gamma(k),r_k}^\ast(\lambda_k^{t+1})$:  $[\nabla \mathcal{W}_{\gamma(k),r_k}^\ast(\lambda_k^{t+1})]_i = \sum\limits_{j = 1}^n[p_k]_j \frac{\exp(([\lambda]_i-[C_k]_{ij})/\gamma(k))}{\sum\limits_{s=1}^n \exp(([\lambda_s - [C_k]_{sj})/\gamma(k))}$  \\
Share $\nabla \mathcal{W}_{\gamma(k),r_k}^\ast(\lambda_k^{t+1})$ with $\{j|(i,j)\in E\}$ \\
\textbf{Randomly choose coordinate $s$} from $\{1, 2, ..., n\}$: \\
\textbf{Update} $[\xi_k^{t+1} ]_s= [\xi_k^t]_s - [\alpha_{t+1}\sum\limits_{j = 1}^m \mathcal{W}_{kj}\nabla \mathcal{W}_{\gamma(k),r_k}^\ast(\lambda_k^{t+1})]_s$ \\
\textbf{Update} $[\eta_k^{t+1}]_s = (\alpha_{t+1} [\xi_k^{t+1}]_s + A_k[\eta_k^{t+1}]_s)/A_{t+1}$ \\
\textbf{Update} $[q_k^{t+1}]_s = \frac{1}{A_{t+1}}\sum\limits_{k=0}^{t+1}\alpha_i[q_i(\lambda_k^{t+1})]_s = (\alpha_{t+1}[q_i([\lambda_{t+1}]_k)]_s+A_t[ q_k^t]_s)/A_{t+1}$ where $q_k(\cdot)$ is defined as $\nabla \mathcal{W}_{\gamma(k), r_k}^\ast(\cdot)$ \\
\textbf{end for} \\
\textbf{Output: } $\mathrm{q}^N = [q_1^T, ..., q_m^T]^T$
\end{algorithm}

\begin{figure*}[!ht]
\begin{minipage}[b]{.5\textwidth}
\includegraphics[width=65mm,height=48mm]{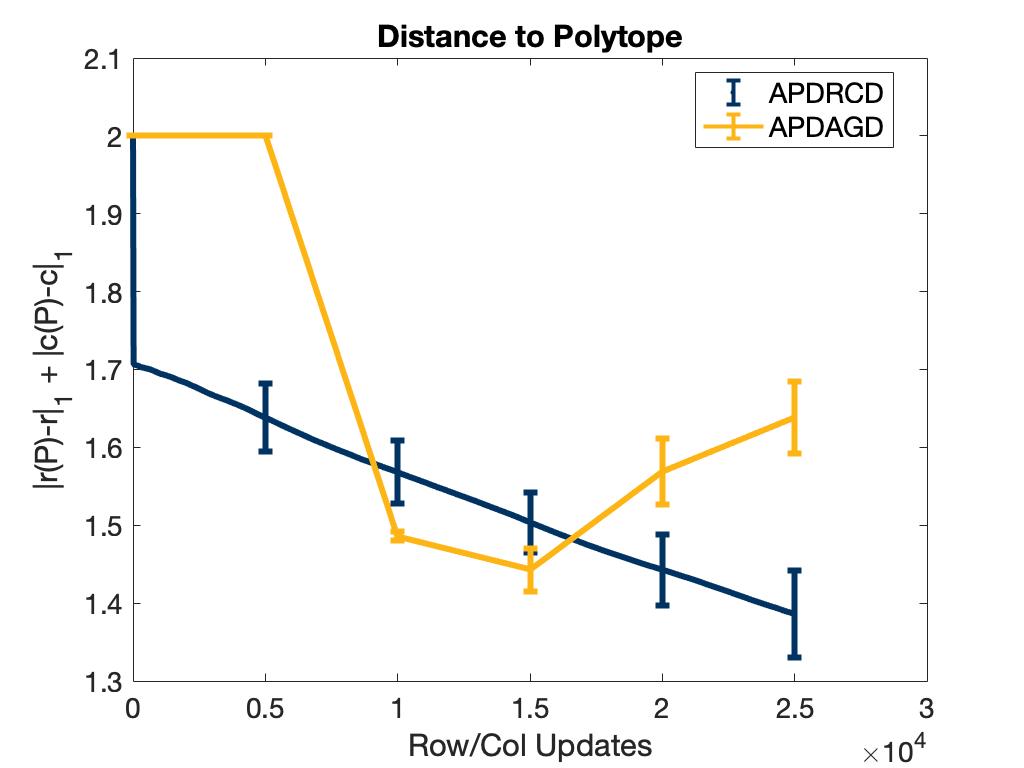}
\end{minipage}
\quad
\begin{minipage}[b]{.5\textwidth}
\includegraphics[width=65mm,height=48mm]{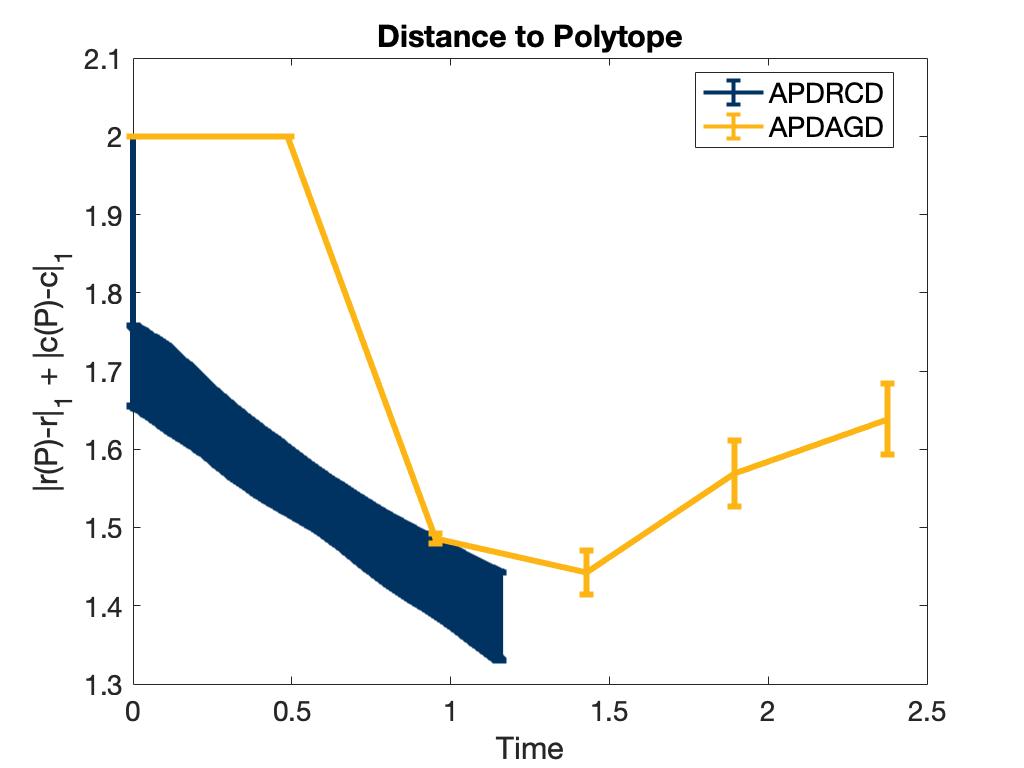}
\end{minipage}
\\
\begin{minipage}[b]{.5\textwidth}
\includegraphics[width=65mm,height=48mm]{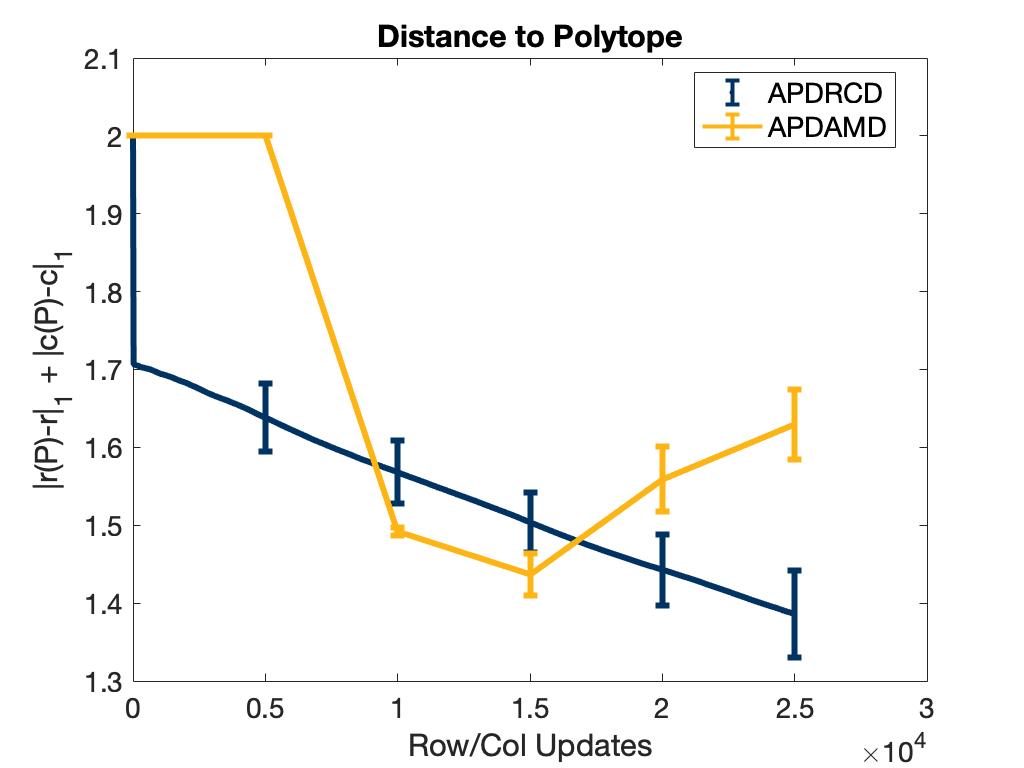}
\end{minipage}
\quad
\begin{minipage}[b]{.5\textwidth}
\includegraphics[width=65mm,height=48mm]{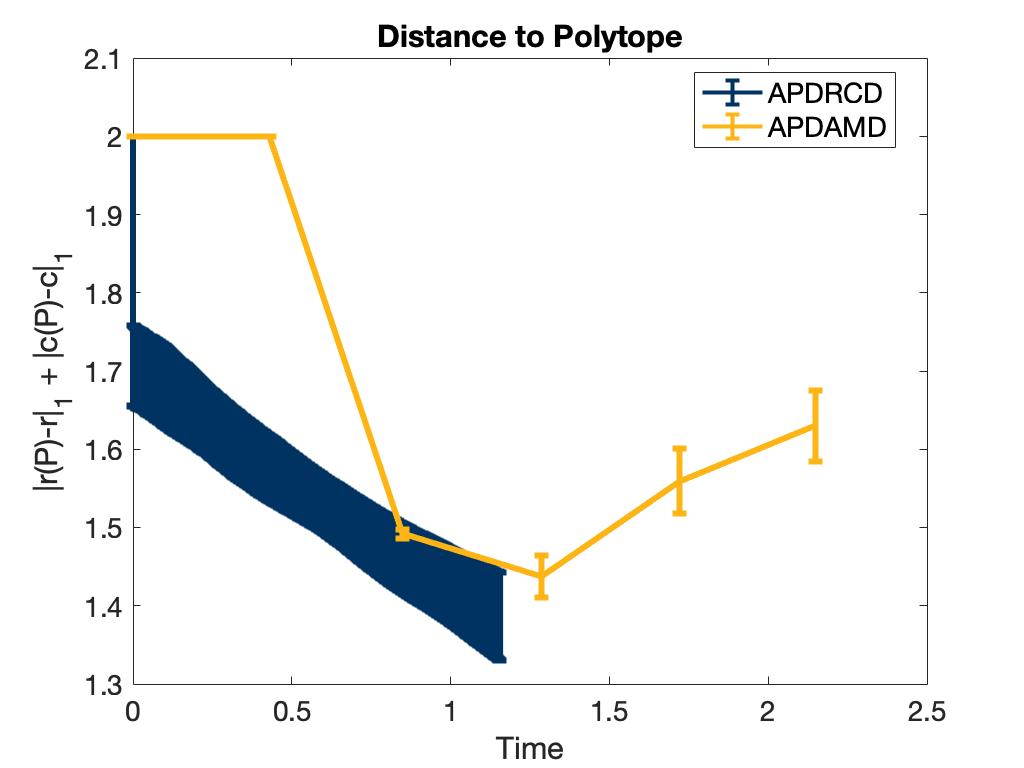}
\end{minipage}
\caption{\small Performance of APDRCD, APDAGD and APDAMD algorithms on $50*50$ synthetic images.}\label{fig:syn50}
\end{figure*}

\subsection{Dual Formulation of the Wasserstein Barycenter Problem}

To construct the dual problem, we first rewrite problem~\ref{regWB} as:
\begin{align}
\min_{\substack{q_1, ...,q_m \in \Delta^n \\ q_1=...=q_m}} \mathit{W}_\gamma (\mathrm{r}, \mathrm{q}) := \Sigma_{k=1}^mw_k \mathcal{W}_{\gamma(k)}(r_k, q_k)
\end{align}
where $\mathrm{r} = [r_1, ...,r_m]^T, \mathrm{q} = [q_1, ...,q_m]^T $.  

By using the property that $\sqrt{\mathcal{W}}\mathrm{q}=0$ if and only if $q_1 = ... = q_m$ where  $\mathrm{q} = [q_1, ...,q_m]^T \in \br^{mn}$, the above optimization problem can be further rewritten as:
\begin{align} \label{WBprimal}
\max_{\substack{q_1, ...,q_m \in \Delta^n \\ \sqrt{\mathcal{W}}\mathrm{q}=0}} - \Sigma_{k=1}^mw_k \mathcal{W}_{\gamma(k)}(r_k, q_k)
\end{align}

Given the above optimization problem with linear constrains, we introduce a vector of dual variables $\mathrm{\lambda}=[\lambda_1^T, ..., \lambda_m^T]^T \in \br^{mn}$ for the constraints $\sqrt{\mathcal{W}}\mathrm{q}=0$. The Lagrangian dual problem for (\ref{WBprimal}) is:
\begin{align} \label{WBdual}
 & \min_{\mathrm{\lambda} \in \br^{mn}}\max_{\mathrm{q}\in \br^{mn}} \nonumber \\
  & \biggr \{\Sigma_{k=1}^m \left \langle \lambda_k,[\sqrt{\mathit{W}}\mathrm{q}]_k \right \rangle  
 - \Sigma_{k=1}^mw_k \mathcal{W}_{\gamma(k)}(r_k, q_k) \biggr \} \nonumber \\
 &= \min_{\mathrm{\lambda}\in \br^{mn} }\Sigma_{k=1}^m w_k\mathcal{W}_{\gamma(k), r_k}^\ast([\sqrt{\mathcal{W}}\mathrm{\lambda}]_k/w_k)
\end{align}
where $\mathcal{W}_{\gamma(k),r_k}^\ast(\cdot)$ is the Fenchel-Legendre transofrm of $\mathcal{W}_{\gamma(k)}(r_k, \cdot)$. The function $\mathcal{W}_{\gamma(k),r_k}^\ast(\cdot)$ enjoys the nice property that it is a smooth function with Lipschitz-continous gradient as shown in ~\citep{dvurechenskii2018decentralize}.

\begin{algorithm}[!ht]
\caption{Generalized APDGCD for Computing Wasserstein Barycenters}
\label{Algorithm:WB-GCD}
\textbf{Input:}: Each agent $k \in V$ is assigned measure $r_l$, an upper bound $L$ for the Lipschitz constant of the gradient of the dual objective, and $N$. \\
\textbf{For all agents} $k \in V$, set $r_k = (1-\frac{\epsilon}{8})(r_k+\frac{\epsilon}{n(8-\epsilon)}\one)$,$ \gamma(k) = \frac{\epsilon}{4mw_k \ln n}, \eta_k^0 = \xi_k^0 = \lambda_k^0 = \hat{q}_k^0 = \bm{0} \in \br^n, A_0 = \alpha_0 = 0$ \\
\textbf{For each agents}  $k \in V$:\\
\textbf{for} $t = 0, ..., N-1$ \textbf{do} \\
Compute $\alpha_{t+1}$ as the largest root of \ \  $A_{t+1} := A_t + \alpha_{t+1} = 2L\alpha_{t+1}^2$  \\
Update $\lambda_k^{t+1} =\frac{\alpha_{t+1}\xi_k^t + A_k\eta_k^t}{A_{t+1}} $. \\
Calculate $\nabla \mathcal{W}_{\gamma(k),r_k}^\ast(\lambda_k^{t+1})$:  $[\nabla \mathcal{W}_{\gamma(k),r_k}^\ast(\lambda_k^{t+1})]_i = \sum\limits_{j = 1}^n[p_k]_j \frac{\exp(([\lambda]_i-[C_k]_{ij})/\gamma(k))}{\sum\limits_{s=1}^n \exp(([\lambda_s - [C_k]_{sj})/\gamma(k))}$  \\
Share $\nabla \mathcal{W}_{\gamma(k),r_k}^\ast(\lambda_k^{t+1})$ with $\{j|(i,j)\in E\}$ \\
\textbf{Select coordinate $s$} from $\{1, 2, ..., n\}$ where $s = \argmax_s |[\nabla \mathcal{W}_{\gamma(k), r_k}^\ast(\cdot)]_s| $ \\
\textbf{Update} $[\xi_k^{t+1} ]_s= [\xi_k^t]_s - [\alpha_{t+1}\sum\limits_{j = 1}^m \mathcal{W}_{kj}\nabla \mathcal{W}_{\gamma(k),r_k}^\ast(\lambda_k^{t+1})]_s$ \\
\textbf{Update} $[\eta_k^{t+1}]_s = (\alpha_{t+1} [\xi_k^{t+1}]_s + A_k[\eta_k^{t+1}]_s)/A_{t+1}$ \\
\textbf{Update} $[q_k^{t+1}]_s = \frac{1}{A_{t+1}}\sum\limits_{k=0}^{t+1}\alpha_i[q_i(\lambda_k^{t+1})]_s = (\alpha_{t+1}[q_i([\lambda_{t+1}]_k)]_s+A_t[ q_k^t]_s)/A_{t+1}$ where $q_k(\cdot)$ is defined as $\nabla \mathcal{W}_{\gamma(k), r_k}^\ast(\cdot)$ \\
\textbf{end for} \\
\textbf{Output: } $\mathrm{q}^N = [q_1^T, ..., q_m^T]^T$
\end{algorithm}
\subsection{Approximate Wasserstein Barycenters with Accelerated Coordinate Descent}
We apply the accelerated primal-dual randomized coordinate descent  (APDRCD) algorithm and accelerated primal-dual greedy coordinate descent  (APDGCD) algorithm to solve the pair of primal and dual problems for Wasserstein barycenter. The algorithms are presented in Algorithm \ref{Algorithm:WB-RCD} and Algorithm \ref{Algorithm:WB-GCD}.

\begin{figure*}[!ht]
\begin{minipage}[b]{.5\textwidth}
\includegraphics[width=65mm,height=48mm]{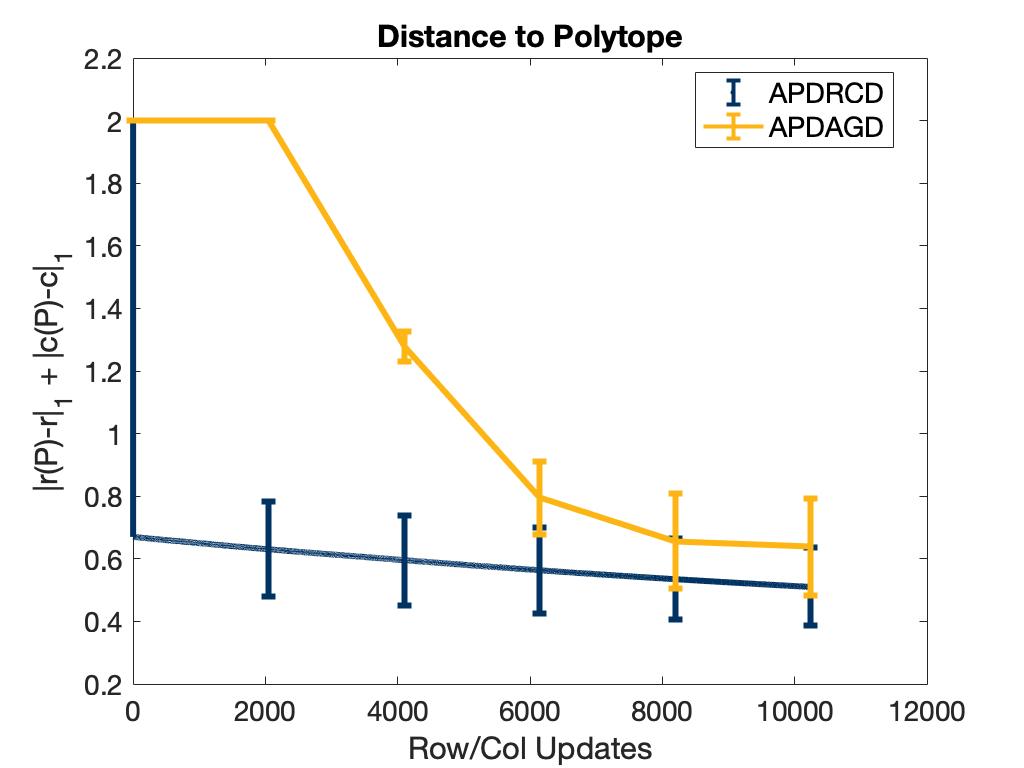}
\end{minipage}
\quad
\begin{minipage}[b]{.5\textwidth}
\includegraphics[width=65mm,height=48mm]{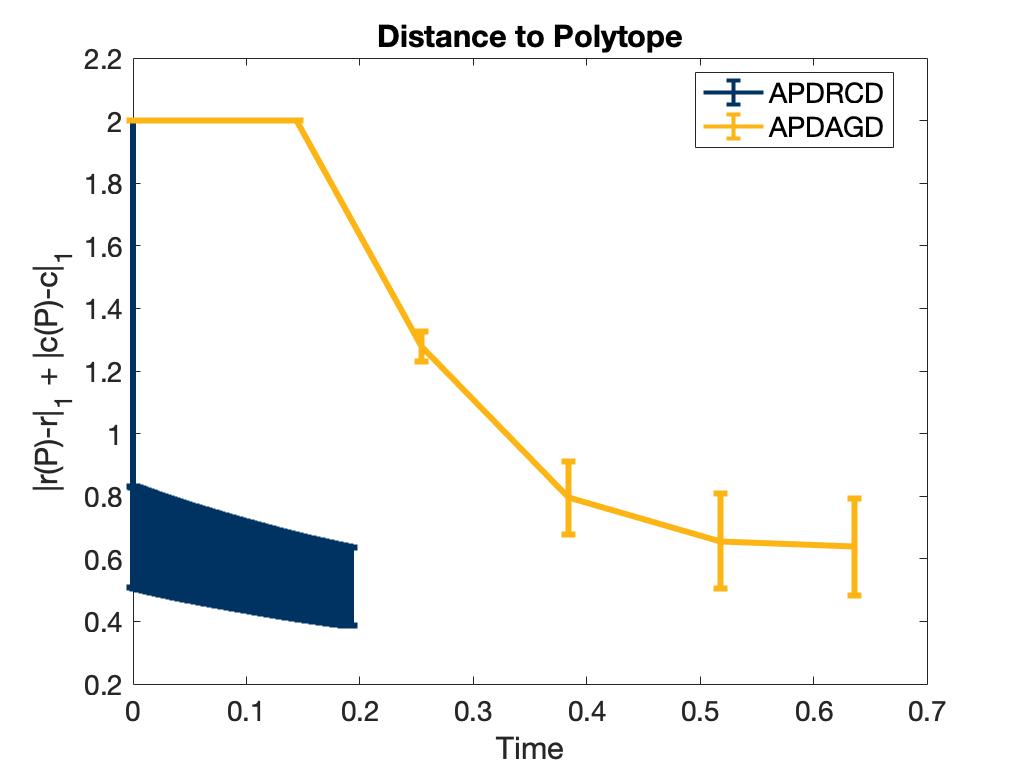}
\end{minipage}
\\
\begin{minipage}[b]{.5\textwidth}
\includegraphics[width=65mm,height=48mm]{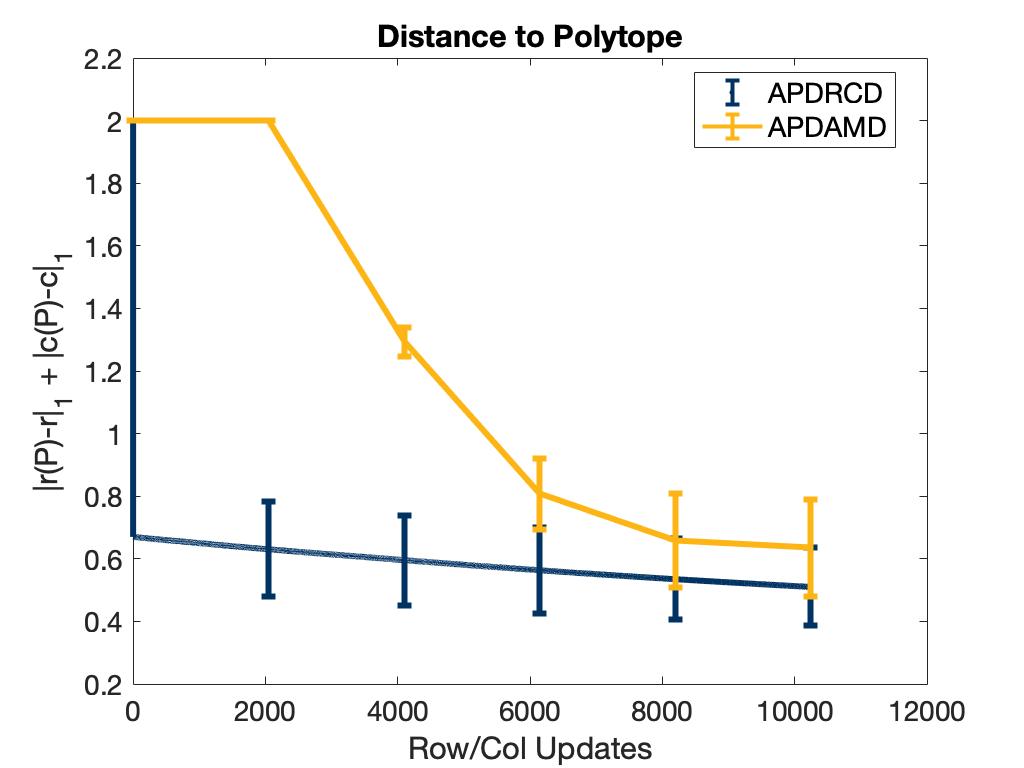}
\end{minipage}
\quad
\begin{minipage}[b]{.5\textwidth}
\includegraphics[width=65mm,height=48mm]{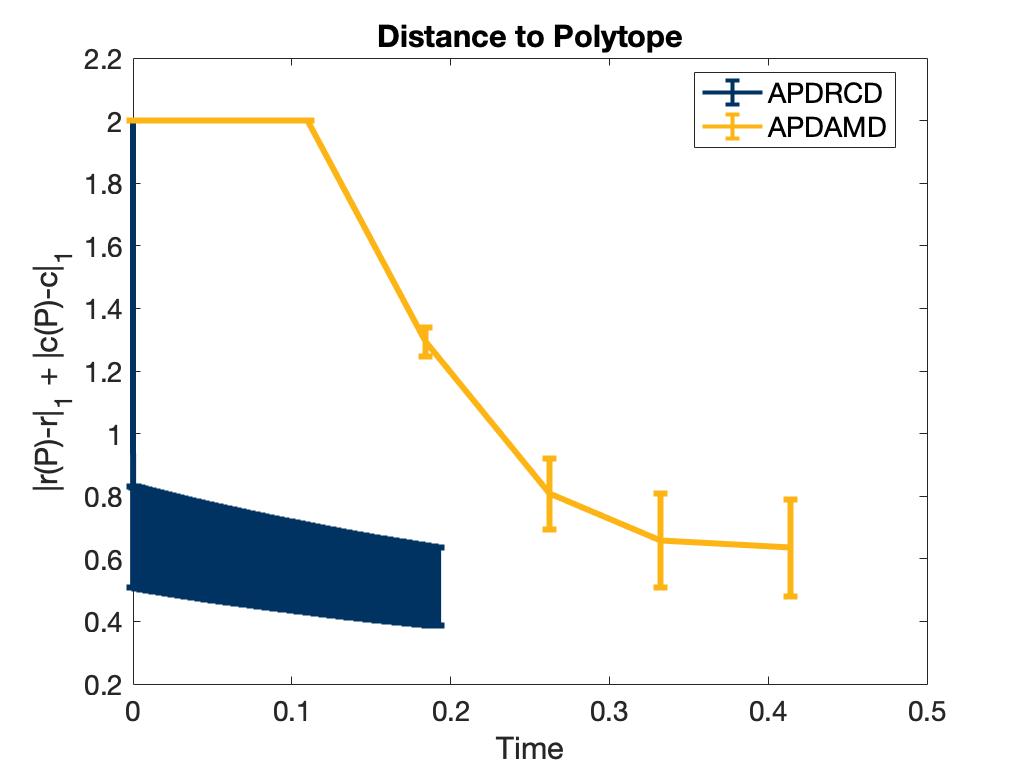}
\end{minipage}
\caption{\small Performance of APDRCD, APDAGD and APDAMD algorithms on 10000 CIFAR10 test images.}\label{fig:cifar}
\end{figure*}

\section{Further Experimental Results}
\label{subsec:further_exp}
In this appendix, we provide further comparative experiments between APDGCD algorithm versus APDRCD, APDAGD and APDAMD algorithms. We also provide further results on the performance of the APDRCD algorithm using larger synthetic datasets and CIFAR10. 

Experiments in Section~\ref{Sec:experiments} and Appendix~\ref{subsec:further_exp}) show that APDRCD enjoys consistent favorable practical performance than APDAGD, APDAMD on larger synthetic datasets and CIFAR10. Besides, APDGCD enjoys favorable practical performance than APDAGD, APDAMD, and APDRCD algorithms on both synthetic and real datasets. This demonstrates the benefit of choosing the best coordinate to descent to optimize the dual objective function of entropic regularized OT problems in the APDGCD algorithm comparing to choosing the random descent coordinate in the APDRCD algorithm. 

\subsection{APDGCD algorithm with synthetic images}
\label{subsec:APDGCD_synthetic}
The generation of synthetic images as well as the evaluation metrics are similar to those in Section~\ref{subsec:APDRCD_synthetic}. We respectively present in Figure~\ref{fig:gcd-agd-syn}, Figure~\ref{fig:gcd-amd-syn} and Figure~\ref{fig:gcd-rcd-syn} the comparisons between APDGCD algorithm versus APDAGD, APDAMD and APDRCD algorithms.  

According to Figure~\ref{fig:gcd-agd-syn}, Figure~\ref{fig:gcd-amd-syn} and Figure~\ref{fig:gcd-rcd-syn}, the
APDGCD algorithm enjoys better performance than the APDAGD, APDAMD and also the APDRCD algorithms in terms
of the iteration numbers in terms of both the evaluation metrics. Besides, at the same number of iteration number, the APDGCD
algorithm achieves even faster decrements than other three algorithms with regard to both
the distance to polytope and the value of OT metrics during the computing process. This is beneficial in practice
for easier tuning and smaller error when the update number is limited.

\subsection{APDGCD algorithm with MNIST images}
\label{subsec:APDGCD_mnist}
We present comparisons between APDGCD algorithm versus APDAGD, APDAMD, and APDRCD algorithms in Figure~\ref{fig:gcd-agd-amd-rcd-mnist} with MNIST images. 

According to Figure~\ref{fig:gcd-agd-amd-rcd-mnist}, the
APDGCD algorithm enjoys better performance than the APDAGD, APDAMD and also the APDRCD algorithms in terms
of the iteration numbers in terms of both the evaluation metrics. Furthermore, the convergence of the APDGCD
algorithm is faster than other three algorithms with regard to both
the distance to polytope and the value of OT metrics during the computing process when the number of iterations are small. This is beneficial in practice
for easier tuning and smaller error when the total update number is limited.

\subsection{APDRCD algorithm with larger synthetic image datasets and CIFAR10}
In this section, we provide further experiments on the APDRCD algorithm on larger syntheric image datasets and CIFAR10 dataset. Results are included in Figure~\ref{fig:syn30}, Figure~\ref{fig:syn50} and Figure~\ref{fig:cifar}. First, we provided results on the comparisons of APDRCD with APDAGD, APDAMD algorithms on larger synthetic datasets, with $n = 30*30$ and $n = 50*50$. We also provided the results on the CIFAR10 dataset. For each comparison, we provide the plots of the error of the dual variable versus the number of updates; and the error of the dual variable versus CPUtime (CPU: 3.1 GHz Intel Core i7) per iteration. The supplementary experiments show that APDRCD is more stable and achieve faster convergence in both number of row/col updates of the dual variables, and CPU time/iteration. The experimental setup are the same as the previous experiments except the change of dataset, hence omitted here. 

\end{document}